\documentclass{mscs}

\title[A Game-Semantic Model of Computation, Revisited]
  {A Game-Semantic Model of Computation, Revisited: an Automata-Theoretic View}
\author[N. Yamada]
  {N\ls O\ls R\ls I\ls H\ls I\ls R\ls O\ns
   Y\ls A\ls M\ls A\ls D\ls A$^\dagger$\ns \\
   $^\dagger$Department of Computer Science, University of Oxford \\ Wolfson Building, Parks Rd, OX1 3QD
   \if0 Email: \texttt{norihiro.yamada@cs.ox.ac.uk} \\ Email: \texttt{samson.abramsky@cs.ox.ac.uk} \fi }
\date{March 2019}

\pubyear{2001}
\pagerange{\pageref{firstpage}--\pageref{lastpage}}
\volume{00}
\doi{S0960129501003425}

\usepackage{amsmath,amsfonts,amssymb,stmaryrd,mathrsfs,mathdots,amsbsy,float,cmll,xcolor}
\usepackage[small]{diagrams}
\usepackage{hyperref}
\usepackage{bm}
\usepackage{tikz}
\usetikzlibrary{tikzmark}
\usepackage[small]{diagrams}
\usepackage{float}
\usetikzlibrary{positioning,calc}

\newtheorem{definition}{Definition}[section]
\newtheorem{example}{Example}[section]
\newtheorem{theorem}{Theorem}[section]

\newtheorem{proposition}{Proposition}[section]
\newtheorem{corollary}{Corollary}[section]

\newarrow{Corresponds}<--->

\begin{document}

\label{firstpage}
\maketitle

\begin{abstract}
In the previous work (Yamada, 2019), we have given a novel, game-semantic model of computation in an intrinsic (i.e., without recourse to any established model of computation), non-inductive, non-axiomatic fashion, which is similar to Turing machines yet beyond computation on natural numbers, e.g., higher-order computation.  
As the main theorem of the work, it has been shown that the game-semantic framework may accommodate all the computations of the  programming language PCF. 
The present paper revisits this result from an automata-theoretic perspective: It shows that \emph{deterministic, non-erasing pushdown automata} whose input tape is equipped with directed edges between cells can implement all the game-semantic PCF-computations, where the edges rather \emph{restrict} the cells of the tape which the automata may move to. 
The non-trivial point of the result is that non-deterministic, erasing pushdown automata are already strictly weaker than Turing machines, let alone than PCF.
In this manner, the present work introduces the game-semantic approach to automata theory, demonstrating high-potential of the resulting framework.
\end{abstract}

\ifprodtf \newpage \else \vspace*{-1\baselineskip}\fi

\section{Introduction}
The previous work \cite{yamada2019game} has given a novel, \emph{game-semantic} \cite{abramsky1999game,abramsky1997semantics,hyland1997game} model of computation in an intrinsic (i.e., without recourse to another model of computation), non-inductive, non-axiomatic fashion, which is similar to classic \emph{Turing machines (TMs)} \cite{turing1937computable} yet beyond computation on natural numbers (which let us call \emph{classical computation}), e.g., higher-order computation \cite{longley2015higher}.
As the main theorem, the work has shown that the game-semantic framework may accommodate all the computations of the higher-order programming language \emph{PCF} \cite{scott1993type,plotkin1977lcf}, and thus it is \emph{Turing complete} in particular \cite{gunter1992semantics,longley2015higher}.
In this manner, the work has established a mathematical foundation of computation beyond classical approaches such as TMs.

In hindsight, symbol manipulations in the game-semantic PCF-computations are \emph{very} simple from the automata-theoretic perspective. 
Hence, we are led to:
\begin{conjecture*}
There is some class of automata that are strictly weaker than TMs yet powerful enough to implement all the game-semantic PCF-computations.
\end{conjecture*}

The present work is dedicated to showing that it is not only a conjecture but a mathematical \emph{fact}.
More specifically, we shall establish:
\begin{theorem*}[Main theorem, informally]
Deterministic, non-erasing pushdown automata whose input tape is equipped with directed edges between cells of the tape, which restrict in a certain manner the cells that the automata may move to, can implement all the game-semantic PCF-computations. 
\end{theorem*}

Note that the theorem is seemingly contradictory because PCF is Turing complete, and non-deterministic, erasing pushdown automata are already not Turing complete \cite{hopcroft1979introduction,sipser2006introduction,kozen2012automata}.
Nevertheless, it \emph{does} hold, for which the seeming contradiction disappears as explained shortly.  
For now, let us just remark that the \emph{interactive} nature of game semantics is one of the two main contributors of the power of the pushdown automata (n.b., the other is the restriction by the edges). 

Prohibiting any interaction with other computational agents, however, the proof of the main theorem implies:
\begin{corollary*}[Main corollary, informally]
Deterministic, non-erasing stack automata whose stack is equipped with directed edges between stack cells, which restrict in a certain manner the stack cells that the automata may access, are Turing complete without any interaction with another computational agent. 
\end{corollary*}
Non-deterministic, non-erasing stack automata are strictly weaker than TMs \cite{hopcroft1967nonerasing} (yet more powerful than non-deterministic, non-erasing pushdown automata), and hence, the corollary is non-trivial. 
The point of the corollary is that the \emph{restriction} on the stack cells that the stack automata may access (without interaction with another agent) actually brings the stack automata Turing completeness (i.e., here, the interactive nature of game semantics contributes only to the higher-order aspect of the game-semantic PCF-computations, not to Turing completeness).

\begin{remark*}
In the following, we are often casual about the distinction between an occurrence of an element in a sequence and the element itself.
This convention would not bring any serious confusion in practice, and it is in fact standard in the literature of game semantics. 
See, e.g., \cite{yamada2019game} for how to formalize the distinction if necessary.
\end{remark*}

\subsection{The idea in a nutshell}
Let us sketch how we shall prove the theorem and the corollary.
First, in the game-semantic model of computation \cite{yamada2019game}, a computational agent or \emph{Player (P)} and an oracle or \emph{Opponent (O)} alternately perform \emph{moves} allowed by the rule of the underlying \emph{game}, where O always acts first.
Thus, a \emph{play} of the game proceeds as: 
\begin{equation}
\label{Plays}
\boldsymbol{\epsilon} \mapsto o_1 \mapsto o_1 p_1 \mapsto o_1 p_1 o_2 \mapsto o_1 p_1 o_2 p_2 \mapsto \dots
\end{equation}
where $\boldsymbol{\epsilon}$ is the \emph{empty sequence}, and $o_i$ (resp. $p_i$) with $i \in \mathbb{N}$ is O's (O-) move (resp. P's (P-) move).
Note that a game is to specify its possible plays, and it interprets a \emph{type} of a programming language.
Each element of the play (\ref{Plays}), i.e., an alternating (with respect to the OP-parity) finite sequence of moves that is `valid' in the game, is called a \emph{(valid) position} of the game.
Strictly speaking, there is the distinction between \emph{initial} and \emph{non-initial} moves of each game, where initial moves are distinguished moves that may initiate a play of the game, and every occurrence $m$ of a non-initial O- (resp. P-) move in a position of the game is associated with a previous occurrence $m'$ of a P- (resp. O-) move in the position, called its \emph{justifier}; we also say that there is a (necessarily unique) \emph{pointer} from $m$ to $m'$ in the position.
That is, positions of a game are certain finite sequences equipped with a structure of pointers.
For instance, there are two possible patters of justifiers for the position $o_1 p_1 o_2 p_2$, where only $o_1$ is initial, in the play (\ref{Plays}): In the first pattern, $o_1$ is the justifier of $p_1$, $p_1$ is the justifier of $o_2$, and $o_2$ is the justifier of $p_2$; the second pattern differs from the first only in the point that $o_1$ is the justifier of $p_2$.

On the other hand, a \emph{strategy} $\sigma$ on a game $G$, for which we write $\sigma : G$, is what tells P `how to play on $G$', and it interprets a \emph{term} or a \emph{program} of a programming language.
More precisely, a strategy is a partial function that maps an odd-length position $\boldsymbol{s} o$ of the underlying game to the next P-move $p$, where it also specifies the justifier of $p$, such that the concatenation $\boldsymbol{s} o p$ (equipped with the specified justifier of $p$) is a position of the game.
For example, maximal positions of the game $N$ of natural numbers \cite{abramsky1999game,hyland1997game} are $q n$, where $q$ is an initial O-move representing O's question `What is your number?', $n \in \mathbb{N}$ is a non-initial P-move representing P's answer `My number is $n$!', and $q$ is the justifier of $n$.
Then, there is a strategy $\underline{n} : N$ for the natural number $n \in \mathbb{N}$ that maps $q \mapsto n$ (and points $q$ as the justifier of $n$).

Hence, `effective computability' of strategies in game semantics must be defined on how to calculate the next P-move for a given odd-length position of the underlying game. 
For this point, following the previous work \cite{yamada2019game}, let us represent moves of each game for PCF-computations by a fixed alphabet, particularly in the following form:
\begin{equation*}
[m]_{e_1 e_2 \dots e_k} \stackrel{\mathrm{df. }}{=} (m, e_1 e_2 \dots e_k)
\end{equation*}
where $m$ is the `essence' of the move, and the finite sequence $e_1 e_2 \dots e_k$ is the `tag' on the move for disjoint union of sets of moves (for constructions on games).
Then, each step of PCF-computations of the previous work \cite{yamada2019game} is executed by:
\begin{enumerate}

\item Locating, with the help of pointers, a bounded number of `relevant' moves in the \emph{P-view} (see Appendix~\ref{DefViews}) of a given odd-length position of the underlying game;

\item Calculating the symbolic representation of the next P-move (and its justifier) from those of the `relevant' moves.

\end{enumerate}

Then, the main idea of the present work is to implement the game-semantic PCF-computations by \emph{deterministic, non-erasing pushdown automata} such that the input tape is equipped with directed edges between cells of the tape, where the cells that the automata may move to are rather \emph{restricted} in a certain manner (specifically to the ones containing symbols for moves in the P-view of the current position), called \emph{j-pushdown automata}. 
We assume that each position during a play is recorded by someone or \emph{Judge (J)} of the game on the input tape, and j-pushdown automata compute the next P-move into the stack. 

More concretely, each position $\boldsymbol{s}$ of a game $G$ is written on the input tape of a j-pushdown automaton in the following form:
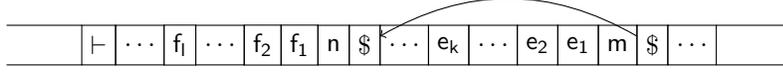
\begin{center}
\begin{tikzpicture}[every node/.style={block},
        block/.style={minimum height=1.5em,outer sep=0pt,draw,rectangle,node distance=0pt}]
   \node (A) {$\mathsf{e_k}$};
   \node (B) [left=of A] {$\ldots$};
   \node (C) [left=of B] {$\mathsf{\$}$};
   \node (G) [left=of C] {$\mathsf{n}$};
   \node (H) [left=of G] {$\mathsf{f_1}$};
   \node (K) [left=of H] {$\mathsf{f_2}$}; 
   \node (I) [left=of K] {$\ldots$};
   \node (J) [left=of I] {$\mathsf{f_l}$};
   \node (O) [left=of J] {$\ldots$};
   \node (P) [left=of O] {$\vdash$};
   \node (D) [right=of A] {$\ldots$};
   \node (E) [right=of D] {$\mathsf{e_2}$};
   \node (L) [right=of E] {$\mathsf{e_1}$};
   \node (M) [right=of L] {$\mathsf{m}$};
   \node (N) [right=of M] {$\mathsf{\$}$}; 
   \node (O) [right=of N] {$\ldots$};
   \draw (P.north west) -- ++(-1cm,0) (P.south west) -- ++ (-1cm,0) 
                 (O.north east) -- ++(1cm,0) (O.south east) -- ++ (1cm,0);
   \draw [->] (N) [bend right] to (C);
\end{tikzpicture}
\end{center}
where $[m]_{e_1 e_2 \dots e_k}$ is any occurrence of a non-initial move in $\boldsymbol{s}$, and $[n]_{f_1 f_2 \dots f_l}$ is its justifier in $\boldsymbol{s}$, and the distinguished symbol $\vdash$ is to signify where $\boldsymbol{s}$ begins. 
That is, $\boldsymbol{s}$ is written on the tape from left to right, where each element $[m]_{e_1 e_2 \dots e_k}$ is represented by an expression $\mathsf{e_k \dots e_2 e_1 m}$ postfixed by the distinguished symbol $\mathsf{\$}$, and each pointer is represented by a directed edge between the cells containing the $\mathsf{\$}$'s associated with the pointer's domain and codomain occurrences of moves. 

In addition, we require that j-pushdown automata must \emph{jump} from the current cell $c$ containing $\mathsf{\$}$ to another $c'$ (necessarily containing $\mathsf{\$}$ as well) if the move written on the immediate left of $c$ (i.e., the move which $\mathsf{\$}$ in $c$ is attached to) is a non-initial O-move, and there is a (necessarily unique) edge from $c$ to $c'$ (so that they can move only to the cells that contain symbols representing moves in P-views; see Appendix \ref{DefViews}).  

Let us emphasize here that j-pushdown automata are a rather \emph{restricted} kind of deterministic, non-erasing pushdown automata, and more powerful non-deterministic, erasing pushdown automata are already strictly weaker than TMs.
In this sense, j-pushdown automata are strictly weaker than TMs; see Proposition~\ref{PropStrictWeakness}. 

Note in particular that j-pushdown automata can execute only the following:
\begin{itemize}

\item To move its reading head on a cell of the input tape to another cell on the left, following the restriction by edges as defined above;

\item To change the current state;

\item To push a symbol into the stack.

\end{itemize}
They compute the next P-move $[p]_{g_1 g_2 \dots g_r}$ and its justifier for a given odd-length position $\boldsymbol{s}$ by pushing symbols into the stack so that its content becomes: \\
\begin{center}
\begin{tikzpicture}[every node/.style={block},
        block/.style={minimum height=1.5em,outer sep=0pt,draw,rectangle,node distance=0pt}]
   \node (H) [left=of G] {$\mathsf{g_r}$};
   \node (K) [left=of H] {$\ldots$}; 
   \node (I) [left=of K] {$\mathsf{g_2}$};
   \node (J) [left=of I] {$\mathsf{g_1}$};
   \node (O) [left=of J] {$\mathsf{p}$};
   \node (P) [left=of O] {$\mathsf{J}$};
   \node (Q) [left=of P] {$\ldots$};
   \node (R) [left=of Q] {$\vdash$};
   \node (S) [right=of H] {$\mathsf{\$}$};
   \draw 
                 (H.north east) -- ++(1cm,0) (H.south east) -- ++ (1cm,0);
\end{tikzpicture}
\end{center}
where the bottom of the stack is on the left (indicated by the fixed symbol $\vdash$), and $\mathsf{J}$ is $\mathsf{i}$, $\mathsf{ii}$ or $\mathsf{iii}$, indicating the justifier of the P-move (n.b., in the game-semantic PCF-computations \cite{yamada2019game}, there are only three patterns for the justifier of a P-move occurring in a position of a game).

Then, the main theorem is spelled out as follows. Let $\sigma : G$ be the strategy that interprets a term of PCF \cite{yamada2019game}.
Then, there is a j-pushdown automaton $\mathscr{A}_\sigma$ such that, for any computation $\boldsymbol{s} [o]_{f_1 f_2 \dots f_l} \mapsto [p]_{g_1 g_2 \dots g_r}$ of $\sigma$, if the odd-length position $\boldsymbol{s} [o]_{f_1 f_2 \dots f_l}$ is written on the  tape together with the pointers represented by directed edges as specified above, then the computation of $\mathscr{A}_\sigma$ terminates with the stack content representing the P-move $[p]_{g_1 g_2 \dots g_r}$ in the format specified above.

This may sound too good to be true and even contradictory to the non-equivalence of TMs and pushdown automata; however, the theorem \emph{does} hold. 
The trick is actually the edges on the input tape of j-pushdown automata. 
Recall that we have required that j-pushdown automata can move only to the cells that contain symbols representing moves in the P-view of the current position of the underlying game.
At first glance, this condition restricts the computational power of the automata; however, it implicitly serves as a kind of a `route recorder' and \emph{saves} their computation to locate the cells to read off.\footnote{N.b., it does not mean that the cells to be read off are automatically computed because it is j-pushdown automata that compute justifiers during a play.}
In fact, if we had adopted the ordinary input tape (without edges), then we would need another (erasing) stack for locating the cells to read off (i.e., for following pointers encoded by symbols on the input tape); then, it is a well-known fact that deterministic, erasing pushdown automata \emph{with two stacks} are computationally equivalent to TMs \cite{hopcroft1979introduction}, and thus the seeming contradiction mentioned above has disappeared.\footnote{It is not a very surprising fact that TMs can implement all the game-semantic PCF-computations; see, e.g., the \emph{universality theorem} in \cite{yamada2019game}.}

Nevertheless, the point of the theorem is that we do not add any computational ability to deterministic, non-erasing pushdown automata; instead, we rather \emph{restrict} the cells to be read off. 
The miracle is then that the game-semantic framework gives such highly restricted j-pushdown automata the computational power at least as strong as PCF, which we call \emph{PCF-completeness}.

\subsection{Further investigation}
Seeing more closely what contributes PCF-completeness of j-pushdown automata, there are mainly two contributors:
\begin{enumerate}

\item \textsc{(Game-semantic compromise).} J-pushdown automata compute the next P-move \emph{into the stack}, not onto the input tape, assuming that J instead reads the move (and its justifier) in the stack and modifies the content of the input tape accordingly;

\item \textsc{(Edges on the input tape).} As already explained, edges on the input tape save certain computation by j-pushdown automata.

\end{enumerate}

The game-semantic compromise is somewhat unusual for automata theory \cite{hopcroft1979introduction,sipser2006introduction,kozen2012automata}, though natural from the game-semantic viewpoint, for an automaton is usually expected to execute every computational step by itself, i.e., in the stand-alone fashion. 
Hence, one may wonder what would happen if we prohibit the game-semantic compromise.
Clearly, j-pushdown automata are no longer Turing complete, let alone PCF-complete, without the compromise because they have no means to write on the input tape; even if positions of games are recorded in the stack, they can read off only the symbol in the top cell of the stack.

This situation then suggests us to employ deterministic, non-erasing \emph{stack automata} \cite{ginsburg1967stack,hopcroft1967nonerasing} such that the stack is equipped with directed edges similarly to the input tape of j-pushdown automata, and the stack automata can access only the stack cells corresponding to P-views, where positions occurring in a play of a game are recorded in the stack. 
Let us call such stack automata \emph{j-stack automata}.
Clearly, j-stack automata are more powerful than j-pushdown automata (which is why we employ the latter for the main theorem as it would be more surprising), but they are still strictly weaker than TMs as in the case of j-pushdown automata \cite{hopcroft1967nonerasing}.

Then, essentially in the same way as the proof of the theorem on j-pushdown automata, we may show that j-stack automata are PCF-complete, which is, given the theorem, not surprising at all.
However, if we focus on classical computation, then, unlike j-pushdown automata, j-stack automata compute completely \emph{in the stand-alone fashion}: Given an input in the stack, a j-stack automaton computes an output in the stack without any interaction with another agent such as O or J (where the input tape is not used at all).
That is, the main corollary sketched above has been spelled out as follows: J-stack automata per se are Turing complete.

Unlike the theorem, the corollary is completely automata-theoretic because it does not rely on any assumption specific to game semantics, where the game-semantic framework only leads to the proof. 
Our motivation for the corollary is to carve out the power of the edges on the stack (without interaction with another agent), which brings deterministic, non-erasing stack automata Turing completeness.
To summarize:

\begin{center}
\begin{tabular}{| c | c |} \hline
{\bfseries Turing completeness} & {\bfseries PCF-completeness} \\ \hline
J-pushdown automata (JPAs) with J & JPAs with J \& O \\ 
J-stack automata (JSAs) & JSAs with O \\ \hline
\end{tabular}
\end{center}

\subsection{Our contribution and related work}
We believe that the main theorem, i.e., PCF-completeness of j-pushdown automata, is somewhat surprising from the view of theory of computation and recursion theory because it in a sense overturns the well-established hierarchy of automata \cite{chomsky1956three,sipser2006introduction}. 
The main corollary, i.e., Turing completeness of j-stack automata, is also non-trivial for it shows that the game-semantic compromise is not necessary for their Turing completeness; rather, the \emph{restriction} on stack cells by edges is the key contributor.

From a methodological viewpoint, the present work indicates high potential of the game-semantic approach for automata theory; see Section~\ref{ConclusionAndFutureWork} for further directions.
For this point, we have written this paper essentially in a self-contained manner, recalling the game-semantic PCF-computations in Section~\ref{Review}, so that it would be accessible to logicians and mathematicians who have been unfamiliar with game semantics.  
Hence, another (though not main) contribution of the paper is to introduce the game-semantic approach to theory of computation and recursion theory \cite{yamada2019game} to wider audience, rephrasing it by the more familiar automata-theoretic setting. 

As related work, let us mention the work on a correspondence between collapsible pushdown automata and recursion schemes by Ong et al. \cite{hague2008collapsible}. 
Roughly, \emph{collapsible pushdown automata} are \emph{higher-order pushdown automata} \cite{knapik2002higher} such that each symbol in the stack is equipped with a \emph{link} to another stack occurring below, and there is an additional stack operation, called \emph{collapse}, that `collapses' a stack $\boldsymbol{s}$ to the prefix of $\boldsymbol{s}$ as indicated by the link from the $\textit{top}_1$-symbol of $\boldsymbol{s}$ (see \cite{hague2008collapsible} for the precise definition); \emph{recursion schemes} or \emph{simply-typed $\lambda Y$-calculi} are simply-typed $\lambda$-calculi equipped with fixed-point combinators $\mathsf{Y_A}$ for each type $\mathsf{A}$ \cite{amadio1998domains}. 
They have shown, as the main result, that collapsible pushdown automata and recursion schemes have the same expressive power as generators of node-labelled ranked trees. 
Therefore, collapsible pushdown automata can be seen as a computational device that generates the trees that represent terms of \emph{finitary PCF}, i.e., the fragment of PCF that has the boolean type as the sole ground type, and thus they are relevant to the present work.
However, our automata and collapsible pushdown automata are employed for rather different purposes: J-pushdown and j-stack automata are to compute the next P-move from a given P-view (in an interaction with O and/or J), while collapsible pushdown automata are to generate the \emph{entire} (possibly infinite in depth) tree of a term (without any interaction with O or J).
In other words, the former only computes a \emph{single P-move} for a given odd-length position of a game, while the latter enumerates \emph{all positions} of a game.
Hence, it is not very surprising that our automata do not need higher-order stacks or the collapse operation, but they implement (non-finitary) PCF. 

\subsection{Structure of the paper}
The rest of the present paper proceeds as follows. 
This introduction ends with fixing notation. 
Recalling the variant of games and strategies employed in the previous work \cite{yamada2019game} in Section~\ref{Review}, we define j-pushdown automata (resp. j-stack automata) and establish their PCF-completeness (resp. Turing completeness) in Section~\ref{JPointingAutomata}.
Finally, we draw a conclusion and propose future work in Section~\ref{ConclusionAndFutureWork}.

\begin{notation*}
We use the following notations throughout the present paper:
\begin{itemize}

\item We use bold letters $\boldsymbol{s}, \boldsymbol{t}, \boldsymbol{u}, \boldsymbol{v}$, etc. for sequences, in particular $\boldsymbol{\epsilon}$ for the \emph{empty sequence}, and letters $a, b, c, d$, etc. for elements of sequences;

\item We often abbreviate a finite sequence $\boldsymbol{s} = (x_1, x_2, \dots, x_{|\boldsymbol{s}|})$ as $x_1 x_2 \dots x_{|\boldsymbol{s}|}$, where $|\boldsymbol{s}|$ denotes the \emph{length} (i.e., the number of elements) of $\boldsymbol{s}$, and write $\boldsymbol{s}(i)$, where $i \in \{ 1, 2, \dots, |\boldsymbol{s}| \}$, as another notation for $x_i$;


\item A \emph{concatenation} of sequences is represented by the juxtaposition of them, but we write $a \boldsymbol{s}$, $\boldsymbol{t} b$, $\boldsymbol{u} c \boldsymbol{v}$ for $(a) \boldsymbol{s}$, $\boldsymbol{t} (b)$, $\boldsymbol{u} (c) \boldsymbol{v}$, etc., and also $\boldsymbol{s} . \boldsymbol{t}$ for $\boldsymbol{s t}$;

\item We define $\boldsymbol{s}^n \stackrel{\mathrm{df. }}{=} \underbrace{\boldsymbol{s} \boldsymbol{s} \cdots \boldsymbol{s}}_n$ for a sequence $\boldsymbol{s}$ and a natural number $n \in \mathbb{N}$;

\item We write $\mathsf{Even}(\boldsymbol{s})$ (resp. $\mathsf{Odd}(\boldsymbol{s})$) iff $\boldsymbol{s}$ is of even-length (resp. odd-length);

\item Let $S^\mathsf{P} \stackrel{\mathrm{df. }}{=} \{ \boldsymbol{s} \in S \mid \mathsf{P}(\boldsymbol{s}) \}$ for a set $S$ of sequences and $\mathsf{P} \in \{ \mathsf{Even}, \mathsf{Odd} \}$;

\item $\boldsymbol{s} \preceq \boldsymbol{t}$ means $\boldsymbol{s}$ is a \emph{prefix} of $\boldsymbol{t}$, i.e., $\boldsymbol{t} = \boldsymbol{s} . \boldsymbol{u}$ for some sequence $\boldsymbol{u}$, and given a set $S$ of sequences, we define $\mathsf{Pref}(S) \stackrel{\mathrm{df. }}{=} \{ \boldsymbol{s} \mid \exists \boldsymbol{t} \in S . \ \! \boldsymbol{s} \preceq \boldsymbol{t} \ \! \}$;

\item For a poset $P$ and a subset $S \subseteq P$, $\mathsf{Sup}(S)$ denotes the \emph{supremum} of $S$;


\item $X^* \stackrel{\mathrm{df. }}{=} \{ x_1 x_2 \dots x_n \mid n \in \mathbb{N}, \forall i \in \{ 1, 2, \dots, n \} . \ \! x_i \in X \ \! \}$ for each set $X$;

\item For a function $f : A \to B$ and a subset $S \subseteq A$, we define $f \upharpoonright S : S \to B$ to be the \emph{restriction} of $f$ to $S$, and $f^\ast : A^\ast \to B^\ast$ by $f^\ast(a_1 a_2 \dots a_n) \stackrel{\mathrm{df. }}{=} f(a_1) f(a_2) \dots f(a_n)$;


\item Given sets $X_1, X_2, \dots, X_n$, and an index $i \in \{ 1, 2, \dots, n \}$, we write $\pi_i$ (or $\pi_i^{(n)}$) for the \emph{$i^{\text{th}}$-projection function} $X_1 \times X_2 \times \dots \times X_n \to X_i$ that maps $(x_1, x_2, \dots, x_n) \mapsto x_i$.


\end{itemize}
\end{notation*}

\section{Review: games and strategies for PCF-computation}
\label{Review}
This section presents a brief, self-contained review of the game-semantic PCF-computations given in the previous work \cite{yamada2019game}, which only focuses on the contents relevant to the present work.\footnote{The present section actually constitutes about half of the present paper. As already mentioned, it is to introduce our game-semantic approach to wider audience who has been unfamiliar with game semantics, where note that the previous work \cite{yamada2019game} consists of 79 pages. Hence, we believe that a summary of its relevant part like the present section is meaningful.}
We therefore encourage the reader who is already familiar with the previous work to skip the present section. 
On the other hand, we recommend the reader who is unfamiliar with game semantics to consult first with the first 29 pages of the very gentle introduction to ordinary game semantics \cite{abramsky1999game}.

In addition, let us remark that our exposition is technically more involved than ordinary game semantics, in particular by `tags' for disjoint union. 
However, it is a `necessary evil', i.e., any mathematically serious account on theory of computation and recursion theory seems to need to some degree an involved formalism, especially to prove for the first time the computational power of the approach; see, e.g., \cite{hopcroft1979introduction,kleene1952introduction}.

We first recall the general definitions of games and strategies in Section~\ref{GamesAndStrategies}, and standard constructions on them in Section~\ref{Constructions}.
Finally, we recall the games and strategies for the game-semantic PCF-computations in Section~\ref{GamesAndStrategiesForPCF}.

\begin{remark*}
The variant of games and strategies employed in \cite{yamada2019game} are the \emph{dynamic} one introduced for the first time in \cite{yamada2016dynamic}, which we call in this paper \emph{\bfseries games} and \emph{\bfseries strategies}, respectively. 
For brevity, we simplify some of the original definitions, forgetting structures not necessary for the present work. 
\end{remark*}

\subsection{Games and strategies}
\label{GamesAndStrategies}
A \emph{game}, roughly, is a certain kind of a rooted forest whose branches represent possible `developments' or \emph{(valid) positions} of a `game in the usual sense' (such as chess, poker, etc.).
\emph{Moves} of a game are nodes of the game, where some moves are distinguished and called \emph{initial}; only initial moves can be the first element (or occurrence) of a position of the game. 
\emph{Plays} of a game are (finitely or infinitely) increasing sequences $(\boldsymbol{\epsilon}, m_1, m_1 m_2, \dots)$ of positions of the game. 
For our purpose, it suffices to focus on rather standard \emph{sequential} (as opposed to \emph{concurrent} \cite{abramsky1999concurrent}), \emph{unpolarized} (as opposed to \emph{polarized }\cite{laurent2002polarized}) games played by two participants, \emph{Player} (\emph{P}), who represents a computational agent, and \emph{Opponent} (\emph{O}), who represents an oracle or an environment, in each of which O always starts a play (i.e., unpolarized), and then they alternately and separately perform moves allowed by the rules of the game (i.e., sequential).
Strictly speaking, a position is not just a finite sequence of moves: Each occurrence $m$ of O's (O-) (resp. P's (P-)) non-initial move in a position $\boldsymbol{s}$ is assigned a previous occurrence $m'$ of P- (resp. O-) move in $\boldsymbol{s}$, representing that $m$ is performed specifically as a response to $m'$; we call $m'$ the \emph{justifier} of $m$ in $\boldsymbol{s}$, and we also say that there is a (necessarily unique) \emph{pointer} from $m$ to $m'$ in $\boldsymbol{s}$.

In addition, the work \cite{yamada2016dynamic} introduces the \emph{external/internal}-parity on each move of a game, where \emph{external} moves are `official' ones of the game, while \emph{internal} ones represent `internal calculation' by P in the game.
Hence, internal moves are `invisible' to O, and an internal O-move occuring in a position of the game is always a mere `dummy' of the last P-move (see the axiom \textsc{Dum} in Definition~\ref{DefLegalPositions}) so that the internal part of the position consists essentially of P's calculation only.

Having explained the rough idea on what (our variant of) games are, let us recall their precise definition below.
First, as a finitary representation of moves of games, the previous work \cite{yamada2019game} employs \emph{inner tags} for standard constructions on games except \emph{exponential} $\oc$, for which it uses \emph{outer tags}:
\begin{definition}[Inner tags \cite{yamada2019game}]
\label{DefInnerTags}
Let $\mathscr{W}$, $\mathscr{E}$, $\mathscr{N}$ and $\mathscr{S}$ be arbitrarily fixed, pairwise distinct elements. 
An \emph{\bfseries inner tag} is any finite sequence $\boldsymbol{s} \in \{ \mathscr{W}, \mathscr{E}, \mathscr{N}, \mathscr{S} \}^\ast$.
\end{definition}

\begin{definition}[Outer tags \cite{yamada2019game}]
\label{DefOuterTags}
An \emph{\bfseries outer tag} is an expression $\boldsymbol{e} \in (\{ \ell, \hbar, \Lbag, \Rbag \})^\ast$, where $\ell$, $\hbar$, $\Lbag$ and $\Rbag$ are arbitrarily fixed, pairwise distinct elements, generated by the grammar $\boldsymbol{e} \stackrel{\mathrm{df. }}{\equiv} \boldsymbol{\gamma} \mid \boldsymbol{e}_1 \hbar \ \! \boldsymbol{e}_2 \mid \Lbag \boldsymbol{e} \Rbag$, where $\boldsymbol{\gamma} \in \{ \ell, \hbar \}^\ast$.
\end{definition}

An outer tag $\boldsymbol{e}$ is to denote a finite sequence $\mathit{de}(\boldsymbol{e}) \in \mathbb{N}^\ast$ defined by:
\begin{align*}
\mathit{de}(\boldsymbol{\gamma}) &\stackrel{\mathrm{df. }}{=}  (i_1, i_2, \dots, i_k) \ \text{if $\boldsymbol{\gamma} = \ell^{i_1} \hbar \ \! \ell^{i_2} \hbar \dots \ell^{i_{k-1}} \hbar \ \! \ell^{i_k}$} \\
\mathit{de}(\boldsymbol{e}_1 \hbar \ \! \boldsymbol{e}_2) &\stackrel{\mathrm{df. }}{=} \mathit{de}(\boldsymbol{e}_1) . \mathit{de}(\boldsymbol{e}_2) \\
\mathit{de}(\Lbag \boldsymbol{e} \Rbag) &\stackrel{\mathrm{df. }}{=} (\wp(\mathit{de}(\boldsymbol{e})))
\end{align*}
where $\wp : \mathbb{N}^\ast \stackrel{\sim}{\to} \mathbb{N}$ is any recursive bijection fixed throughout the present paper such that $\wp (i_1, i_2, \dots, i_k) \neq \wp ( j_1, j_2, \dots, j_l)$ whenever $k \neq l$ (see, e.g.,  \cite{cutland1980computability}).

Unlike the previous work \cite{yamada2019game}, we embed the \emph{depth} of each occurrence of $\Lbag$ or $\Rbag$ into outer tags for our automata-theoretic implementation of PCF-computations:
\begin{definition}[Depths of $\Lbag$ and $\Rbag$]
\label{DefDepths}
In an outer tag $\boldsymbol{e}$, pairing each occurrence of $\Rbag$ with the most recent yet unpaired occurrence of $\Lbag$, we call one component of such a pair the \emph{\bfseries mate} of the other in $\boldsymbol{e}$.
The \emph{\bfseries depth} of an occurrence of $\Lbag$ in $\boldsymbol{e}$ is the number of previous occurrences of $\Lbag$ in $\boldsymbol{e}$ whose mate does not occur before that occurrence, and the \emph{\bfseries depth} of an occurrence of $\Rbag$ in $\boldsymbol{e}$ is the depth of its mate in $\boldsymbol{e}$.
\end{definition}

\begin{definition}[Extended outer tags]
\label{DefOuterTags}
An \emph{\bfseries extended outer tag} is an expression $\mathscr{O}(\boldsymbol{e})$ obtained from an outer tag $\boldsymbol{e}$ by replacing each occurrence of $\Lbag$ (resp. $\Rbag$) with $\Lbag . \lhd . \ell^d . \rhd$ (resp. $\Rbag . \lhd . \ell^d . \rhd$), where $d \in \mathbb{N}$ is the depth of the occurrence, and $\lhd$ and $\rhd$ are arbitrarily fixed, distinct elements such that $\{ \lhd, \rhd \} \cap \{ \ell, \hbar, \Lbag, \Rbag \} = \emptyset$. 
\end{definition}

\begin{remark*}
We write $\Sigma^\star$ and $\mathcal{T}$ for the sets of all outer tags and of all extended outer tags, respectively.  
We regard $\mathscr{O}$ as the obvious bijection $\Sigma^\star \stackrel{\sim}{\rightarrow} \mathcal{T}$ and define the  \emph{\bfseries decoding function on extended outer tags} to be the composition $\mathit{ede} \stackrel{\mathrm{df. }}{=} \mathcal{T} \stackrel{\mathscr{O}^{-1}}{\rightarrow} \Sigma^\star \stackrel{\mathit{de}}{\rightarrow} \mathbb{N}^\ast$.
\end{remark*}

\begin{convention*}
A \emph{\bfseries tag} refers to an inner or extended outer tag. 
\end{convention*}

\begin{notation*}
We often abbreviate expressions $\Lbag . \lhd . \ell^d . \rhd$ and $\Rbag . \lhd . \ell^d . \rhd$ as $\Lbag^{\underline{d}}$ and $\Rbag^{\underline{d}}$, respectively. 
Given $\boldsymbol{e} \in \mathcal{T}$, we write $\boldsymbol{e}^{+} \in \mathcal{T}$ for the extended outer tag obtained from $\boldsymbol{e}$ by replacing each occurrence of $\Lbag^{\underline{d}}$ (resp. $\Rbag^{\underline{d}}$) with that of $\Lbag^{\underline{d+1}}$ (resp. $\Rbag^{\underline{d+1}}$).
\end{notation*}

Using inner and (extended) outer tags, the previous work \cite{yamada2019game} focuses on games whose moves are all \emph{tagged elements} defined as follows:
\begin{definition}[Inner elements \cite{yamada2019game}]
\label{DefInnerElements}
An \emph{\bfseries inner element} is a finitely nested pair $( \dots ((m, t_1), t_2), \dots, t_k)$, often written $m_{t_1 t_2 \dots t_k}$, such that $m$ is a distinguished element, called the \emph{\bfseries substance} of $m_{t_1 t_2 \dots t_k}$, and $t_1 t_2 \dots t_k$ is an inner tag.
\end{definition}

\begin{definition}[Tagged elements \cite{yamada2019game}]
\label{DefTaggedElements}
A \emph{\bfseries tagged element} is a pair $[m_{t_1 t_2 \dots t_k}]_{\boldsymbol{e}} \stackrel{\mathrm{df. }}{=} (m_{t_1 t_2 \dots t_k}, \boldsymbol{e})$ of an inner element $m_{t_1 t_2 \dots t_k}$ and an extended outer tag $\boldsymbol{e}$.
\end{definition}

\begin{notation*}
We often abbreviate an inner element $m_{t_1 t_2 \dots t_k}$ as $m$ if the inner tag $t_1 t_2 \dots t_k$ is not important. 
Similarly, we often abbreviate a tagged element $[m]_{\boldsymbol{e}}$ as $m$ if the extended outer tag $\boldsymbol{e}$ is not important. 
\end{notation*}

Now, we are ready to recall (a simplified version of) \emph{games}:

\begin{definition}[Arenas \cite{yamada2019game}]
\label{DefArenas}
An \emph{\bfseries arena} is a triple $G = (M_G, \lambda_G, \Delta_G)$ such that:
\begin{itemize}

\item $M_G$ is a set of tagged elements, called \emph{\bfseries moves}, such that the set $\pi_1(M_G)$ of inner elements is finite, equipped with a distinguished subset $M_G^{\mathsf{Init}} \subseteq M_G$ of \emph{\bfseries initial moves};

\item $\lambda_G$ is a map $M_G \rightarrow \{ \mathsf{O}, \mathsf{P} \} \times \{ \mathsf{E}, \mathsf{I} \}$, where $\mathsf{O}$, $\mathsf{P}$, $\mathsf{E}$ and $\mathsf{I}$ are arbitrary, pairwise distinct symbols, called the \emph{\bfseries labeling function}, such that $\forall m \in M_G^{\mathsf{Init}} . \ \! \lambda_G(m) = (\mathsf{O}, \mathsf{E})$;

\item $\Delta_G$ is a bijection $M_G^{\mathsf{PI}} \stackrel{\sim}{\rightarrow} M_G^{\mathsf{OI}}$, where $M_G^{\mathsf{XY}} \stackrel{\mathrm{df. }}{=} \lambda_G^{-1}(\mathsf{X}, \mathsf{Y})$, $\mathsf{X} \in \{ \mathsf{O}, \mathsf{P} \}$ and $\mathsf{Y} \in \{ \mathsf{E}, \mathsf{I} \}$, called the \emph{\bfseries dummy function}, such that there is a \emph{finite} partial map $\delta_G$ on inner tags with $\forall [m_{\boldsymbol{t}}]_{\boldsymbol{e}} \in M_G^{\mathsf{PI}}, [n_{\boldsymbol{u}}]_{\boldsymbol{f}} \in M_G^{\mathsf{OI}} . \ \!\Delta_G([m_{\boldsymbol{t}}]_{\boldsymbol{e}}) = [n_{\boldsymbol{u}}]_{\boldsymbol{f}} \Rightarrow m = n \wedge \boldsymbol{e} = \boldsymbol{f} \wedge \boldsymbol{u} = \delta_G(\boldsymbol{t})$.

\end{itemize}

\end{definition}

\begin{definition}[Legal positions \cite{yamada2019game}]
\label{DefLegalPositions}
A \emph{\bfseries legal positions} of an arena $G$ is a finite sequence $\boldsymbol{s} \in M_G^\ast$ (equipped with \emph{pointers} given below) that satisfies:
\begin{itemize}

\item \textsc{(Alt)} $\forall i \in \{ 1, 2, \dots |\boldsymbol{s}| \} . \ \! \mathsf{Odd}(i) \Leftrightarrow \lambda_G^{\mathsf{OP}}(\boldsymbol{s}(i)) = \mathsf{O}$, where $\lambda_G^{\mathsf{OP}} \stackrel{\mathrm{df. }}{=} \pi_1 \circ \lambda_G$;

\item \textsc{(Jus)} To each occurrence $\boldsymbol{s}(i)$ of a non-initial move, a unique occurrence $\boldsymbol{s}(j)$ such that $0 < j < i$, $\mathsf{Even}(i) \Leftrightarrow \mathsf{Odd}(j)$ and $\lambda_G^{\mathsf{EI}}(\boldsymbol{s}(i)) \neq \lambda_G^{\mathsf{EI}}(\boldsymbol{s}(j)) \Rightarrow \lambda_G^{\mathsf{OP}}(\boldsymbol{s}(i)) = \mathsf{P}$, where $\lambda_G^{\mathsf{EI}} \stackrel{\mathrm{df. }}{=} \pi_2 \circ \lambda_G$, called the \emph{\bfseries justifier} of $\boldsymbol{s}(i)$ and written $\mathcal{J}_{\boldsymbol{s}}(\boldsymbol{s}(i))$, is assigned, for which we say that there is a \emph{\bfseries pointer} from $\boldsymbol{s}(i)$ to $\boldsymbol{s}(j)$ in $\boldsymbol{s}$;

\item \textsc{(EI)} $\boldsymbol{s} = \boldsymbol{t} . m . n . \boldsymbol{u}$ and $\lambda_G^{\mathsf{EI}}(m) \neq \lambda_G^{\mathsf{EI}}(n)$ imply $\lambda_G^{\mathsf{OP}}(m) = \mathsf{O}$;

\item \textsc{(Dum)} $\boldsymbol{s} = \boldsymbol{t} . p . o' . \boldsymbol{u} . p' . o$ (resp. $\boldsymbol{s} = \boldsymbol{t} . o' . \boldsymbol{u} . p' . o$), $o \in M_G^{\mathsf{OI}}$, $o' \in M_G^{\mathsf{OI}}$ (resp. $o' \in M_G^{\mathsf{OE}}$), and $o' = \mathcal{J}_{\boldsymbol{s}}(p')$ imply $o = \Delta_G(p')$ and $\mathcal{J}_{\boldsymbol{s}}(o) = p$ (resp. $\mathcal{J}_{\boldsymbol{s}}(o) = p'$).

\end{itemize}
We write $\mathscr{L}_G$ for the set of all legal positions of $G$.
\end{definition}

\begin{definition}[Games \cite{yamada2019game}]
\label{DefGames}
A \emph{\bfseries game} is a quadruple $G = (M_G, \lambda_G, \Delta_G, P_G)$ such that the triple $(M_G, \lambda_G, \Delta_G)$, also written $G$, is an arena, and $P_G \subseteq \mathscr{L}_G$ is a non-empty, prefix-closed set of \emph{\bfseries (valid) positions} of $G$.
A \emph{\bfseries play} of $G$ is a (finitely or infinitely) increasing sequence $(\boldsymbol{\epsilon}, m_1, m_1m_2, \dots)$ of positions of $G$.
\end{definition}

\begin{convention*}
Given a game $G$, a move $m \in M_G$ is specifically called an \emph{\bfseries O-move} (resp. a \emph{\bfseries P-move}) if $\lambda_G^{\mathsf{OP}}(m) = \mathsf{O}$ (resp. if $\lambda_G^{\mathsf{OP}}(m) = \mathsf{P}$), and \emph{\bfseries external} (resp. \emph{\bfseries internal}) if $\lambda_G^{\mathsf{EI}}(m) = \mathsf{E}$ (resp. if $\lambda_G^{\mathsf{EI}}(m) = \mathsf{I}$).
\end{convention*}

\begin{notation*}
Given a game $G$, we write $\boldsymbol{s} = \boldsymbol{t}$ for any positions $\boldsymbol{s}, \boldsymbol{t} \in P_G$ iff $\boldsymbol{s}$ and $\boldsymbol{t}$ are the same finite sequence of moves equipped with the same structure of pointers, i.e., $\forall i \in \{ 1, 2, \dots, |\boldsymbol{s}| \} . \ \! \boldsymbol{s}(i) \in M_G^{\mathsf{Init}} \Leftrightarrow \boldsymbol{t}(i) \in M_G^{\mathsf{Init}} \wedge (\mathcal{J}_{\boldsymbol{s}}(\boldsymbol{s}(i)) = \boldsymbol{s}(j) \wedge \mathcal{J}_{\boldsymbol{t}}(\boldsymbol{t}(i)) = \boldsymbol{t}(k) \Rightarrow j = k)$.
\end{notation*}

Let us remark again that games given in Definition~\ref{DefGames} are a simplified version of what is given in \cite{yamada2019game}, where \emph{enabling relations}, \emph{qustions/answers}, \emph{views}, \emph{visibility}, \emph{priority orders}, etc.
are omitted.
Of course, we could recall the original definition, but the simplified one suffices for the present work.
Moreover, it is easy to see that theorems and constructions on games in \cite{yamada2019game} recalled below are valid for the simplified games as well, for which simplified proofs are applied. 

\begin{definition}[Subgames \cite{yamada2019game}]
A game $H$ is a \emph{\bfseries subgame} of a game $G$, written $H \trianglelefteqslant G$, if $M_H \subseteq M_G$, $\lambda_H = \lambda_G \upharpoonright M_H$, $\Delta_H = \Delta_G \upharpoonright M_H$ and $P_H \subseteq P_G$.
\end{definition}

The intuition behind the notion of games has been explained at the beginning of Section~\ref{GamesAndStrategies}.
Here, let us comment briefly on the axioms on an arbitrary game $G$:
\begin{itemize}

\item The set $\pi_1(M_G)$ is finite so that each inner element of $G$ is `recognizable';

\item Each initial move of $G$ is an external O-move because internal moves are `invisible' to O, and O has to initiate a play of $G$ (by the axiom \textsc{Alt});

\item $\Delta_G(m) \in M_G^{\mathsf{OI}}$ is the `dummy' of each $m \in M_G^{\mathsf{PI}}$ such that they differ only in inner tags, and the inner tag of the former is obtainable from that of the latter by a finitary computation $\delta_G$;

\item The set $P_G$ is non-empty for the domain-theoretic reason \cite{amadio1998domains}, and prefix-closed because each non-empty position or `moment' of $G$ must have the previous `moment';

\item Each position $\boldsymbol{s}$ of $G$ is a finite sequence such that $(\lambda_G^{\mathsf{OP}})^\ast(\boldsymbol{s}) = \mathsf{OPOP} \dots$ (by \textsc{Alt}) equipped with justifiers on occurrences of non-initial moves (by the axiom \textsc{Jus}), where note that the first element $\boldsymbol{s}(1)$ must be an initial O-move;

\item The axiom \textsc{EI} states that each external/internal-parity change during a play of $G$ must be made by P because internal moves are `invisible' to O;

\item The axiom \textsc{Dum} requires that each internal O-move in a position of $G$ must be the mere `dummy' of the previous internal P-move, where the slightly involved pointers capture the phenomenon of \emph{concatenation} $\ddagger$ of games (Appendix~\ref{DefConcatenationOfGames}). 

\end{itemize}

A game is \emph{\bfseries normalized} if it has no internal moves.
There is an important operation (Definition~\ref{DefOmegaHidingOnGames}) that maps every game to a normalized one:
\begin{definition}[J-subsequences \cite{yamada2019game}]
\label{DefJSubsequences}
Let $G$ be a game, and $\boldsymbol{s} \in P_G$. 
A \emph{\bfseries j-subsequence} of $\boldsymbol{s}$ is a subsequence $\boldsymbol{t}$ of $\boldsymbol{s}$ equipped with pointers such that $\mathcal{J}_{\boldsymbol{t}}(n) = m$ iff there are elements $\boldsymbol{s}(i_1), \boldsymbol{s}(i_2), \dots, \boldsymbol{s}(i_{2k-1}), \boldsymbol{s}(i_{2k})$ of $\boldsymbol{s}$ eliminated in $\boldsymbol{t}$ such that $\mathcal{J}_{\boldsymbol{s}}(n) = \boldsymbol{s}(i_1) \wedge \mathcal{J}_{\boldsymbol{s}}(\boldsymbol{s}(i_1)) = \boldsymbol{s}(i_2) \dots \wedge \mathcal{J}_{\boldsymbol{s}}(\boldsymbol{s}(i_{2k-1})) = \boldsymbol{s}(i_{2k}) \wedge \mathcal{J}_{\boldsymbol{s}}(\boldsymbol{s}(i_{2k})) = m$.
\end{definition}

\begin{definition}[$\omega$-hiding operation on games \cite{yamada2019game}]
\label{DefOmegaHidingOnGames}
The \emph{\bfseries $\boldsymbol{\omega}$-hiding operation} $\mathcal{H}^\omega$ on games maps each game $G$ to the normalized one $\mathcal{H}^{\omega}(G)$ given by:
\begin{itemize}

\item $M_{\mathcal{H}^\omega(G)} \stackrel{\mathrm{df. }}{=} \{ m \in M_G \mid \lambda_G^{\mathsf{EI}}(m) = \mathsf{E} \ \! \}$;

\item $M_{\mathcal{H}^\omega(G)}^{\mathsf{Init}} \stackrel{\mathrm{df. }}{=} M_G^{\mathsf{Init}}$;

\item $\lambda_{\mathcal{H}^\omega(G)} \stackrel{\mathrm{df. }}{=} \lambda_G \upharpoonright M_{\mathcal{H}^\omega(G)}$;

\item $\Delta_{\mathcal{H}^\omega(G)} \stackrel{\mathrm{df. }}{=} \emptyset$;

\item $P_{\mathcal{H}^\omega(G)} \stackrel{\mathrm{df. }}{=} \{ \mathcal{H}^\omega(\boldsymbol{s}) \mid \boldsymbol{s} \in P_G \ \! \}$, where $\mathcal{H}^\omega(\boldsymbol{s})$ is the j-subsequence of $\boldsymbol{s}$ that consists of external moves of $G$.

\end{itemize}
\end{definition}

It is shown in \cite{yamada2016dynamic} that the $\omega$-hiding operation $\mathcal{H}^\omega$ on games is well-defined (e.g., the axiom \textsc{Alt} on $\mathcal{H}^\omega(G)$ is satisfied by the axioms \textsc{Alt} and \textsc{EI} on a given game $G$).
Originally in \cite{yamada2016dynamic}, the \emph{(one-step) hiding operation} $\mathcal{H}$ on games is defined to capture (small-step) operational semantics, and $\mathcal{H}^\omega$ is defined to be the countably-infinite iteration of $\mathcal{H}$.
Nevertheless, we need only $\mathcal{H}^\omega$ for the present work, and so we have introduced it directly as above.

\begin{notation*}
Given a game $G$, we often write $M_G^{\mathsf{Ext}}$ for the set $M_{\mathcal{H}^\omega(G)}$ of all external moves of $G$, and $M_G^{\mathsf{Int}}$ for the set $M_G \setminus M_G^{\mathsf{Ext}}$ of all internal moves of $G$.
\end{notation*}

On the other hand, a \emph{strategy} on a game is what tells P which move (together with its justifier) she should perform at each of her turns (i.e., odd-length positions) of the game.
More precisely, it is defined as follows: 
\begin{definition}[Strategies \cite{yamada2019game}]
\label{DefStrategies}
A \emph{\bfseries strategy} $\sigma$ on a game $G$, written $\sigma : G$, is a subset $\sigma \subseteq P_G^{\mathsf{Even}}$ that satisfies:
\begin{itemize}

\item \textsc{(S1)} It is non-empty and \emph{even-prefix-closed} (i.e., $\boldsymbol{s}mn \in \sigma \Rightarrow \boldsymbol{s} \in \sigma$);

\item \textsc{(S2)} It is \emph{deterministic} ($\boldsymbol{s}mn, \boldsymbol{s'}m'n' \in \sigma \wedge \boldsymbol{s}m = \boldsymbol{s'}m' \Rightarrow \boldsymbol{s}mn = \boldsymbol{s'}m'n'$).

\end{itemize}
\end{definition}

\begin{remark*}
We usually skip describing justifiers in strategies if they are obvious. 
\end{remark*}

\begin{proposition}[Strategies on subgames \cite{yamada2016dynamic}]
\label{PropStrategiesOnSubgames}
$A \trianglelefteqslant B \wedge \alpha : A \Rightarrow \alpha : B$.
\end{proposition}

A strategy $\sigma : G$ is \emph{\bfseries normalized} if no internal moves occur in any element of $\sigma$.
Similarly to the case of games, there is an operation that normalizes strategies:
\begin{definition}[$\omega$-hiding operation on strategies \cite{yamada2019game}]
The \emph{\bfseries $\boldsymbol{\omega}$-hiding operation} $\mathcal{H}^\omega$ on strategies maps $(\sigma : G) \mapsto \{ \boldsymbol{s} \natural \mathcal{H}_G^d \mid \boldsymbol{s} \in \sigma \ \! \}$, where $\boldsymbol{s} \natural \mathcal{H}_G^d \stackrel{\mathrm{df. }}{=} \begin{cases} \mathcal{H}_G^d(\boldsymbol{s}) &\text{if $\boldsymbol{s}$ ends with an external move;} \\ \boldsymbol{t} &\text{otherwise, where $\mathcal{H}_G^d(\boldsymbol{s}) = \boldsymbol{t}m$.} \end{cases}$
\end{definition}

Similarly to the case of games, the $\omega$-hiding operation $\mathcal{H}^\omega$ on strategies is originally defined in \cite{yamada2016dynamic} as the countably-infinite iteration of the \emph{(one-step) hiding operation} $\mathcal{H}$ on strategies.
But again, we need only the $\omega$-hiding one $\mathcal{H}^\omega$ for the present work, and therefore, we have directly defined it as above. 

\begin{theorem}[Hiding theorem \cite{yamada2016dynamic}]
\label{ThmHidingTheorem}
If $\sigma : G$, then $\mathcal{H}^\omega(\sigma) : \mathcal{H}^\omega(G)$.
\end{theorem} 

It is worth noting here that normalized games and strategies are equivalent to conventional games and strategies \cite{abramsky1999game}, respectively \cite{yamada2016dynamic}. Moreover, identifying these equivalent variants of games and strategies, the $\omega$-hiding operations $\mathcal{H}^\omega$ on games and strategies form a 2-functor from the bicategory $\mathcal{DG}$ of dynamic games and strategies to the category (or the trivial 2-category) $\mathcal{G}$ of conventional games and strategies; see \cite{yamada2016dynamic} for the details.

\begin{convention*}
Henceforth, we shall employ the identification, i.e., we regard normalized games and strategies as conventional games and strategies. 
\end{convention*}

\begin{example}
The \emph{\bfseries terminal game} $T$ is given by $T \stackrel{\mathrm{df. }}{=} (\emptyset, \emptyset, \emptyset, \{ \boldsymbol{\epsilon} \})$.
The unique strategy $\top \stackrel{\mathrm{df. }}{=} \{ \boldsymbol{\epsilon} \}$ on $T$ is called the \emph{\bfseries top strategy}. 
\end{example}

\begin{example}
Consider the \emph{\bfseries boolean game} $\boldsymbol{2}$, whose maximal positions are $[\hat{q}] [\mathit{tt}]$ and $[\hat{q}] [\mathit{ff}]$, where $[\mathit{tt}]$ and $[\mathit{ff}]$ are both justified by $[\hat{q}]$, or diagrammatically:
\begin{center}
\begin{tabular}{ccccccccc}
&$\boldsymbol{2}$ &&&&&$\boldsymbol{2}$& \\ \cline{1-3} \cline{6-8}
& \tikzmark{cbool1} $[\hat{q}]$&&&&& \tikzmark{cbool2} $[\hat{q}]$& \\
& \tikzmark{dbool1} $[\mathit{tt}]$&&&&& \tikzmark{dbool2} $[\mathit{ff}]$ & 
\end{tabular}
\begin{tikzpicture}[overlay, remember picture, yshift=.25\baselineskip]
\draw [->] ({pic cs:dbool1}) [bend left] to ({pic cs:cbool1});
\draw [->] ({pic cs:dbool2}) [bend left] to ({pic cs:cbool2});
\end{tikzpicture}
\end{center}
where each arrow $m' \leftarrow m$ means that $m'$ is the justifier of $m$.
We shall employ this notation in the rest of the paper. 
These plays can be read as follows:
\begin{enumerate}

\item O's question $[\hat{q}]$ for an output (`What is your boolean value?');

\item P's answer $[\mathit{tt}]$ (resp. $[\mathit{ff}]$) to $[\hat{q}]$ (`My value is \emph{true} (resp. \emph{false})!').

\end{enumerate}

Formally, the game $\boldsymbol{2}$ is given by:
\begin{itemize}

\item $M_{\boldsymbol{2}} \stackrel{\mathrm{df. }}{=} \{ [\hat{q}], [\mathit{tt}], [\mathit{ff}] \}$; $M_{\boldsymbol{2}}^{\mathsf{Init}} \stackrel{\mathrm{df. }}{=} \{ [\hat{q}] \}$;

\item $\lambda_{\boldsymbol{2}} : [\hat{q}] \mapsto (\mathsf{O}, \mathsf{E}), [\mathit{tt}] \mapsto (\mathsf{P}, \mathsf{E}), [\mathit{ff}] \mapsto (\mathsf{P}, \mathsf{E})$;

\item $\Delta_{\boldsymbol{2}} \stackrel{\mathrm{df. }}{=} \emptyset$;

\item $P_{\boldsymbol{2}} \stackrel{\mathrm{df. }}{=} \mathsf{Pref}(\{ [\hat{q}] [\mathit{tt}], [\hat{q}] [\mathit{ff}] \})$, where $[\mathit{tt}]$ and $[\mathit{ff}]$ are both justified by $[\hat{q}]$.


\end{itemize}
\end{example}

As expected, the boolean values $\mathit{tt}, \mathit{ff} \in \mathbb{B} \stackrel{\mathrm{df. }}{=} \{ \mathit{tt}, \mathit{ff} \}$, where $\mathit{tt}$ (resp. $\mathit{ff}$) denotes \emph{true} (resp. \emph{false}), are represented respectively by the strategies $\underline{\mathit{tt}}, \underline{\mathit{ff}} : \boldsymbol{2}$ given by $\underline{\mathit{tt}} \stackrel{\mathrm{df. }}{=} \{ \boldsymbol{\epsilon}, [\hat{q}] [\mathit{tt}] \}$ and $\underline{\mathit{ff}} \stackrel{\mathrm{df. }}{=} \{ \boldsymbol{\epsilon}, [\hat{q}] [\mathit{ff}] \}$.

\begin{example}
Let us recall the \emph{\bfseries lazy natural number game} $\mathcal{N}$, which is a game for natural numbers, playing a fundamental role in the previous work \cite{yamada2019game}.

A maximal position of $\mathcal{N}$ is either of the following forms:
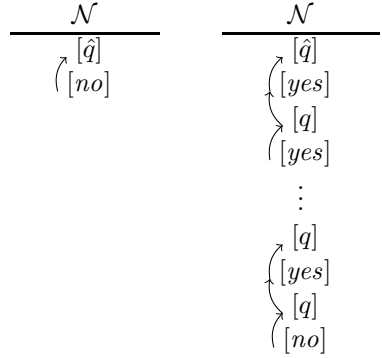
\begin{center}
\begin{tabular}{ccccccccc}
&$\mathcal{N}$ &&&&&$\mathcal{N}$& \\ \cline{1-3} \cline{6-8}
&\tikzmark{cLazyN10} $[\hat{q}]$&&&&&\tikzmark{cLazyN1} $[\hat{q}]$& \\
&\tikzmark{dLazyN10} $[\mathit{no}]$&&&&& \tikzmark{cLazyN2} $[\mathit{yes}]$ & \\
&&&&&& \tikzmark{dLazyN2} $[q]$ & \\
&&&&&& \tikzmark{cLazyN4} $[\mathit{yes}]$ & \\
&&&&&&$\vdots$& \\
&&&&&& \tikzmark{cLazyN6} $[q]$ & \\
&&&&&& \tikzmark{cLazyN7} $[\mathit{yes}]$ & \\
&&&&&& \tikzmark{dLazyN7} $[q]$ & \\
&&&&&& \tikzmark{dLazyN8} $[\mathit{no}]$ & \\
\end{tabular}
\begin{tikzpicture}[overlay, remember picture, yshift=.25\baselineskip]
\draw [->] ({pic cs:cLazyN2}) [bend left] to ({pic cs:cLazyN1});
\draw [->] ({pic cs:dLazyN2}) [bend left] to ({pic cs:cLazyN2});
\draw [->] ({pic cs:cLazyN4}) [bend left] to ({pic cs:dLazyN2});
\draw [->] ({pic cs:cLazyN7}) [bend left] to ({pic cs:cLazyN6});
\draw [->] ({pic cs:dLazyN7}) [bend left] to ({pic cs:cLazyN7});
\draw [->] ({pic cs:dLazyN8}) [bend left] to ({pic cs:dLazyN7});
\draw [->] ({pic cs:dLazyN10}) [bend left] to ({pic cs:cLazyN10});
\end{tikzpicture}
\end{center}
where the number $n$ of $[\mathit{yes}]$ in the positions ranges over $\mathbb{N}$, which represents the number intended by P.
Note that the initial question $[\hat{q}]$ must be distinguished from the non-initial one $[q]$ for the axiom \textsc{Jus} (i.e., a move cannot be both initial and non-initial).

Formally, the game $\mathcal{N}$ is defined by:
\begin{itemize}

\item $M_{\mathcal{N}} \stackrel{\mathrm{df. }}{=} \{ [\hat{q}], [q], [\mathit{yes}], [\mathit{no}] \}$; $M_{\mathcal{N}}^{\mathsf{Init}} \stackrel{\mathrm{df. }}{=} \{ [\hat{q}] \}$;

\item $\lambda_{\mathcal{N}} : [\hat{q}] \mapsto (\mathsf{O}, \mathsf{E}), [q] \mapsto (\mathsf{O}, \mathsf{E}), [\mathit{yes}] \mapsto (\mathsf{P}, \mathsf{E}), [\mathit{no}] \mapsto (\mathsf{P}, \mathsf{E})$;

\item $\Delta_{\mathcal{N}} \stackrel{\mathrm{df. }}{=} \emptyset$;

\item $P_{\mathcal{N}} \stackrel{\mathrm{df. }}{=} \mathsf{Pref}(\{ [\hat{q}] . ([\mathit{yes}] [q])^n . [\mathit{no}] \mid n \in \mathbb{N} \ \! \})$, where each occurrence of a non-initial move is justified by the last occurrence.

\end{itemize}

As expected, each natural number $n \in \mathbb{N}$ is represented by the strategy $\underline{n} : \mathcal{N}$ given by $\underline{n} \stackrel{\mathrm{df. }}{=} \mathsf{Pref}(\{ [\hat{q}] . ([\mathit{yes}] [q])^n . [\mathit{no}] \})^{\mathsf{Even}}$.
\end{example}

\subsection{Constructions on games and strategies}
\label{Constructions}
Next, let us recall standard constructions on games and strategies.
First, there is \emph{tensor (product)} $\otimes$ on games.
A position $\boldsymbol{s}$ of the tensor $A \otimes B$ of games $A$ and $B$ is an interleaving mixture of a position $\boldsymbol{t}$ of $A$ and a position $\boldsymbol{u}$ of $B$ such that an  $AB$-parity change (i.e., a switch between $\boldsymbol{t}$ and $\boldsymbol{u}$) is always by O. 

For $A \otimes B$, we need to take a disjoint union $M_{A \otimes B} \stackrel{\mathrm{df. }}{=} M_A + M_B$ to distinguish moves of $A$ from those of $B$.
The previous work \cite{yamada2019game} formalizes `tags' for $M_{A \otimes B}$ by inner tags $\mathscr{W}$ and $\mathscr{E}$, e.g., typical positions of the tensor $\mathcal{N} \otimes \mathcal{N}$ are:\footnote{The diagram is only make it explicit which component $\mathcal{N}$ each move belongs to; the positions are just finite sequences \if0 $[(\hat{q}, \mathscr{W})] [(\mathit{yes}, \mathscr{W})] [(\hat{q}, \mathscr{E})] [(\mathit{no}, \mathscr{E})] [(q, \mathscr{W})] [(\mathit{no}, \mathscr{W})]$ and $[(\hat{q}, \mathscr{E})] [(\mathit{yes}, \mathscr{E})] [(q, \mathscr{E})] [(\mathit{yes}, \mathscr{E})] [(\hat{q}, \mathscr{W})] [(\mathit{no}, \mathscr{W})] [(q, \mathscr{E})] [(\mathit{no}, \mathscr{E})]$ \fi equipped with the pointers represented by the arrows in the diagrams.}
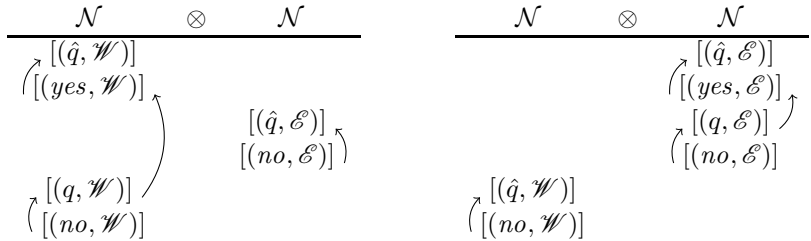
\begin{center}
\begin{tabular}{ccccccccc}
$\mathcal{N}$ & $\otimes$ & $\mathcal{N}$ &&&& $\mathcal{N}$ & $\otimes$ & $\mathcal{N}$ \\ \cline{1-3} \cline{7-9}
\tikzmark{LazyTensorC1} $[(\hat{q}, \mathscr{W})]$ &&&&&&&& \tikzmark{LazyTensorC6}  $[(\hat{q}, \mathscr{E})]$ \\
\tikzmark{LazyTensorD1} $[(\mathit{yes}, \mathscr{W})]$ \tikzmark{LazyTensorC2}&&&&&&&&\tikzmark{LazyTensorD6} $[(\mathit{yes}, \mathscr{E})]$ \tikzmark{LazyTensorC7} \\
&&$[(\hat{q}, \mathscr{E})]$ \tikzmark{LazyTensorC4}&&&&&& \tikzmark{LazyTensorC8} $[(q, \mathscr{E})]$ \tikzmark{LazyTensorD7} \\
&&$[(\mathit{no}, \mathscr{E})]$ \tikzmark{LazyTensorD4}&&&&&& \tikzmark{LazyTensorD8} $[(\mathit{no}, \mathscr{E})]$ \tikzmark{LazyTensorC9} \\
\tikzmark{LazyTensorC3} $[(q, \mathscr{W})]$ \tikzmark{LazyTensorD2}&&&&&& \tikzmark{LazyTensorC5} $[(\hat{q}, \mathscr{W})]$&& \\
\tikzmark{LazyTensorD3} $[(\mathit{no}, \mathscr{W})]$ &&&&&& \tikzmark{LazyTensorD5} $[(\mathit{no}, \mathscr{W})]$ && 
\end{tabular}
\begin{tikzpicture}[overlay, remember picture, yshift=.25\baselineskip]
\draw [->] ({pic cs:LazyTensorD1}) [bend left] to ({pic cs:LazyTensorC1});
\draw [->] ({pic cs:LazyTensorD2}) [bend right] to ({pic cs:LazyTensorC2});
\draw [->] ({pic cs:LazyTensorD3}) [bend left] to ({pic cs:LazyTensorC3});
\draw [->] ({pic cs:LazyTensorD4}) [bend right] to ({pic cs:LazyTensorC4});
\draw [->] ({pic cs:LazyTensorD5}) [bend left] to ({pic cs:LazyTensorC5});
\draw [->] ({pic cs:LazyTensorD6}) [bend left] to ({pic cs:LazyTensorC6});
\draw [->] ({pic cs:LazyTensorD7}) [bend right] to ({pic cs:LazyTensorC7});
\draw [->] ({pic cs:LazyTensorD8}) [bend left] to ({pic cs:LazyTensorC8});
\end{tikzpicture}
\end{center}

Formally, tensor of games is given as follows:
\begin{definition}[Tensor of games \cite{yamada2019game}]
The \emph{\bfseries tensor (product)} $A \otimes B$ of games $A$ and $B$ is defined by:
\begin{itemize}

\item $M_{A \otimes B} \stackrel{\mathrm{df. }}{=} \{ [(a, \mathscr{W})]_{\boldsymbol{e}} \mid [a]_{\boldsymbol{e}} \in M_A \ \! \} \cup \{ [(b, \mathscr{E})]_{\boldsymbol{f}} \mid [b]_{\boldsymbol{f}} \in M_B \ \! \}$;

\item $M_{A \otimes B}^{\mathsf{Init}} \stackrel{\mathrm{df. }}{=} \{ [(a, \mathscr{W})]_{\boldsymbol{e}} \mid [a]_{\boldsymbol{e}} \in M_A^{\mathsf{Init}} \ \! \} \cup \{ [(b, \mathscr{E})]_{\boldsymbol{f}} \mid [b]_{\boldsymbol{f}} \in M_B^{\mathsf{Init}} \ \! \}$;

\item $\lambda_{A \otimes B} ([(m, X)]_{\boldsymbol{e}}) \stackrel{\mathrm{df. }}{=} \begin{cases} \lambda_A([m]_{\boldsymbol{e}}) &\text{if $X = \mathscr{W}$;} \\ \lambda_B([m]_{\boldsymbol{e}}) &\text{otherwise;} \end{cases}$

\item $\Delta_{A \otimes B} ([(m, X)]_{\boldsymbol{e}}) \stackrel{\mathrm{df. }}{=} \begin{cases} [(m', \mathscr{W})]_{\boldsymbol{e}} &\text{if $X = \mathscr{W}$, where $\Delta_A([m]_{\boldsymbol{e}}) = [m']_{\boldsymbol{e}}$;} \\ [(m'', \mathscr{E})]_{\boldsymbol{e}} &\text{otherwise, where $\Delta_B([m]_{\boldsymbol{e}}) = [m'']_{\boldsymbol{e}}$;} \end{cases}$

\item $P_{A \otimes B} \stackrel{\mathrm{df. }}{=} \{ \boldsymbol{s} \in \mathscr{L}_{A \otimes B} \mid \boldsymbol{s} \upharpoonright \mathscr{W} \in P_A, \boldsymbol{s} \upharpoonright \mathscr{E} \in P_B \ \! \}$, where $\boldsymbol{s} \upharpoonright X$ is the j-subsequence of $\boldsymbol{s}$ that consists of moves of the form $[(m, X)]_{\boldsymbol{e}}$ changed into $[m]_{\boldsymbol{e}}$.

\end{itemize}
\end{definition}

As shown in \cite{abramsky1997semantics}, in fact only O can switch between the component games $A$ and $B$ during a play of the tensor $A \otimes B$ (by the axiom \textsc{Alt}).

Next, let us recall a fundamental construction $\oc$ on games, called \emph{exponential}, which is the countably infinite iteration of tensor, i.e., $!A \cong A \otimes A \otimes A \dots$ for each game $A$, where the `tag' for each copy of $A$ is typically given by a natural number $i \in \mathbb{N}$ \cite{abramsky1999game}. In contrast, as in \cite{yamada2019game}, we implement the `tags' in a finitary manner by extended outer tags: Each move $[m]_{\boldsymbol{e}} \in M_A$ is duplicated as $[m]_{\Lbag^{\underline{0}} \boldsymbol{f}^+ \Rbag^{\underline{0}} \hbar \boldsymbol{e}} \in M_{\oc A}$ for each $\boldsymbol{f} \in \mathcal{T}$, intended to be the `tag' $i \in \mathbb{N}$ such that $\mathit{ede}(\Lbag^{\underline{0}} \boldsymbol{f}^+ \Rbag^{\underline{0}}) = (i)$.

For example, some typical positions of the exponential $\oc \boldsymbol{2}$ are as follows:
\begin{center}
\begin{tabular}{ccccc}
$\oc \boldsymbol{2}$ &&&& $\oc \boldsymbol{2}$ \\ \cline{1-1} \cline{5-5}
\tikzmark{ExC1} $[\hat{q}]_{\Lbag^{\underline{0}} \Rbag^{\underline{0}} \hbar}$&&&&\tikzmark{ExC3} $[\hat{q}]_{\Lbag^{\underline{0}} \ell^2 \hbar \ell^3 \hbar \ell^5 \Rbag^{\underline{0}} \hbar}$ \\
\tikzmark{ExD1} $[\mathit{tt}]_{\Lbag^{\underline{0}} \Rbag^{\underline{0}} \hbar}$&&&&\tikzmark{ExD3} $[\mathit{tt}]_{\Lbag^{\underline{0}} \ell^2 \hbar \ell^3 \hbar \ell^5 \Rbag^{\underline{0}} \hbar}$ \\
\tikzmark{ExC2} $[\hat{q}]_{\Lbag^{\underline{0}} \ell \Rbag^{\underline{0}} \hbar}$&&&&\tikzmark{ExC4} $[\hat{q}]_{\Lbag^{\underline{0}} \Lbag^{\underline{1}} \ell^2 \hbar \ell^3 \Rbag^{\underline{1}} \hbar \ell^5 \Rbag^{\underline{0}} \hbar}$ \\
\tikzmark{ExD2} $[\mathit{ff}]_{\Lbag^{\underline{0}} \ell \Rbag^{\underline{0}} \hbar}$&&&&\tikzmark{ExD4} $[\mathit{tt}]_{\Lbag^{\underline{0}} \Lbag^{\underline{1}} \ell^2 \hbar \ell^3 \Rbag^{\underline{1}} \hbar \ell^5 \Rbag^{\underline{0}} \hbar}$ 
\end{tabular}
\begin{tikzpicture}[overlay, remember picture, yshift=.25\baselineskip]
\draw [->] ({pic cs:ExD1}) [bend left] to ({pic cs:ExC1});
\draw [->] ({pic cs:ExD2}) [bend left] to ({pic cs:ExC2});
\draw [->] ({pic cs:ExD3}) [bend left] to ({pic cs:ExC3});
\draw [->] ({pic cs:ExD4}) [bend left] to ({pic cs:ExC4});
\end{tikzpicture}
\end{center}

Formally, exponential on games is given as follows:
\begin{definition}[Exponential of games \cite{yamada2019game}]
\label{DefExponential}
The \emph{\bfseries exponential} $\oc A$ of a game $A$ is defined by:
\begin{itemize}

\item $M_{\oc A} \stackrel{\mathrm{df. }}{=} \{ [m]_{\Lbag^{\underline{0}} \boldsymbol{f}^+ \Rbag^{\underline{0}} \hbar \boldsymbol{e}} \mid [m]_{\boldsymbol{e}} \in M_A, \boldsymbol{f} \in \mathcal{T} \ \! \}$; 

\item $M_{\oc A}^{\mathsf{Init}} \stackrel{\mathrm{df. }}{=} \{ [m]_{\Lbag^{\underline{0}} \boldsymbol{f}^+ \Rbag^{\underline{0}} \hbar \boldsymbol{e}} \mid [m]_{\boldsymbol{e}} \in M_A^{\mathsf{Init}}, \boldsymbol{f} \in \mathcal{T} \ \! \}$; 

\item  $\lambda_{\oc A} ([m]_{\Lbag^{\underline{0}} \boldsymbol{f}^+ \Rbag^{\underline{0}} \hbar \boldsymbol{e}}) \stackrel{\mathrm{df. }}{=} \lambda_A([m]_{\boldsymbol{e}})$;

\item  $\Delta_{\oc A} ([m]_{\Lbag^{\underline{0}} \boldsymbol{f}^+ \Rbag^{\underline{0}} \hbar \boldsymbol{e}}) \stackrel{\mathrm{df. }}{=} [m']_{\Lbag^{\underline{0}} \boldsymbol{f}^+ \Rbag^{\underline{0}} \hbar \boldsymbol{e}}$, where $\Delta_A([m]_{\boldsymbol{e}}) = [m']_{\boldsymbol{e}}$;

\item $P_{\oc A} \stackrel{\mathrm{df. }}{=} \{ \boldsymbol{s} \in \mathscr{L}_{\oc A} \mid \forall \boldsymbol{f} \in \mathcal{T} . \ \! \boldsymbol{s} \upharpoonright \boldsymbol{f} \in P_A \wedge (\boldsymbol{s} \upharpoonright \boldsymbol{f} \neq \boldsymbol{\epsilon} \Rightarrow \forall \boldsymbol{g} \in \mathcal{T} . \ \! \boldsymbol{s} \upharpoonright \boldsymbol{g} \neq \boldsymbol{\epsilon} \Rightarrow \mathit{ede}(\boldsymbol{f}) \neq \mathit{ede}(\boldsymbol{g})) \ \! \}$, where $\boldsymbol{s} \upharpoonright \boldsymbol{f}$ is the j-subsequence of $\boldsymbol{s}$ that consists of moves of the form $[m]_{\Lbag^{\underline{0}} \boldsymbol{f}^+ \Rbag^{\underline{0}} \hbar \boldsymbol{e}}$ yet changed into $[m]_{\boldsymbol{e}}$.

\end{itemize}
\end{definition}

Note that a sub-expression $\boldsymbol{f}$ in an extended outer tag $[\_]_{\Lbag^{\underline{0}} \boldsymbol{f}^+ \Rbag^{\underline{0}} \hbar \boldsymbol{e}}$ that represents a natural number $i \in \mathbb{N}$, i.e., $\mathit{ede}(\Lbag^{\underline{0}} \boldsymbol{f}^+ \Rbag^{\underline{0}}) = (i)$, is unique in each $\boldsymbol{s} \in P_{\oc A}$.

Another central construction $\multimap$, called \emph{linear implication}, which originally comes from \emph{linear logic} \cite{girard1987linear}, captures the notion of \emph{linear functions}, i.e., functions that consume exactly one input to produce an output. 
A position of the linear implication $A \multimap B$ from $A$ to $B$ is almost like a position of the tensor $A \otimes B$ except the following three points:
\begin{enumerate}

\item The first element of the position must be a move of $B$; 

\item A change of $AB$-parity in the position must be made by P; 

\item Each occurrence of an initial move (called an \emph{initial occurrence}) of $A$ points to an initial occurrence of $B$.

\end{enumerate}

Thus, a typical position of the game $\boldsymbol{2} \multimap \boldsymbol{2}$ is the following:
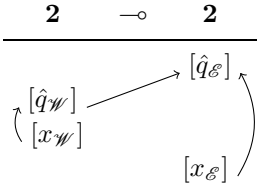
\begin{center}
\begin{tabular}{ccc}
$\boldsymbol{2}$ & $\multimap$ & $\boldsymbol{2}$ \\ \hline 
&&\tikzmark{cmultimap1} $[\hat{q}_{\mathscr{E}}]$ \tikzmark{cmultimap3} \\
\tikzmark{cmultimap2} $[\hat{q}_{\mathscr{W}}]$ \tikzmark{dmultimap1}&& \\
\tikzmark{dmultimap2} $[x_{\mathscr{W}}]$&& \\
&&$[x_{\mathscr{E}}]$ \tikzmark{dmultimap3}
\end{tabular}
\begin{tikzpicture}[overlay, remember picture, yshift=.25\baselineskip]
\draw [->] ({pic cs:dmultimap1}) to ({pic cs:cmultimap1});
\draw [->] ({pic cs:dmultimap2}) [bend left] to ({pic cs:cmultimap2});
\draw [->] ({pic cs:dmultimap3}) [bend right] to ({pic cs:cmultimap3});
\end{tikzpicture}
\end{center}
where $x \in \mathbb{B}$, which can be read as follows:
\begin{enumerate}
\item O's question $[\hat{q}_{\mathscr{E}}]$ for an output (`What is your output?');
\item P's question $[\hat{q}_{\mathscr{W}}]$ for an input (`Wait, what is your input?');
\item O's answer, say, $[\mathit{tt}_{\mathscr{W}}]$ (resp. $[\mathit{ff}_{\mathscr{W}}]$), to $[\hat{q}_{\mathscr{W}}]$ (`My input is \emph{true} (resp. \emph{false}).');
\item P's answer, say, $[\mathit{tt}_{\mathscr{E}}]$ (resp. $[\mathit{ff}_{\mathscr{E}}]$), to $[\hat{q}_{\mathscr{E}}]$ (`My output is \emph{true} (resp. \emph{false}).').
\end{enumerate}

This play is actually by the \emph{copy-cat strategy} $\mathit{cp}_{\boldsymbol{2}} : \boldsymbol{2} \multimap \boldsymbol{2}$ that always `copy-cats' the last O-move, which is the game-semantic counterpart of the identity function $\mathit{id}_{\mathbb{B}} : \mathbb{B} \rightarrow \mathbb{B}$.

Also, there is the \emph{negation strategy} $\neg : \boldsymbol{2} \multimap \boldsymbol{2}$ that plays as follows:
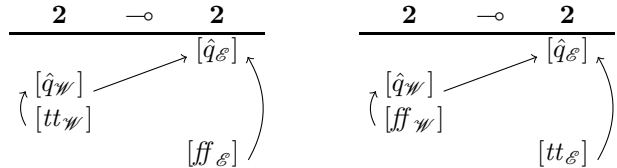
\begin{center}
\begin{tabular}{ccccccccc}
$\boldsymbol{2}$ & $\multimap$ & $\boldsymbol{2}$ &&&& $\boldsymbol{2}$ & $\multimap$ & $\boldsymbol{2}$ \\ \cline{1-3} \cline{7-9}
&& \tikzmark{cneg2} $[\hat{q}_{\mathscr{E}}]$ \tikzmark{cneg6} &&&&&& \tikzmark{cneg1} $[\hat{q}_{\mathscr{E}}]$ \tikzmark{cneg5} \\
\tikzmark{cneg4} $[\hat{q}_{\mathscr{W}}]$ \tikzmark{dneg2} &&&&&& \tikzmark{cneg3} $[\hat{q}_{\mathscr{W}}]$ \tikzmark{dneg1} && \\
\tikzmark{dneg4} $[\mathit{tt}_{\mathscr{W}}]$&&&&&& \tikzmark{dneg3} $[\mathit{ff}_{\mathscr{W}}]$&& \\
&&$[\mathit{ff}_{\mathscr{E}}]$ \tikzmark{dneg6} &&&&&&$[\mathit{tt}_{\mathscr{E}}]$ \tikzmark{dneg5}
\end{tabular}
\begin{tikzpicture}[overlay, remember picture, yshift=.25\baselineskip]
\draw [->] ({pic cs:dneg1}) to ({pic cs:cneg1});
\draw [->] ({pic cs:dneg2}) to ({pic cs:cneg2});
\draw [->] ({pic cs:dneg3}) [bend left] to ({pic cs:cneg3});
\draw [->] ({pic cs:dneg4}) [bend left] to ({pic cs:cneg4});
\draw [->] ({pic cs:dneg5}) [bend right] to ({pic cs:cneg5});
\draw [->] ({pic cs:dneg6}) [bend right] to ({pic cs:cneg6});
\end{tikzpicture}
\end{center}

Formally, linear implication between games is given as follows:
\begin{definition}[Linear implication between games \cite{yamada2019game}]
\label{DefLinearImplication}
The \emph{\bfseries linear implication} $A \multimap B$ between games $A$ and $B$ is defined by:
\begin{itemize}

\item $M_{A \multimap B} \stackrel{\mathrm{df. }}{=} \{ [(a, \mathscr{W})]_{\boldsymbol{e}} \mid [a]_{\boldsymbol{e}} \in M_{\mathcal{H}^\omega(A)} \ \! \} \cup \{ [(b, \mathscr{E})]_{\boldsymbol{f}} \mid [b]_{\boldsymbol{f}} \in M_B \ \! \}$;

\item $M_{A \multimap B}^{\mathsf{Init}} \stackrel{\mathrm{df. }}{=} \{ [(b, \mathscr{E})]_{\boldsymbol{f}} \mid [b]_{\boldsymbol{f}} \in M_B^{\mathsf{Init}} \ \! \}$;

\item $\lambda_{A \multimap B} ([(m, X)]_{\boldsymbol{e}}) \stackrel{\mathrm{df. }}{=} \begin{cases} \overline{\lambda_{\mathcal{H}^\omega(A)}}([m]_{\boldsymbol{e}}) &\text{if $X = \mathscr{W}$;} \\ \lambda_B([m]_{\boldsymbol{e}}) \ &\text{otherwise;} \end{cases}$ 

\item $\Delta_{A \multimap B}([(b, \mathscr{E})]_{\boldsymbol{f}}) \stackrel{\mathrm{df. }}{=} [(b', \mathscr{E})]_{\boldsymbol{f}}$, where $\Delta_B([b]_{\boldsymbol{f}}) = [b']_{\boldsymbol{f}}$;

\item $P_{A \multimap B} \stackrel{\mathrm{df. }}{=} \{ \boldsymbol{s} \in \mathscr{L}_{\mathcal{H}^\omega(A) \multimap B} \mid \boldsymbol{s} \upharpoonright \mathscr{W} \in P_{\mathcal{H}^\omega(A)}, \boldsymbol{s} \upharpoonright \mathscr{E} \in P_B \ \! \}$, where pointers in $\boldsymbol{s}$ from occurrences of $A$ to those of $B$ are deleted in $\boldsymbol{s} \upharpoonright \mathscr{W}$ and $\boldsymbol{s} \upharpoonright \mathscr{E}$

\end{itemize}
where $\overline{\lambda_{G}} \stackrel{\mathrm{df. }}{=} \langle \overline{\lambda_{G}^\textsf{OP}}, \lambda_{G}^\mathsf{EI} \rangle$ and $\overline{\lambda_{G}^\textsf{OP}} (x) \stackrel{\mathrm{df. }}{=} \begin{cases} \mathsf{P} &\text{if $\lambda_{G}^\textsf{OP} (x) = \mathsf{O}$;} \\ \mathsf{O} &\text{otherwise} \end{cases}$ for any game $G$.
\end{definition}

Note that the roles of O and P are interchanged in the domain $A$ of the linear implication $A \multimap B$, for which $A$  is \emph{normalized} into $\mathcal{H}^\omega(A)$ because:
\begin{itemize}

\item Conceptually, P, who is O in $A$, should not be able to `see' internal moves of $A$;

\item Technically, for the axioms \textsc{EI} and \textsc{Dum} to be preserved under $\multimap$.

\end{itemize}

Note also that indeed only P can switch between the component games $A$ and $B$ during a play of the linear implication $A \multimap B$ by the axiom \textsc{Alt}; see \cite{abramsky1997semantics}.

Let us remark that the following plays, which correspond to \emph{constant} maps $x \mapsto m$, where $m \in \mathbb{B}$ is fixed, for all $x \in \mathbb{B}$ is also possible in $\boldsymbol{2} \multimap \boldsymbol{2}$: 
\begin{center}
\begin{tabular}{ccccccccc}
$\boldsymbol{2}$ & $\multimap$ & $\boldsymbol{2}$ &&&& $\boldsymbol{2}$ & $\multimap$ & $\boldsymbol{2}$ \\ \cline{1-3} \cline{7-9}
&& $[\hat{q}_{\mathscr{E}}]$ \tikzmark{cneg16} &&&&&& $[\hat{q}_{\mathscr{E}}]$ \tikzmark{cneg15} \\
&&$[\mathit{ff}_{\mathscr{E}}]$ \tikzmark{dneg16} &&&&&&$[\mathit{tt}_{\mathscr{E}}]$ \tikzmark{dneg15}
\end{tabular}
\begin{tikzpicture}[overlay, remember picture, yshift=.25\baselineskip]
\draw [->] ({pic cs:dneg15}) [bend right] to ({pic cs:cneg15});
\draw [->] ({pic cs:dneg16}) [bend right] to ({pic cs:cneg16});
\end{tikzpicture}
\end{center}
Therefore, strictly speaking, $\multimap$ should be called \emph{affine implication}, but we have followed the standard convention to call it linear implication.

Also, copy-cats are given formally as follows:
\begin{definition}[Copy-cats \cite{yamada2019game}]
The \emph{\bfseries copy-cat} on a normalized game $A$ is the strategy $\mathit{cp}_A : A \multimap A$ given by:
\begin{equation*}
\mathit{cp}_A \stackrel{\mathrm{df. }}{=} \{ \boldsymbol{s} \in P_{A \multimap A}^\mathsf{Even} \mid \forall \boldsymbol{t} \preceq{\boldsymbol{s}}. \ \mathsf{Even}(\boldsymbol{t}) \Rightarrow \boldsymbol{t} \upharpoonright \mathscr{W} = \boldsymbol{t} \upharpoonright \mathscr{E} \ \! \}.
\end{equation*}
\end{definition}

Another construction $\&$ on games, called \emph{product}, is similar to yet simpler than tensor: A position $\boldsymbol{s}$ of the product $A \& B$ is a position of $A$ or $B$ up to tags.
It is the product in the category $\mathcal{G}$ of normalized games and strategies, e.g., there is the \emph{pairing} $\langle \sigma, \tau \rangle : C \multimap (A \& B)$ of normalized strategies $\sigma : C \multimap A$ and $\tau : C \multimap B$ that plays as $\sigma$ (resp. $\tau$) if O initiates the play by a move of $A$ (resp. $B$); see \cite{abramsky1999game}.

\begin{notation*}
Tensor $\otimes$ and product $\&$ are both left associative, while linear implication $\multimap$ is right associative. 
Exponential $\oc$ precedes any other constructions on games, and tensor $\otimes$ and product $\&$ both precede linear implication $\multimap$.
\end{notation*}

For example, typical positions of the product $\boldsymbol{2} \& \boldsymbol{2}$ are as follows:
\begin{center}
\begin{tabular}{ccccccccc}
$\boldsymbol{2}$ & $\&$ & $\boldsymbol{2}$ &&&& $\boldsymbol{2}$ & $\&$ & $\boldsymbol{2}$ \\ \cline{1-3} \cline{7-9}
\tikzmark{cproduct1} $[\hat{q}_{\mathscr{W}}]$&&&&&&&&\tikzmark{cproduct2} $[\hat{q}_{\mathscr{E}}]$ \\
\tikzmark{dproduct1} $[\mathit{tt}_{\mathscr{W}}]$&&&&&&&&\tikzmark{dproduct2} $[\mathit{ff}_{\mathscr{E}}]$ 
\end{tabular}
\begin{tikzpicture}[overlay, remember picture, yshift=.25\baselineskip]
\draw [->] ({pic cs:dproduct1}) [bend left] to ({pic cs:cproduct1});
\draw [->] ({pic cs:dproduct2}) [bend left] to ({pic cs:cproduct2});
\end{tikzpicture}
\end{center}

Formally, product of games is given as follows: 
\begin{definition}[Product of games \cite{yamada2019game}]
The \emph{\bfseries product} $A \& B$ of games $A$ and $B$ is given by:
\begin{itemize}

\item $M_{A \& B} \stackrel{\mathrm{df. }}{=} \{ [(a, \mathscr{W})]_{\boldsymbol{e}} \mid [a]_{\boldsymbol{e}} \in M_A \ \! \} \cup \{ [(b, \mathscr{E})]_{\boldsymbol{f}} \mid [b]_{\boldsymbol{f}} \in M_B \ \! \}$;

\item $M_{A \& B}^{\mathsf{Init}} \stackrel{\mathrm{df. }}{=} \{ [(a, \mathscr{W})]_{\boldsymbol{e}} \mid [a]_{\boldsymbol{e}} \in M_A^{\mathsf{Init}} \ \! \} \cup \{ [(b, \mathscr{E})]_{\boldsymbol{f}} \mid [b]_{\boldsymbol{f}} \in M_B^{\mathsf{Init}} \ \! \}$;

\item $\lambda_{A \& B} ([(m, X)]_{\boldsymbol{e}}) \stackrel{\mathrm{df. }}{=} \begin{cases} \lambda_A([m]_{\boldsymbol{e}}) &\text{if $X = \mathscr{W}$;} \\ \lambda_B([m]_{\boldsymbol{e}}) &\text{otherwise;} \end{cases}$

\item $\Delta_{A \& B} ([(m, X)]_{\boldsymbol{e}}) \stackrel{\mathrm{df. }}{=} \begin{cases} [(m', \mathscr{W})]_{\boldsymbol{e}} &\text{if $X = \mathscr{W}$, where $\Delta_A([m]_{\boldsymbol{e}}) = [m']_{\boldsymbol{e}}$;} \\ [(m'', \mathscr{E})]_{\boldsymbol{e}} &\text{otherwise, where $\Delta_B([m]_{\boldsymbol{e}}) = [m'']_{\boldsymbol{e}}$;} \end{cases}$

\item $P_{A \& B} \stackrel{\mathrm{df. }}{=} \{ \boldsymbol{s} \in \mathscr{L}_{A \& B} \mid (\boldsymbol{s} \upharpoonright \mathscr{W} \in P_A \wedge \boldsymbol{s} \upharpoonright \mathscr{E} = \boldsymbol{\epsilon}) \vee (\boldsymbol{s} \upharpoonright \mathscr{W} = \boldsymbol{\epsilon} \wedge \boldsymbol{s} \upharpoonright \mathscr{E} \in P_B) \ \! \}$.

\end{itemize}
\end{definition}

As another example, the pairing $\langle \mathit{cp}_{\boldsymbol{2}}, \neg \rangle : \boldsymbol{2} \multimap \boldsymbol{2} \& \boldsymbol{2}$ plays as:
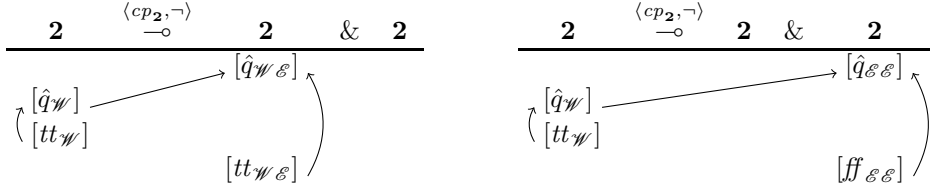
\begin{center}
\begin{tabular}{ccccccccccccc}
$\boldsymbol{2}$ & $\stackrel{\langle \mathit{cp}_{\boldsymbol{2}}, \neg \rangle}{\multimap}$ & $\boldsymbol{2}$ & $\&$ & $\boldsymbol{2}$ &&&& $\boldsymbol{2}$ & $\stackrel{\langle \mathit{cp}_{\boldsymbol{2}}, \neg \rangle}{\multimap}$ & $\boldsymbol{2}$ & $\&$ & $\boldsymbol{2}$  \\
\cline{1-5} \cline{9-13}
&& \tikzmark{cpair64} $[\hat{q}_{\mathscr{W} \mathscr{E}}]$ \tikzmark{cpair66} && &&&& &&&& \tikzmark{cpair61} $[\hat{q}_{\mathscr{E} \mathscr{E}}]$ \tikzmark{cpair63} \\
\tikzmark{cpair65} $[\hat{q}_{\mathscr{W}}]$ \tikzmark{dpair64} &&&&&& && \tikzmark{cpair62} $[\hat{q}_{\mathscr{W}}]$ \tikzmark{dpair61} &&&& \\
\tikzmark{dpair65} $[\mathit{tt}_{\mathscr{W}}]$&&&& &&&& \tikzmark{dpair62} $[\mathit{tt}_{\mathscr{W}}]$ &&&& \\
&&$[\mathit{tt}_{\mathscr{W} \mathscr{E}}]$ \tikzmark{dpair66} && && &&&&&& $[\mathit{ff}_{\mathscr{E} \mathscr{E}}]$ \tikzmark{dpair63}
\end{tabular}
\begin{tikzpicture}[overlay, remember picture, yshift=.25\baselineskip]
\draw [->] ({pic cs:dpair61}) to ({pic cs:cpair61});
\draw [->] ({pic cs:dpair62}) [bend left] to ({pic cs:cpair62});
\draw [->] ({pic cs:dpair63}) [bend right] to ({pic cs:cpair63});
\draw [->] ({pic cs:dpair64}) to ({pic cs:cpair64});
\draw [->] ({pic cs:dpair65}) [bend left] to ({pic cs:cpair65});
\draw [->] ({pic cs:dpair66}) [bend right] to ({pic cs:cpair66});
\end{tikzpicture}
\end{center}

Formally, pairing of strategies is given as follows:
\begin{definition}[Pairing of strategies \cite{yamada2019game}]
\label{Pairing}
Given normalized games $A$, $B$ and $C$, and normalized strategies $\sigma : C \multimap A$ and $\tau : C \multimap B$, the \emph{\bfseries pairing} $\langle \sigma, \tau \rangle : C \multimap A \& B$ of $\sigma$ and $\tau$ is defined by: 
\begin{align*}
\langle \sigma, \tau \rangle \stackrel{\mathrm{df. }}{=} & \ \! \{ \boldsymbol{s} \in \mathscr{L}_{C \multimap A \& B} \mid \boldsymbol{s} \upharpoonright (\mathscr{W} \multimap \mathscr{W} \mathscr{E}) \in \sigma, \boldsymbol{s} \upharpoonright (\mathscr{W} \multimap \mathscr{E} \mathscr{E}) = \boldsymbol{\epsilon} \ \! \} \\ &\cup \ \! \{ \boldsymbol{s} \in \mathscr{L}_{C \multimap A \& B} \mid \boldsymbol{s} \upharpoonright (\mathscr{W} \multimap \mathscr{E} \mathscr{E}) \in \tau, \boldsymbol{s} \upharpoonright (\mathscr{W} \multimap \mathscr{W} \mathscr{E}) = \boldsymbol{\epsilon} \ \! \}
\end{align*}
where $\boldsymbol{s} \upharpoonright (\mathscr{W} \multimap \mathscr{W} \mathscr{E})$ (resp. $\boldsymbol{s} \upharpoonright (\mathscr{W} \multimap \mathscr{E} \mathscr{E})$) is the j-subsequence of $\boldsymbol{s}$ that consists of moves $[(c, \mathscr{W})]_{\boldsymbol{e}}$ or $[((a, \mathscr{W}), \mathscr{E})]_{\boldsymbol{f}}$ with $[a] \in M_A$ (resp. or $[((b, \mathscr{E}), \mathscr{E})]_{\boldsymbol{g}}$ with $[b] \in M_B$) yet the latter changed into $[(a, \mathscr{E})]_{\boldsymbol{f}}$ (resp. $[(b, \mathscr{E})]_{\boldsymbol{g}}$).
\end{definition}

The constructions $\otimes$, $!$, $\multimap$ and $\&$ originally come from the corresponding ones in linear logic; see \cite{abramsky1994games}. 
Therefore, the usual \emph{implication} (or the \emph{function space}) $\Rightarrow$ is recovered by \emph{Girard translation} \cite{girard1987linear}: $A \Rightarrow B \stackrel{\mathrm{df. }}{=} \oc A \multimap B$.

Girard translation makes explicit the point that some functions need to refer to an input \emph{more than once} to produce an output, i.e., there are non-linear functions.
For instance, the strategy on $(\boldsymbol{2} \Rightarrow \boldsymbol{2}) \Rightarrow \boldsymbol{2}$ that computes the disjunction $f(\mathit{true}) \vee f(\mathit{false})$ for a given boolean function $f : \mathbb{B} \Rightarrow \mathbb{B}$ plays as: 
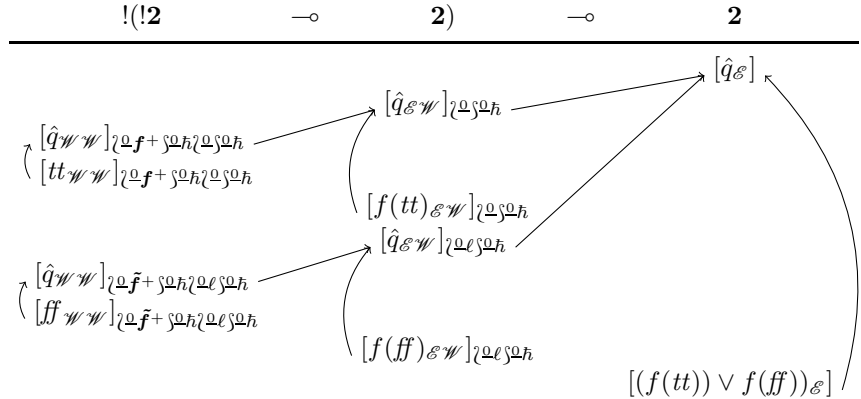
\begin{center}
\begin{tabular}{ccccc}
$\oc (\oc \boldsymbol{2}$ & $\multimap$ & $\boldsymbol{2})$ & $\multimap$ & $\boldsymbol{2}$ \\ \hline 
&&&& \tikzmark{chigher11} $[\hat{q}_{\mathscr{E}}]$ \tikzmark{chigher19} \\
&&\tikzmark{chigher12} $[\hat{q}_{\mathscr{E} \mathscr{W}}]_{\Lbag^{\underline{0}} \Rbag^{\underline{0}} \hbar}$ \tikzmark{dhigher11} && \\
\tikzmark{chigher13} $[\hat{q}_{\mathscr{W} \mathscr{W}}]_{\Lbag^{\underline{0}} \boldsymbol{f}^+ \Rbag^{\underline{0}} \hbar \Lbag^{\underline{0}} \Rbag^{\underline{0}} \hbar}$ \tikzmark{dhigher12} && && \\
\tikzmark{dhigher13} $[\mathit{tt}_{\mathscr{W} \mathscr{W}}]_{\Lbag^{\underline{0}} \boldsymbol{f}^+ \Rbag^{\underline{0}} \hbar \Lbag^{\underline{0}} \Rbag^{\underline{0}} \hbar}$ && && \\
&& \tikzmark{dhigher14} $[f(\mathit{tt})_{\mathscr{E} \mathscr{W}}]_{\Lbag^{\underline{0}} \Rbag^{\underline{0}} \hbar}$&& \\
&&\tikzmark{chigher18} $[\hat{q}_{\mathscr{E} \mathscr{W}}]_{\Lbag^{\underline{0}} \ell \Rbag^{\underline{0}} \hbar}$ \tikzmark{dhigher15} && \\
\tikzmark{chigher17} $[\hat{q}_{\mathscr{W} \mathscr{W}}]_{\Lbag^{\underline{0}} \boldsymbol{\tilde{f}}^+ \Rbag^{\underline{0}} \hbar \Lbag^{\underline{0}} \ell \Rbag^{\underline{0}} \hbar}$ \tikzmark{dhigher16} && && \\
\tikzmark{dhigher17} $[\mathit{ff}_{\mathscr{W} \mathscr{W}}]_{\Lbag^{\underline{0}} \boldsymbol{\tilde{f}}^+ \Rbag^{\underline{0}} \hbar \Lbag^{\underline{0}} \ell \Rbag^{\underline{0}} \hbar}$&& && \\
&& \tikzmark{dhigher18} $[f(\mathit{ff})_{\mathscr{E} \mathscr{W}}]_{\Lbag^{\underline{0}} \ell \Rbag^{\underline{0}} \hbar}$ && \\
&&&& $[(f(\mathit{tt})) \vee f(\mathit{ff}))_{\mathscr{E}}]$ \tikzmark{dhigher19}
\end{tabular}
\begin{tikzpicture}[overlay, remember picture, yshift=.25\baselineskip]
\draw [->] ({pic cs:dhigher11}) to ({pic cs:chigher11});
\draw [->] ({pic cs:dhigher12}) to ({pic cs:chigher12});
\draw [->] ({pic cs:dhigher13}) [bend left] to ({pic cs:chigher13});
\draw [->] ({pic cs:dhigher14}) [bend left] to ({pic cs:chigher12});
\draw [->] ({pic cs:dhigher15}) to ({pic cs:chigher11});
\draw [->] ({pic cs:dhigher16}) to ({pic cs:chigher18});
\draw [->] ({pic cs:dhigher17}) [bend left] to ({pic cs:chigher17});
\draw [->] ({pic cs:dhigher18}) [bend left] to ({pic cs:chigher18});
\draw [->] ({pic cs:dhigher19}) [bend right] to ({pic cs:chigher19});
\end{tikzpicture}
\end{center}
where $\boldsymbol{\epsilon}, \ell \in \mathcal{T}$ occurring in the intermediate $\oc \boldsymbol{2}$ are arbitrarily chosen by P, i.e., any $\boldsymbol{g}, \boldsymbol{g'} \in \mathcal{T}$ work as long as $\mathit{ede}(\boldsymbol{g}) \neq \mathit{ede}(\boldsymbol{g'})$, and $\boldsymbol{f}, \boldsymbol{\tilde{f}} \in \mathcal{T}$ occurring in the leftmost $\oc \oc \boldsymbol{2}$ are chosen by O.
In this play, P asks O \emph{twice} about an input strategy $\boldsymbol{2} \Rightarrow \boldsymbol{2}$.
Clearly, such a play is not possible on the linear implication $(\boldsymbol{2} \multimap \boldsymbol{2}) \multimap \boldsymbol{2}$ or $(\boldsymbol{2} \Rightarrow \boldsymbol{2}) \multimap \boldsymbol{2}$.

Next, recall that any normalized strategy $\phi : \oc A \multimap B$ gives its \emph{promotion} $\phi^\dagger : \oc A \multimap \oc B$ such that if $\phi$ plays, for instance, as:
\begin{center}
\begin{tabular}{ccc}
$\oc A$ & $\multimap$ & $B$ \\ \hline 
&&\tikzmark{cpromotion1} $[b^{(1)}_{\mathscr{E}}]$ \tikzmark{cpromotion3} \\
\tikzmark{cpromotion2} $[a^{(1)}_{\mathscr{W}}]_{\Lbag^{\underline{0}} \boldsymbol{f}^+ \Rbag^{\underline{0}} \hbar}$ \tikzmark{dpromotion1}&& \\
\tikzmark{dpromotion2} $[a^{(2)}_{\mathscr{W}}]_{\Lbag^{\underline{0}} \boldsymbol{f}^+ \Rbag^{\underline{0}} \hbar}$ && \\
&&$[b^{(2)}_{\mathscr{E}}]$ \tikzmark{dpromotion3}
\end{tabular}
\begin{tikzpicture}[overlay, remember picture, yshift=.25\baselineskip]
\draw [->] ({pic cs:dpromotion1}) to ({pic cs:cpromotion1});
\draw [->] ({pic cs:dpromotion2}) [bend left] to ({pic cs:cpromotion2});
\draw [->] ({pic cs:dpromotion3}) [bend right] to ({pic cs:cpromotion3});
\end{tikzpicture}
\end{center}
then $\phi^\dagger$ plays as:
\begin{center}
\begin{tabular}{ccc}
$\oc A$ & $\multimap$ & $\oc B$ \\ \hline 
&&\tikzmark{cpromotion11} $[b^{(1)}_{\mathscr{E}}]_{\Lbag^{\underline{0}} \boldsymbol{e}^+ \Rbag^{\underline{0}} \hbar}$ \tikzmark{cpromotion13} \\
\tikzmark{cpromotion12} $[a^{(1)}_{\mathscr{W}}]_{\Lbag^{\underline{0}} \Lbag^{\underline{1}} \boldsymbol{e}^{++} \Rbag^{\underline{1}} \hbar \Lbag^{\underline{1}} \boldsymbol{f}^{++} \Rbag^{\underline{1}} \Rbag^{\underline{0}} \hbar}$ \tikzmark{dpromotion11} && \\
\tikzmark{dpromotion12} $[a^{(2)}_{\mathscr{W}}]_{\Lbag^{\underline{0}} \Lbag^{\underline{1}} \boldsymbol{e}^{++} \Rbag^{\underline{1}} \hbar \Lbag^{\underline{1}} \boldsymbol{f}^{++} \Rbag^{\underline{1}} \Rbag^{\underline{0}} \hbar}$ && \\
&& $[b^{(2)}_{\mathscr{E}}]_{\Lbag \boldsymbol{e} \Rbag \hbar}$ \tikzmark{dpromotion13} \\
&& \tikzmark{cpromotion14} $[b^{(1)}_{\mathscr{E}}]_{\Lbag^{\underline{0}} \boldsymbol{e'}^{+} \Rbag^{\underline{0}} \hbar}$ \tikzmark{cpromotion16} \\
\tikzmark{cpromotion15} $[a^{(1)}_{\mathscr{W}}]_{\Lbag^{\underline{0}} \Lbag^{\underline{1}} \boldsymbol{e'}^{++} \Rbag^{\underline{1}} \hbar \Lbag^{\underline{1}} \boldsymbol{f}^{++} \Rbag^{\underline{1}} \Rbag^{\underline{0}} \hbar}$ \tikzmark{dpromotion14} && \\
\tikzmark{dpromotion15} $[a^{(2)}_{\mathscr{W}}]_{\Lbag^{\underline{0}} \Lbag^{\underline{1}} \boldsymbol{e'}^{++} \Rbag^{\underline{1}} \hbar \Lbag^{\underline{1}} \boldsymbol{f}^{++} \Rbag^{\underline{1}} \Rbag^{\underline{0}} \hbar}$ && \\
&& $[b^{(2)}_{\mathscr{E}}]_{\Lbag^{\underline{0}} \boldsymbol{e'}^+ \Rbag^{\underline{0}} \hbar}$ \tikzmark{dpromotion16}
\end{tabular}
\begin{tikzpicture}[overlay, remember picture, yshift=.25\baselineskip]
\draw [->] ({pic cs:dpromotion11}) to ({pic cs:cpromotion11});
\draw [->] ({pic cs:dpromotion12}) [bend left] to ({pic cs:cpromotion12});
\draw [->] ({pic cs:dpromotion13}) [bend right] to ({pic cs:cpromotion13});
\draw [->] ({pic cs:dpromotion14}) to ({pic cs:cpromotion14});
\draw [->] ({pic cs:dpromotion15}) [bend left] to ({pic cs:cpromotion15});
\draw [->] ({pic cs:dpromotion16}) [bend right] to ({pic cs:cpromotion16});
\end{tikzpicture}
\end{center}
where $\boldsymbol{e}, \boldsymbol{e'} \in \mathcal{T}$ are chosen by O, and $\boldsymbol{f} \in \mathcal{T}$ by P.
That is, $\phi^\dagger$ plays as $\phi$ for each \emph{thread} in a position of $\oc A \multimap \oc B$ that corresponds to a position of $\oc A \multimap B$ (n.b., a thread is a certain kind of a j-subsequence; see \cite{abramsky1999game,yamada2019game} for the precise definition). 

Formally, promotion on strategies is given as follows:
\begin{definition}[Promotion on strategies \cite{yamada2019game}]
\label{DefPromotionOfStrategies}
Given normalized games $A$ and $B$, and a normalized strategy $\phi : \oc A \multimap B$, the \emph{\bfseries promotion} $\phi^{\dagger} : \oc A \multimap \oc B$ of $\phi$ is defined by: 
\begin{equation*}
\phi^{\dagger} \stackrel{\mathrm{df. }}{=} \{ \boldsymbol{s} \in \mathscr{L}_{\oc A \multimap \oc B} \mid \forall \boldsymbol{e} \in \mathcal{T} . \ \! \boldsymbol{s} \upharpoonright \boldsymbol{e} \in \phi \ \! \}
\end{equation*}
where $\boldsymbol{s} \upharpoonright \boldsymbol{e}$ is the j-subsequence of $\boldsymbol{s}$ that consists of moves of the form $[(b, \mathscr{E})]_{\Lbag^{\underline{0}} \boldsymbol{e}^+ \Rbag^{\underline{0}} \hbar \boldsymbol{e'}}$ with $[b]_{\boldsymbol{e'}} \in M_B$ or $[(a, \mathscr{W})]_{\Lbag^{\underline{0}} \Lbag^{\underline{1}} \boldsymbol{e}^{++} \Rbag^{\underline{1}} \hbar \Lbag^{\underline{1}} \boldsymbol{f}^{++} \Rbag^{\underline{1}} \Rbag^{\underline{0}} \hbar \boldsymbol{f'}}$ with $[a]_{\boldsymbol{f'}} \in M_A$, which are respectively changed into $[(b, \mathscr{E})]_{\boldsymbol{e'}}$ or $[(a, \mathscr{W})]_{\Lbag^{\underline{0}} \boldsymbol{f}^+ \Rbag^{\underline{0}} \hbar \boldsymbol{f'}}$.
\end{definition}

Constructions introduced so far preserve normalization of games and strategies; they are in fact employed for conventional games and strategies \cite{abramsky1999game}.
This point no longer holds as soon as we introduce \emph{concatenations} $\ddagger$ on games and strategies, which are first introduced in \cite{yamada2016dynamic}.
The idea is to decompose the standard \emph{composition} $\phi ; \psi : A \multimap C$ of normalized strategies $\phi : A \multimap B$ and $\psi : B \multimap C$, where $A$, $B$ and $C$ are normalized games, into \emph{concatenation} $\phi \ddagger \psi : (A \multimap B) \ddagger (B \multimap C)$ plus \emph{hiding} $\mathcal{H}^\omega$; in fact, the work \cite{yamada2016dynamic} shows $\mathcal{H}^\omega((A \multimap B) \ddagger (B \multimap C)) \trianglelefteqslant A \multimap C$ and $\mathcal{H}^\omega(\phi \ddagger \psi) = \phi ; \psi$, whence $\phi ; \psi : A \multimap C$ by Proposition~\ref{PropStrategiesOnSubgames} and Theorem~\ref{ThmHidingTheorem}.

In addition, the work \cite{yamada2016dynamic} generalizes this phenomenon by concatenation on (not necessarily normalized) games $J$ and $K$ such that $\mathcal{H}^\omega(J) \trianglelefteqslant A \multimap B$ and $\mathcal{H}^\omega(K) \trianglelefteqslant B \multimap C$ in such a way that satisfies $\mathcal{H}^\omega(J \ddagger K) \trianglelefteqslant A \multimap C$.
Also, it defines concatenation $\iota \ddagger \kappa : J \ddagger K$ on (not necessarily normalized) strategies $\iota : J$ and $\kappa : K$.
This enables us to apply concatenations on games and strategies in an \emph{iterated} manner, e.g., we may obtain $(\phi \ddagger \psi) \ddagger \mu : ((A \multimap B) \ddagger (B \multimap C)) \ddagger C \multimap D$ for any normalized game $D$ and strategy $\mu : C \multimap D$.

Roughly, a position $\boldsymbol{s}$ of the concatenation $J \ddagger K$ is an interleaving mixture of a positions $\boldsymbol{t}$ of $J$ and a position $\boldsymbol{u}$ of $K$, which are `synchronized' to each other via moves of $B$: Each O-move of $B$ occurring in $\boldsymbol{s}$ is a mere `dummy' of the last P-move. That is, a position of $J \ddagger K$ begins with an initial occurrence of $\boldsymbol{u}$, and then a play of $K$ proceeds until a P-move of $B$ occurs; when a P-move of $B$ occurs in $K$, then it is copied and performed as an O-move in $J$, and then a play of $J$ proceeds until a P-move of $B$ occurs; if a P-move of $B$ occurs in $J$, then it is copied and performed as an O-move in $K$, and then a play of $K$ resumes and proceeds until a P-move of $B$ occurs, and so on.
In addition, moves of $B$ get \emph{internal} in $J \ddagger K$, while other external moves remain external.

Specifically, the previous work \cite{yamada2019game} formalizes the `tags' for the concatenation $J \ddagger K$ as follows:
\begin{itemize}

\item It does not change moves of $A$ or $C$, i.e., $[a_{\mathscr{W}}]_{\boldsymbol{e}} \in M_J^{\mathsf{Ext}}$ or $[c_{\mathscr{E}}]_{\boldsymbol{f}} \in M_K^{\mathsf{Ext}}$;

\item It changes moves of $B$ in $J$, i.e., $[b_{\mathscr{E}}]_{\boldsymbol{g}} \in M_J^{\mathsf{Ext}}$, into $[b_{\mathscr{ES}}]_{\boldsymbol{g}}$;

\item It changes moves of $B$ in $K$, i.e., $[b_{\mathscr{W}}]_{\boldsymbol{g}} \in M_K^{\mathsf{Ext}}$, into $[b_{\mathscr{WN}}]_{\boldsymbol{g}}$;

\item It changes internal moves $[m]_{\boldsymbol{l}}$ of $J$ into $[(m_{\mathscr{S}}]_{\boldsymbol{l}}$;

\item It changes internal moves $[n]_{\boldsymbol{r}}$ of $K$ into $[n_{\mathscr{N}}]_{\boldsymbol{r}}$.

\end{itemize}
Of course, this implementation of `tags' for $J \ddagger K$ is far from canonical, but the point is that it certainly achieves the required subgame relation $\mathcal{H}^\omega(J \ddagger K) \trianglelefteqslant A \multimap C$.

On the other hand, the concatenation $\iota \ddagger \kappa : J \ddagger K$ plays as $\iota$ if the last O-moves is of $J$, and as $\kappa$ otherwise.

For instance, the concatenation $\neg \ddagger \neg : (\boldsymbol{2} \multimap \boldsymbol{2}) \ddagger (\boldsymbol{2} \multimap \boldsymbol{2})$ of negation $\neg : \boldsymbol{2} \multimap \boldsymbol{2}$ with itself plays as follows:
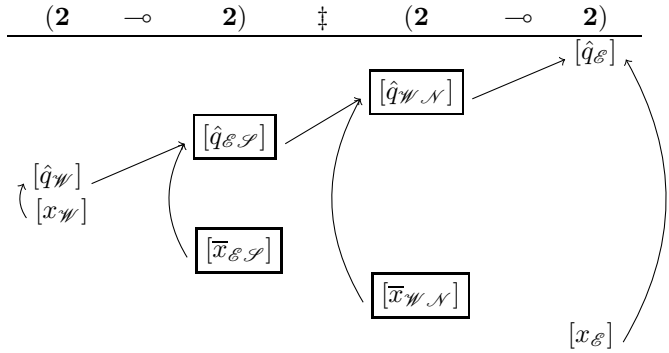
\begin{center}
\begin{tabular}{ccccccc}
$(\boldsymbol{2}$ & $\multimap$ & $\boldsymbol{2})$ & $\ddagger$ & $(\boldsymbol{2}$ & $\multimap$ & $\boldsymbol{2})$ \\ \cline{1-7} 
&&&&&&\tikzmark{cCon1} $[\hat{q}_{\mathscr{E}}]$ \tikzmark{cCon7} \\
&&&&\tikzmark{cCon2} \fbox{$[\hat{q}_{\mathscr{WN}}]$} \tikzmark{dCon1}&& \\
&&\tikzmark{cCon3} \fbox{$[\hat{q}_{\mathscr{ES}}]$} \tikzmark{dCon2} &&&& \\
\tikzmark{cCon4} $[\hat{q}_{\mathscr{W}}]$ \tikzmark{dCon3} &&&&&& \\
\tikzmark{dCon4} $[x_{\mathscr{W}}]$ &&&&&& \\
&& \tikzmark{dCon5} \fbox{$[\overline{x}_{\mathscr{ES}}]$} &&&& \\
&&&& \tikzmark{dCon6} \fbox{$[\overline{x}_{\mathscr{WN}}]$} && \\
&&&&&& $[x_{\mathscr{E}}]$ \tikzmark{dCon7}
\end{tabular}
\begin{tikzpicture}[overlay, remember picture, yshift=.25\baselineskip]
\draw [->] ({pic cs:dCon1}) to ({pic cs:cCon1});
\draw [->] ({pic cs:dCon2}) to ({pic cs:cCon2});
\draw [->] ({pic cs:dCon3}) to ({pic cs:cCon3});
\draw [->] ({pic cs:dCon4}) [bend left] to ({pic cs:cCon4});
\draw [->] ({pic cs:dCon5}) [bend left] to ({pic cs:cCon3});
\draw [->] ({pic cs:dCon6}) [bend left] to ({pic cs:cCon2});
\draw [->] ({pic cs:dCon7}) [bend right] to ({pic cs:cCon7});
\end{tikzpicture}
\end{center}
where $x \in \mathbb{B}$ and $\overline{x} \stackrel{\mathrm{df. }}{=} \begin{cases} \mathit{ff} &\text{if $x = \mathit{tt}$;} \\ \mathit{tt} &\text{otherwise.} \end{cases}$
Moves with the inner tag $(\_)_{{\mathscr{WN}}}$ or $(\_)_{{\mathscr{ES}}}$ are \emph{internal}, for which we have marked them by square boxes just for clarity.
In the above play, the two copies of $\neg$ communicate to each other by `synchronizing' the codomain $\boldsymbol{2}$ of the left $\neg$ and the domain $\boldsymbol{2}$ of the right $\neg$, for which P also plays the role of O in these intermediate games by `copying' her last P-moves.
This phenomenon is what the axiom \textsc{Dum} (Definition~\ref{DefLegalPositions}) captures abstractly.

\begin{remark*}
Crucially, the game-semantic PCF-computations \cite{yamada2019game} employ \emph{dynamic} games and strategies \cite{yamada2016dynamic} for composition of dynamic strategies is concatenation, which, unlike composition of conventional strategies, \emph{keeps internal moves}. 
The point is that internal moves represent \emph{step-by-step} processes in computation, or `internal calculation' by P; thus, they enable the intrinsic, non-inductive, non-axiomatic definition of `effective computability' of dynamic strategies \cite{yamada2019game}. 
\end{remark*}

The formal definitions of concatenations on games and strategies are slightly involved, and hence, we leave their details to Appendices~\ref{DefConcatenationOfGames} and \ref{DefConcatenationAndCompositionOfStrategies}.

It is now appropriate to recall that the category $\mathcal{G}$ of conventional games and strategies has normalized games as objects and strategies $\phi : A \Rightarrow B$ as morphisms from $A$ to $B$; the composition of $\phi : A \rightarrow B$ and $\psi : B \rightarrow C$ in $\mathcal{G}$ is given by $\phi^\dagger ; \psi : A \rightarrow C$, and the identity on $A$ in $\mathcal{G}$ is the dereliction $\mathit{der}_A : A \rightarrow A$.
The bicategory $\mathcal{DG}$ \cite{yamada2016dynamic} generalizes the category $\mathcal{G}$: Objects of $\mathcal{DG}$ are normalized games, and 1-cells $A \rightarrow B$ are strategies $\phi : G$ such that $\mathcal{H}^\omega(G) \trianglelefteqslant A \Rightarrow B$; horizontal composition of 1-cells in $\mathcal{DG}$ is given by $(\_)^\dagger \ddagger (\_)$, and horizontal identities in $\mathcal{DG}$ are derelictions (where 2-cells in $\mathcal{DG}$ are the equivalence of 1-cells up to $\mathcal{H}^\omega$).
 
Accordingly, pairing and promotion of strategies in $\mathcal{DG}$ are generalized, for which product and exponential of games are also generalized in a straightforward manner to \emph{pairing and promotion of games} as follows.
Given games $L$ and $R$ such that $\mathcal{H}^\omega(L) \trianglelefteqslant C \multimap A$ and $\mathcal{H}^\omega(R) \trianglelefteqslant C \multimap B$ for normalized games $A$, $B$ and $C$, there is the \emph{pairing} $\langle L, R \rangle$ of $L$ and $R$ such that a position of $\langle L, R \rangle$ is a position of $L$ or $R$ up to inner tags. 
Specifically, to establish the required subgame relation $\mathcal{H}^\omega(\langle L, R \rangle) \trianglelefteqslant C \multimap A \& B$, the previous work \cite{yamada2019game} formalizes the `tags' for $\langle L, R \rangle$ by:
\begin{itemize}

\item Keeping external moves of the form $[c_{\mathscr{W}}]_{\boldsymbol{e}}$ of $L$ or $R$ unchanged, where $[c]_{\boldsymbol{e}}$ must be a move of $C$ by the definition of $\multimap$ (Definition~\ref{DefLinearImplication}); 

\item Changing external moves of the form $[a_{\mathscr{E}}]_{\boldsymbol{f}}$ of $L$, where $[a]_{\boldsymbol{f}}$ must be a move of $A$ by the definition of $\multimap$, into $[a_{\mathscr{WE}}]_{\boldsymbol{f}}$; 

\item Changing external moves of the form $[b_{\mathscr{E}}]_{\boldsymbol{g}}$ of $R$, where $[b]_{\boldsymbol{g}}$ must be a move of $B$ by the definition of $\multimap$, into $[b_{\mathscr{EE}}]_{\boldsymbol{g}}$; 

\item Changing internal moves $[l]_{\boldsymbol{h}}$ of $L$ into $[l_{\mathscr{S}}]_{\boldsymbol{h}}$; 

\item Changing internal moves $[r]_{\boldsymbol{k}}$ of $R$ into $[r_{\mathscr{N}}]_{\boldsymbol{k}}$.

\end{itemize} 

It is then easy to see that we may form the pairing $\langle \sigma, \tau \rangle : \langle L, R \rangle$ of strategies $\sigma : L$ and $\tau : R$ that plays as $\sigma$ if O begins the play by an initial move of $L$, and by $\tau$ otherwise. 
See Appendices~\ref{DefPairingOfGames} and \ref{DefGeneralizedPairingOfStrategies} for their formal definitions. 

Next, given a game $G$ such that $\mathcal{H}^\omega(G) \trianglelefteqslant \oc A \multimap B$ for some normalized games $A$ and $B$, there is the \emph{promotion} $G^\dagger$ of $G$ that coincides with the exponential $\oc G$ of $G$ up to tags. Specifically, to establish the required subgame relation $\mathcal{H}^\omega(G^\dagger) \trianglelefteqslant \oc A \multimap \oc B$, the previous work \cite{yamada2019game} formalizes the `tags' for $G^\dagger$ as follows:
\begin{itemize}

\item It duplicates moves of $G$ coming from $\oc A$, i.e., ones of the form $[a_{\mathscr{W}}]_{\Lbag^{\underline{0}} \boldsymbol{f}^+ \Rbag^{\underline{0}} \hbar \boldsymbol{e}}$, as $[a_{\mathscr{W}}]_{\Lbag^{\underline{0}} \Lbag^{\underline{1}} \boldsymbol{g}^{++} \Rbag^{\underline{1}} \hbar \Lbag^{\underline{1}} \boldsymbol{f}^{++} \Rbag^{\underline{1}} \Rbag^{\underline{0}} \hbar \boldsymbol{e}}$ for each $\boldsymbol{g} \in \mathcal{T}$;

\item It duplicates moves of $G$ coming from $B$, i.e., ones of the form $[b_{\mathscr{E}}]_{\boldsymbol{e}}$, as $[b_{\mathscr{E}}]_{\Lbag^{\underline{0}} \boldsymbol{g}^+ \Rbag^{\underline{0}} \hbar \boldsymbol{e}}$ for each $\boldsymbol{g} \in \mathcal{T}$;

\item It duplicates internal moves $[m]_{\boldsymbol{e}}$ of $G$ as $[m_{\mathscr{S}}]_{\Lbag^{\underline{0}} \boldsymbol{g}^+ \Rbag^{\underline{0}} \hbar \boldsymbol{e}}$ for each $\boldsymbol{g} \in \mathcal{T}$.

\end{itemize}

Given a strategy $\phi : G$, there is the \emph{promotion} $\phi^\dagger : G^\dagger$ that plays as $\phi$ for each copy of $G$ in $G^\dagger$, i.e., the last O-move of each odd-length position of $G^\dagger$ determines the currently active thread of the position, that is a position of $G$ up to tags, and $\phi^\dagger$ plays as $\phi$ on that thread; see Appendix \ref{DefPromotionOfGames} and Definition~\ref{DefPromotionOfStrategies} for the details. 

Finally, let us recall the trivial \emph{curryings} $\Lambda$ on games and strategies.
Roughly, they generalize the maps $A \otimes B \multimap C \mapsto A \multimap (B \multimap C)$ and $(\phi : A \otimes B \multimap C) \mapsto (\Lambda(\phi) : A \multimap (B \multimap C))$, where $A$, $B$ and $C$ are normalized games.
Since strategies in $\mathcal{DG}$ may be non-normalized, curryings need to be generalized as in the case of pairing and promotion, but it is just straightforward; it suffices to replace inner tags appropriately.
See Appendices~\ref{DefCurryingOnGames} and \ref{DefCurryingOnStrategies} for the details.

\subsection{Games and strategies for PCF-computation}
\label{GamesAndStrategiesForPCF}
We are now ready to recall the games and strategies that interpret the prototypical functional programming language \emph{PCF} \cite{scott1993type,plotkin1977lcf}.

\begin{notation*}
We often indicate the form of tags of moves $[m_{X_1 X_2 \dots X_k}]_{\boldsymbol{e}}$ of a game $G$ informally by $[G_{X_1 X_2 \dots X_k}]_{\boldsymbol{e}}$, especially when the tags are involved. 
\end{notation*}

The first one is the \emph{zero strategy} $\mathit{zero}_A : A \Rightarrow \mathcal{N}$ on any normalized game $A$:
\begin{definition}[Zero strategies \cite{yamada2019game}]
Given a normalized game $A$, the \emph{\bfseries zero strategy} on a normalized game $A$ is the strategy $\mathit{zero}_A : A \Rightarrow \mathcal{N}$ defined by: 
\begin{equation*}
\mathit{zero}_A \stackrel{\mathrm{df. }}{=} \mathsf{Pref}(\{ [\hat{q}_\mathscr{E}] [\mathit{no}_\mathscr{E}] \})^{\mathsf{Even}}.
\end{equation*}
\end{definition}

The canonical play by $\mathit{zero}_A$ can be described as follows:
\begin{center}
\begin{tabular}{ccccc}
$[!A_{\mathscr{W}}]_{\Lbag^{\underline{0}} \boldsymbol{e}^+ \Rbag^{\underline{0}} \hbar}$ && $\stackrel{\mathit{zero}_A}{\multimap}$ && $[\mathcal{N}_{\mathscr{E}}]$ \\ \cline{1-5}
&&&& $[\hat{q}_{\mathscr{E}}]$ \tikzmark{czero1} \\
&&&& $[\mathit{no}_{\mathscr{E}}]$ \tikzmark{dzero1}
\end{tabular}
\begin{tikzpicture}[overlay, remember picture, yshift=.25\baselineskip, shorten >=.5pt, shorten <=.5pt]
    \draw [->] ({pic cs:dzero1}) [bend right] to ({pic cs:czero1});
  \end{tikzpicture}
\end{center}

Next, let us recall the \emph{successor strategy}:
\begin{definition}[Successor strategy \cite{yamada2019game}]
\label{DefSuccessorStrategy}
The \emph{\bfseries succesor strategy} is the strategy $\mathit{succ} : \mathcal{N} \Rightarrow \mathcal{N}$ defined by: 
\begin{equation*}
\mathit{succ} \stackrel{\mathrm{df. }}{=} \mathsf{Pref}(\{ [\hat{q}_\mathscr{E}] [\hat{q}_\mathscr{W}]_{\Lbag^{\underline{0}} \Rbag^{\underline{0}} \hbar} ([\mathit{y}_\mathscr{W}]_{\Lbag^{\underline{0}} \Rbag^{\underline{0}} \hbar} [\mathit{y}_\mathscr{E}] [q_\mathscr{E}] [q_\mathscr{W}]_{\Lbag^{\underline{0}} \Rbag^{\underline{0}} \hbar})^i [\mathit{n}_\mathscr{W}]_{\Lbag^{\underline{0}} \Rbag^{\underline{0}} \hbar} [\mathit{y}_\mathscr{E}] [q_\mathscr{E}]  [\mathit{n}_\mathscr{E}] \mid i \in \mathbb{N} \ \! \})^{\mathsf{Even}}
\end{equation*} 
where $\mathit{y}$ and $\mathit{n}$ abbreviate $\mathit{yes}$ and $\mathit{no}$, respectively. 
\end{definition}

The computation of $\mathit{succ}$ can be described as follows:
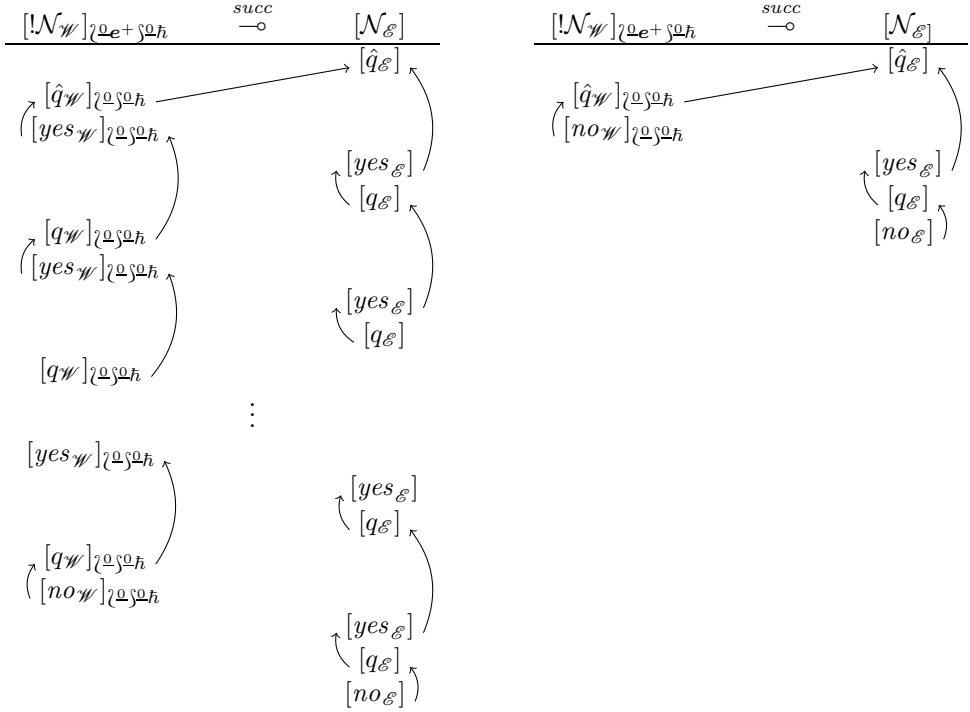
\begin{center}
\begin{tabular}{ccccccccccccc}
$[!\mathcal{N}_{\mathscr{W}}]_{\Lbag^{\underline{0}} \boldsymbol{e}^+ \Rbag^{\underline{0}} \hbar}$ && $\stackrel{\mathit{succ}}{\multimap}$ && $[\mathcal{N}_{\mathscr{E}}]$ &&&& $[!\mathcal{N}_{\mathscr{W}}]_{\Lbag^{\underline{0}} \boldsymbol{e}^+ \Rbag^{\underline{0}} \hbar}$ && $\stackrel{\mathit{succ}}{\multimap}$ && $[\mathcal{N}_{\mathscr{E}]}$ \\ \cline{1-5} \cline{9-13}
&&&&\tikzmark{csucc1} $[\hat{q}_{\mathscr{E}}]$ \tikzmark{csucc3} &&&& &&&&\tikzmark{csucc21} $[\hat{q}_{\mathscr{E}}]$ \tikzmark{csucc23} \\
\tikzmark{csucc2} $[\hat{q}_{\mathscr{W}}]_{\Lbag^{\underline{0}} \Rbag^{\underline{0}} \hbar}$ \tikzmark{dsucc1}&&&& &&&& \tikzmark{csucc22} $[\hat{q}_{\mathscr{W}}]_{\Lbag^{\underline{0}} \Rbag^{\underline{0}} \hbar}$ \tikzmark{dsucc21}&&&& \\
\tikzmark{dsucc2} $[\mathit{yes}_{\mathscr{W}}]_{\Lbag^{\underline{0}} \Rbag^{\underline{0}} \hbar}$ \tikzmark{csucc5}&&&& &&&& \tikzmark{dsucc22} $[\mathit{no}_{\mathscr{W}}]_{\Lbag^{\underline{0}} \Rbag^{\underline{0}} \hbar}$ \tikzmark{csucc25}&&&& \\
&&&&\tikzmark{csucc4} $[\mathit{yes}_{\mathscr{E}}]$ \tikzmark{dsucc3} &&&& &&&&\tikzmark{csucc24} $[\mathit{yes}_{\mathscr{E}}]$ \tikzmark{dsucc23} \\
&&&&\tikzmark{dsucc4} $[q_{\mathscr{E}}]$ \tikzmark{csucc9} &&&& &&&& \tikzmark{dsucc24} $[q_{\mathscr{E}}]$ \tikzmark{csucc29} \\
\tikzmark{csucc8} $[q_{\mathscr{W}}]_{\Lbag^{\underline{0}} \Rbag^{\underline{0}} \hbar}$ \tikzmark{dsucc5}&&&& &&&& &&&&$[\mathit{no}_{\mathscr{E}}]$ \tikzmark{dsucc29} \\
\tikzmark{dsucc8} $[\mathit{yes}_{\mathscr{W}}]_{\Lbag^{\underline{0}} \Rbag^{\underline{0}} \hbar}$ \tikzmark{csucc6}&&&& &&&& &&&& \\
&&&&\tikzmark{csucc7} $[\mathit{yes}_{\mathscr{E}}]$ \tikzmark{dsucc9} &&&& &&&& \\
&&&&\tikzmark{dsucc7} $[q_{\mathscr{E}}]$ &&&& &&&& \\
$[q_{\mathscr{W}}]_{\Lbag^{\underline{0}} \Rbag^{\underline{0}} \hbar}$ \tikzmark{dsucc6}&&&& &&&& &&&& \\
&&$\vdots$&& &&&& &&& \\
$[\mathit{yes}_{\mathscr{W}}]_{\Lbag^{\underline{0}} \Rbag^{\underline{0}} \hbar}$ \tikzmark{csucc10}&&&& &&&& &&&& \\
&&&&\tikzmark{csucc11} $[\mathit{yes}_{\mathscr{E}}]$ &&&& &&&& \\
&&&&\tikzmark{dsucc11} $[q_{\mathscr{E}}]$ \tikzmark{csucc13} &&&&&& && \\
\tikzmark{csucc12} $[q_{\mathscr{W}}]_{\Lbag^{\underline{0}} \Rbag^{\underline{0}} \hbar}$ \tikzmark{dsucc10}&&&& &&&& &&&& \\
\tikzmark{dsucc12} $[\mathit{no}_{\mathscr{W}}]_{\Lbag^{\underline{0}} \Rbag^{\underline{0}} \hbar}$&&&& &&&& &&&& \\
&&&&\tikzmark{csucc14} $[\mathit{yes}_{\mathscr{E}}]$ \tikzmark{dsucc13} &&&& &&&& \\
&&&&\tikzmark{dsucc14} $[q_{\mathscr{E}}]$ \tikzmark{csucc15} &&&& &&&& \\
&&&& $[\mathit{no}_{\mathscr{E}}]$ \tikzmark{dsucc15} &&&& &&&& 
\end{tabular}
\begin{tikzpicture}[overlay, remember picture, yshift=.25\baselineskip]
    \draw [->] ({pic cs:dsucc1}) to ({pic cs:csucc1});
    \draw [->] ({pic cs:dsucc21}) to ({pic cs:csucc21});
    \draw [->] ({pic cs:dsucc2}) [bend left] to ({pic cs:csucc2});
    \draw [->] ({pic cs:dsucc3}) [bend right] to ({pic cs:csucc3});
    \draw [->] ({pic cs:dsucc4}) [bend left] to ({pic cs:csucc4});
    \draw [->] ({pic cs:dsucc5}) [bend right] to ({pic cs:csucc5});
    \draw [->] ({pic cs:dsucc6}) [bend right] to ({pic cs:csucc6});
    \draw [->] ({pic cs:dsucc7}) [bend left] to ({pic cs:csucc7});
    \draw [->] ({pic cs:dsucc8}) [bend left] to ({pic cs:csucc8});
    \draw [->] ({pic cs:dsucc9}) [bend right] to ({pic cs:csucc9});
    \draw [->] ({pic cs:dsucc10}) [bend right] to ({pic cs:csucc10});
    \draw [->] ({pic cs:dsucc11}) [bend left] to ({pic cs:csucc11});
    \draw [->] ({pic cs:dsucc12}) [bend left] to ({pic cs:csucc12});
    \draw [->] ({pic cs:dsucc13}) [bend right] to ({pic cs:csucc13});
    \draw [->] ({pic cs:dsucc14}) [bend left] to ({pic cs:csucc14});
    \draw [->] ({pic cs:dsucc15}) [bend right] to ({pic cs:csucc15});
    \draw [->] ({pic cs:dsucc22}) [bend left] to ({pic cs:csucc22});
    \draw [->] ({pic cs:dsucc23}) [bend right] to ({pic cs:csucc23});
    \draw [->] ({pic cs:dsucc29}) [bend right] to ({pic cs:csucc29});
    \draw [->] ({pic cs:dsucc24}) [bend left] to ({pic cs:csucc24});
  \end{tikzpicture}
\end{center}

Abusing notation, let us define for each $n \in \mathbb{N}$ the strategy $\underline{n} : T \Rightarrow \mathcal{N}$ by $\underline{n} \stackrel{\mathrm{df. }}{=} \mathsf{Pref}(\{ [\hat{q}_{\mathscr{E}}] . ([\mathit{yes}_{\mathscr{E}}] [q_{\mathscr{E}}])^n . [\mathit{no}_{\mathscr{E}}] \})^{\mathsf{Even}}$.
Clearly, $\underline{n}^\dagger ; \mathit{succ} = \underline{n+1} : T \Rightarrow \mathcal{N}$ for all $n \in \mathbb{N}$, and thus $\mathit{succ}$ indeed computes the successor function $\mathbb{N} \rightarrow \mathbb{N}$. 

There is a left inverse of the successor strategy, called the \emph{predecessor strategy}:
\begin{definition}[Predecessor strategy \cite{yamada2019game}]
The \emph{\bfseries predecessor strategy} is the strategy $\mathit{pred} : \mathcal{N} \Rightarrow \mathcal{N}$ defined by:
\begin{align*}
\mathit{pred} \stackrel{\mathrm{df. }}{=} \ &\mathsf{Pref}(\{ [\hat{q}_\mathscr{E}] [\hat{q}_\mathscr{W}]_{\Lbag^{\underline{0}} \Rbag^{\underline{0}} \hbar} [\mathit{y}_\mathscr{W}]_{\Lbag^{\underline{0}} \Rbag^{\underline{0}} \hbar} [q_\mathscr{W}]_{\Lbag^{\underline{0}} \Rbag^{\underline{0}} \hbar} ([\mathit{y}_\mathscr{W}]_{\Lbag^{\underline{0}} \Rbag^{\underline{0}} \hbar} [\mathit{y}_\mathscr{E}] [q_\mathscr{E}] [q_\mathscr{W}]_{\Lbag^{\underline{0}} \Rbag^{\underline{0}} \hbar})^i [\mathit{n}_\mathscr{W}]_{\Lbag^{\underline{0}} \Rbag^{\underline{0}} \hbar} [\mathit{n}_\mathscr{E}] \mid i \in \mathbb{N} \ \! \} \\
&\cup \{ [\hat{q}_\mathscr{E}] [\hat{q}_\mathscr{W}]_{\Lbag^{\underline{0}} \Rbag^{\underline{0}} \hbar} [\mathit{n}_\mathscr{W}]_{\Lbag^{\underline{0}} \Rbag^{\underline{0}} \hbar} [\mathit{n}_\mathscr{E}] \})^{\mathsf{Even}}.
\end{align*}\end{definition}

The computation of $\mathit{pred}$ can be described as follows:
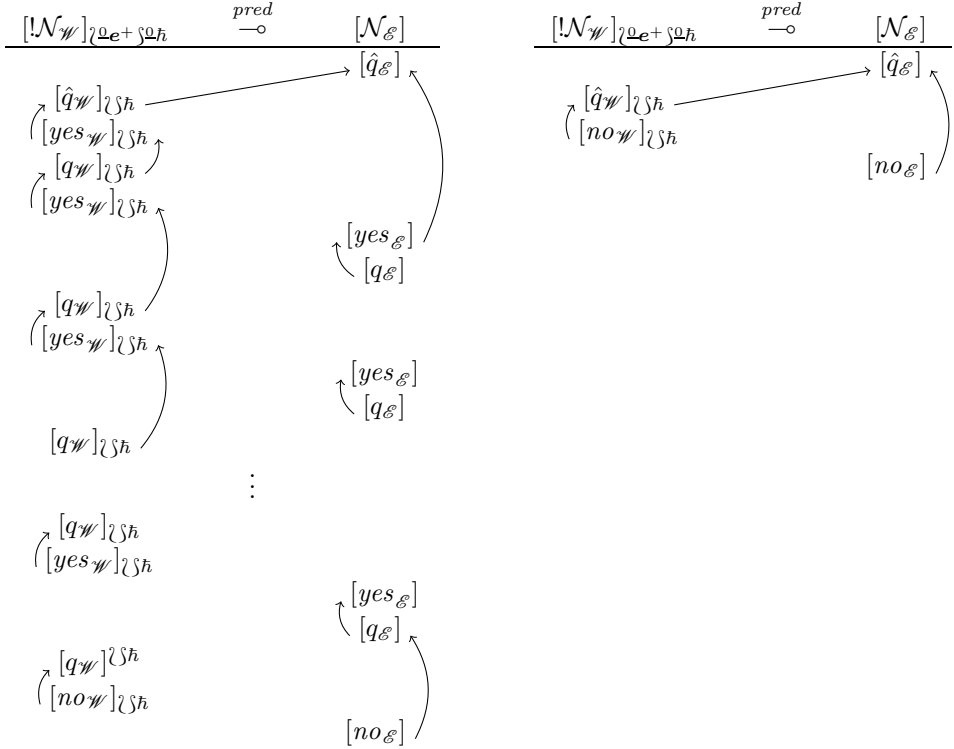
\begin{center}
\begin{tabular}{ccccccccccccc}
$[!\mathcal{N}_{\mathscr{W}}]_{\Lbag^{\underline{0}} \boldsymbol{e}^+ \Rbag^{\underline{0}} \hbar}$ && $\stackrel{\mathit{pred}}{\multimap}$ && $[\mathcal{N}_{\mathscr{E}}]$ &&&& $[!\mathcal{N}_{\mathscr{W}}]_{\Lbag^{\underline{0}} \boldsymbol{e}^+ \Rbag^{\underline{0}} \hbar}$ && $\stackrel{\mathit{pred}}{\multimap}$ && $[\mathcal{N}_{\mathscr{E}}]$ \\ \cline{1-5} \cline{9-13}
&&&&\tikzmark{cpred1} $[\hat{q}_{\mathscr{E}}]$ \tikzmark{cpred3} &&&& &&&&\tikzmark{cpred21} $[\hat{q}_{\mathscr{E}}]$ \tikzmark{cpred23} \\
\tikzmark{cpred2} $[\hat{q}_{\mathscr{W}}]_{\Lbag \Rbag \hbar}$ \tikzmark{dpred1}&&&& &&&& \tikzmark{cpred22} $[\hat{q}_{\mathscr{W}}]_{\Lbag \Rbag \hbar}$ \tikzmark{dpred21}&&&& \\
\tikzmark{dpred2} $[\mathit{yes}_{\mathscr{W}}]_{\Lbag \Rbag \hbar}$ \tikzmark{cpred5}&&&& &&&& \tikzmark{dpred22} $[\mathit{no}_{\mathscr{W}}]_{\Lbag \Rbag \hbar}$ &&&& \\
\tikzmark{cpred8} $[q_{\mathscr{W}}]_{\Lbag \Rbag \hbar}$ \tikzmark{dpred5}&&&& &&&& &&&& $[\mathit{no}_{\mathscr{E}}]$ \tikzmark{dpred23} \\
\tikzmark{dpred8} $[\mathit{yes}_{\mathscr{W}}]_{\Lbag \Rbag \hbar}$ \tikzmark{cpred6} &&&&&& \\
&&&&\tikzmark{cpred7} $[\mathit{yes}_{\mathscr{E}}]$ \tikzmark{dpred3} \\
&&&&\tikzmark{dpred7} $[q_{\mathscr{E}}]$ \\
\tikzmark{cpred9} $[q_{\mathscr{W}}]_{\Lbag \Rbag \hbar}$ \tikzmark{dpred6}&&&& \\
\tikzmark{dpred9} $[\mathit{yes}_{\mathscr{W}}]_{\Lbag \Rbag \hbar}$ \tikzmark{cpred10}&&&& \\
&&&&\tikzmark{cpred11} $[\mathit{yes}_{\mathscr{E}}]$ \\
&&&&\tikzmark{dpred11} $[q_{\mathscr{E}}]$ \\
$[q_{\mathscr{W}}]_{\Lbag \Rbag \hbar}$ \tikzmark{dpred10}&&&& \\
&&$\vdots$&& \\
\tikzmark{cpred12} $[q_{\mathscr{W}}]_{\Lbag \Rbag \hbar}$&&&& \\
\tikzmark{dpred12} $[\mathit{yes}_{\mathscr{W}}]_{\Lbag \Rbag \hbar}$&&&& \\
&&&&\tikzmark{cpred13} $[\mathit{yes}_{\mathscr{E}}]$ \\
&&&&\tikzmark{dpred13} $[q_{\mathscr{E}}]$ \tikzmark{cpred15} \\
\tikzmark{cpred14} $[q_{\mathscr{W}}]^{\Lbag \Rbag \hbar}$&&&& \\
\tikzmark{dpred14} $[\mathit{no}_{\mathscr{W}}]_{\Lbag \Rbag \hbar}$ &&&& \\
&&&&$[\mathit{no}_{\mathscr{E}}]$ \tikzmark{dpred15}
\end{tabular}
\begin{tikzpicture}[overlay, remember picture, yshift=.25\baselineskip]
    \draw [->] ({pic cs:dpred1}) to ({pic cs:cpred1});
    \draw [->] ({pic cs:dpred2}) [bend left] to ({pic cs:cpred2});
    \draw [->] ({pic cs:dpred3}) [bend right] to ({pic cs:cpred3});
    \draw [->] ({pic cs:dpred5}) [bend right] to ({pic cs:cpred5});
    \draw [->] ({pic cs:dpred6}) [bend right] to ({pic cs:cpred6}); 
    \draw [->] ({pic cs:dpred7}) [bend left] to ({pic cs:cpred7});
    \draw [->] ({pic cs:dpred8}) [bend left] to ({pic cs:cpred8});
    \draw [->] ({pic cs:dpred9}) [bend left] to ({pic cs:cpred9});
    \draw [->] ({pic cs:dpred10}) [bend right] to ({pic cs:cpred10});
    \draw [->] ({pic cs:dpred11}) [bend left] to ({pic cs:cpred11});
     \draw [->] ({pic cs:dpred12}) [bend left] to ({pic cs:cpred12});
    \draw [->] ({pic cs:dpred13}) [bend left] to ({pic cs:cpred13});
    \draw [->] ({pic cs:dpred14}) [bend left] to ({pic cs:cpred14});
    \draw [->] ({pic cs:dpred15}) [bend right] to ({pic cs:cpred15});
    \draw [->] ({pic cs:dpred21}) to ({pic cs:cpred21});
    \draw [->] ({pic cs:dpred22}) [bend left] to ({pic cs:cpred22});
    \draw [->] ({pic cs:dpred23}) [bend right] to ({pic cs:cpred23});
  \end{tikzpicture}
\end{center}

It is easy to see that $\underline{n+1}^\dagger ; \mathit{pred} = \underline{n} : T \Rightarrow \mathcal{N}$ for all $n \in \mathbb{N}$, and $\underline{0}^\dagger ; \mathit{pred} = \underline{0} : T \Rightarrow \mathcal{N}$; therefore, $\mathit{pred}$ in fact implements the predecessor function $\mathbb{N} \rightarrow \mathbb{N}$. 

Next, the following \emph{derelictions} play essentially in the same way as copy-cats:
\begin{definition}[Derelictions \cite{abramsky1999game,yamada2019game}]
\label{DefDerelictions}
The dereliction $\mathit{der}_A : A \Rightarrow A$ on a normalized game $A$ is defined by:
\begin{equation*}
\mathit{der}_A \stackrel{\mathrm{df. }}{=} \{ \boldsymbol{s} \in P_{A \Rightarrow A}^{\mathsf{Even}} \mid \forall \boldsymbol{t} \preceq \boldsymbol{s} . \ \! \mathsf{Even}(\boldsymbol{t}) \Rightarrow \boldsymbol{t} \upharpoonright (\mathscr{W})_{\Lbag^{\underline{0}} \Rbag^{\underline{0}} \hbar \_} = \boldsymbol{t} \upharpoonright (\mathscr{E})_{\_} \ \! \}
\end{equation*}
where $\boldsymbol{t} \upharpoonright (\mathscr{W})_{\Lbag^{\underline{0}} \Rbag^{\underline{0}} \hbar \_}$ (resp. $\boldsymbol{t} \upharpoonright (\mathscr{E})_{\_}$) is the j-subsequence of $\boldsymbol{t}$ that consists of moves $[(a, \mathscr{W})]_{\Lbag^{\underline{0}} \Rbag^{\underline{0}} \hbar \boldsymbol{e}}$ (resp. $[(a', \mathscr{E})]_{\boldsymbol{e'}}$) yet changed into $[a]_{\boldsymbol{e}}$ (resp. $[a']_{\boldsymbol{e'}}$).
\end{definition}

Next, let us recall a strategy for the conditional `if...then...else...':
\begin{definition}[Case strategies \cite{yamada2019game}]
Given a normalized game $A$, the \emph{\bfseries case strategy} $\mathit{case}_A : [A_\mathscr{WWW}]_{\Lbag^{\underline{0}} \boldsymbol{e''}^+ \Rbag^{\underline{0}} \hbar \boldsymbol{f}} \& [A_{\mathscr{EWW}}]_{\Lbag^{\underline{0}} \boldsymbol{e'}^+ \Rbag^{\underline{0}} \hbar \boldsymbol{f}} \& [\boldsymbol{2}_{\mathscr{EW}}]_{\Lbag^{\underline{0}} \boldsymbol{e}^+ \Rbag^{\underline{0}} \hbar} \Rightarrow [A_{\mathscr{E}}]_{\boldsymbol{g}}$ on $A$ is defined by:
\begin{align*}
\mathit{case}_A &\stackrel{\mathrm{df. }}{=} \mathsf{Pref}(\{ [a_{\mathscr{E}}]_{\boldsymbol{e}} [\hat{q}_{\mathscr{EW}}]_{\Lbag^{\underline{0}} \boldsymbol{e}^+ \Rbag^{\underline{0}} \hbar} [\mathit{tt}_{\mathscr{EW}}]_{\Lbag^{\underline{0}} \boldsymbol{e}^+ \Rbag^{\underline{0}} \hbar} [a_\mathscr{WWW}]_{\Lbag^{\underline{0}} \Rbag^{\underline{0}} \hbar \boldsymbol{e}} . \boldsymbol{s} \mid [a_{\mathscr{E}}]_{\boldsymbol{e}} [a_\mathscr{WWW}]_{\Lbag^{\underline{0}} \Rbag^{\underline{0}} \hbar \boldsymbol{e}} . \boldsymbol{s} \in \mathit{der}^\mathscr{W}_A \} \\
& \ \ \ \ \cup \{ [a_{\mathscr{E}}]_{\boldsymbol{f}} [\hat{q}_{\mathscr{EW}}]_{\Lbag^{\underline{0}} \boldsymbol{f}^+ \Rbag^{\underline{0}} \hbar} [\mathit{ff}_{\mathscr{EW}}]_{\Lbag^{\underline{0}} \boldsymbol{f}^+ \Rbag^{\underline{0}} \hbar} [a_{\mathscr{EWW}}]_{\Lbag^{\underline{0}} \Rbag^{\underline{0}} \hbar \boldsymbol{f}} . \boldsymbol{t} \mid [a_{\mathscr{E}}]_{\boldsymbol{f}} [a_{\mathscr{EWW}}]_{\Lbag^{\underline{0}} \Rbag^{\underline{0}} \hbar \boldsymbol{f}} . \boldsymbol{t} \in \mathit{der}^{\mathscr{E}}_A \})^{\mathsf{Even}}
\end{align*}
where $\mathit{der}^\mathscr{W}_A : [A_\mathscr{WWW}]_{\Lbag^{\underline{0}} \boldsymbol{e''}^+ \Rbag^{\underline{0}} \hbar \boldsymbol{f}} \Rightarrow  [A_{\mathscr{E}}]_{\boldsymbol{g}}$ and $\mathit{der}^{\mathscr{E}}_A : [A_{\mathscr{EWW}}]_{\Lbag^{\underline{0}} \boldsymbol{e'}^+ \Rbag^{\underline{0}} \hbar \boldsymbol{f}} \Rightarrow  [A_{\mathscr{E}}]_{\boldsymbol{g}}$ are the same as the dereliction $\mathit{der}_A : [A_\mathscr{W}]_{\Lbag^{\underline{0}} \boldsymbol{e'}^+ \Rbag^{\underline{0}} \hbar \boldsymbol{f}} \Rightarrow [A_\mathscr{E}]_{\boldsymbol{g}}$ up to inner tags.
\end{definition}

As the name suggests, it implements the case distinction on $A$: Given input strategies $\sigma_1, \sigma_2 : T \Rightarrow A$ and $\beta : T \Rightarrow \boldsymbol{2}$, the composition $\langle \langle \sigma_1, \sigma_2 \rangle, \beta \rangle^\dagger ; \mathit{case}_A$ is $\sigma_1$ (resp. $\sigma_2$) if $\beta$ is $\underline{\mathit{tt}} \stackrel{\mathrm{df. }}{=} \{ \boldsymbol{\epsilon}, [\hat{q}_{\mathscr{E}}] [\mathit{tt}_{\mathscr{E}}] \}$ (resp. $\underline{\mathit{ff}} \stackrel{\mathrm{df. }}{=} \{ \boldsymbol{\epsilon}, [\hat{q}_{\mathscr{E}}] [\mathit{ff}_{\mathscr{E}}] \}$).

Next, let us recall a strategy that sees if a given input is zero or not:
\begin{definition}[Ifzero strategy \cite{yamada2019game}]
The \emph{\bfseries ifzero strategy} is the strategy $\mathit{zero?} : [!\mathcal{N}_\mathscr{W}]_{\Lbag^{\underline{0}} \boldsymbol{e}^+ \Rbag^{\underline{0}} \hbar} \multimap [\boldsymbol{2}_\mathscr{E}]$ defined by:
\begin{equation*}
\mathit{zero?} \stackrel{\mathrm{df. }}{=} \mathsf{Pref}(\{ [\hat{q}_\mathscr{E}] [\hat{q}_\mathscr{W}]_{\Lbag^{\underline{0}} \Rbag^{\underline{0}} \hbar} [\mathit{no}_\mathscr{W}]_{\Lbag^{\underline{0}} \Rbag^{\underline{0}} \hbar} [\mathit{tt}_\mathscr{E}], [\hat{q}_\mathscr{E}] [\hat{q}_\mathscr{W}]_{\Lbag^{\underline{0}} \Rbag^{\underline{0}} \hbar} [\mathit{yes}_\mathscr{W}]_{\Lbag^{\underline{0}} \Rbag^{\underline{0}} \hbar} [\mathit{ff}_\mathscr{E}] \})^{\mathsf{Even}}.
\end{equation*}
\end{definition}

Clearly, we have $\underline{0}^\dagger ; \mathit{zero?} = \underline{\mathit{tt}} : T \Rightarrow \boldsymbol{2}$ and $\underline{n+1}^\dagger ; \mathit{zero?} = \underline{\mathit{ff}} : T \Rightarrow \boldsymbol{2}$ for all $n \in \mathbb{N}$, in fact checking correctly if the input is zero or not.

Finally, let us recall strategies that model \emph{fixed-point combinators} of PCF \cite{plotkin1977lcf,amadio1998domains}.
Since the precise definition is rather involved, we leave it to \cite{hyland1997game,yamada2019game}; for the present work, the following description suffices: 
\begin{definition}[Fixed-point strategies \cite{yamada2019game}]
\label{DefFixedPointStrategies}
Given a normalized game $A$, the \emph{\bfseries fixed-point strategy} $\mathit{fix}_A : ([A_{\mathscr{WW}}]_{\Lbag^{\underline{0}} \boldsymbol{e'}^+ \Rbag^{\underline{0}} \hbar \Lbag^{\underline{0}} \boldsymbol{e}^+ \Rbag^{\underline{0}} \hbar \boldsymbol{f}} \Rightarrow [A_{\mathscr{EW}}]_{\Lbag^{\underline{0}} \boldsymbol{e'}^+ \Rbag^{\underline{0}} \hbar \boldsymbol{f}}) \Rightarrow [A_{\mathscr{E}}]_{\boldsymbol{g}}$ on $A$ computes as follows:
\begin{itemize}

\item After the first occurrence $[a_\mathscr{E}]_{\boldsymbol{g}}$, $\mathit{fix}_A$ copies it and performs the move $[a_{\mathscr{EW}}]_{\Lbag^{\underline{0}} \Rbag^{\underline{0}} \hbar \boldsymbol{g}}$ with the pointer towards the initial occurrence $[a_\mathscr{E}]_{\boldsymbol{g}}$;

\item If O initiates a new thread in the inner implication by $[a'_{\mathscr{WW}}]_{\Lbag^{\underline{0}} \boldsymbol{e'}^+ \Rbag^{\underline{0}} \hbar \Lbag^{\underline{0}} \boldsymbol{e}^+ \Rbag^{\underline{0}} \hbar \boldsymbol{f}}$, then $\mathit{fix}_A$ copies it and launches a new thread in the outer implication by $[a'_{\mathscr{EW}}]_{\Lbag^{\underline{0}} \Lbag^{\underline{1}} \boldsymbol{e'}^{++} \Rbag^{\underline{1}} \hbar \Lbag^{\underline{1}} \boldsymbol{e}^{++} \Rbag^{\underline{1}} \Rbag^{\underline{0}} \hbar \boldsymbol{f}}$ together with the pointer towards the justifier of the justifier of the O-move;

\item If O performs in an existing thread a move $[a''_{\mathscr{WW}}]_{\Lbag^{\underline{0}} \boldsymbol{e'}^+ \Rbag^{\underline{0}} \hbar \Lbag^{\underline{0}} \boldsymbol{e}^+ \Rbag^{\underline{0}} \hbar \boldsymbol{f}}$ (resp. $[a''_{\mathscr{EW}}]_{\Lbag^{\underline{0}} \Rbag^{\underline{0}} \hbar \boldsymbol{f}}$, $[a''_{\mathscr{EW}}]_{\Lbag^{\underline{0}} \Lbag^{\underline{1}} \boldsymbol{e'}^{++} \Rbag^{\underline{1}} \hbar \Lbag^{\underline{1}} \boldsymbol{e}^{++} \Rbag^{\underline{1}} \Rbag^{\underline{0}} \hbar \boldsymbol{f}}$, $[a''_{\mathscr{E}}]_{\boldsymbol{f}}$), then $\mathit{fix}_A$ copies it and performs in the \emph{dual thread}, i.e., in the thread to which the third last occurrence of the current  P-view (Appendix~\ref{DefViews}) belongs, the move $[a''_{\mathscr{EW}}]_{\Lbag^{\underline{0}} \Lbag^{\underline{1}} \boldsymbol{e'}^{++} \Rbag^{\underline{1}} \hbar \Lbag^{\underline{1}} \boldsymbol{e}^{++} \Rbag^{\underline{1}} \Rbag^{\underline{0}} \hbar \boldsymbol{f}}$ (resp. $[a''_{\mathscr{E}}]_{\boldsymbol{f}}$, $[a''_{\mathscr{WW}}]_{\Lbag^{\underline{0}} \boldsymbol{e'}^+ \Rbag^{\underline{0}} \hbar \Lbag^{\underline{0}} \boldsymbol{e}^+ \Rbag^{\underline{0}} \hbar \boldsymbol{f}}$, $[a''_{\mathscr{EW}}]_{\Lbag^{\underline{0}} \Rbag^{\underline{0}} \hbar \boldsymbol{f}}$) together with the pointer towards the third last move.

\end{itemize}
\end{definition}

We are now ready to recall an enumeration of all the strategies for PCF-computation: 
\begin{definition}[Strategies for PCF \cite{yamada2019game}]
\label{DefStrategiesForPCF}
Let $\mathcal{DPCF}$ be the least set of strategies that satisfies:  
\begin{enumerate}

\item $(\sigma : G) \in \mathcal{DPCF}$ if $\sigma : G$ is `atomic', i.e., $\mathit{der}_A : A \Rightarrow A$, $\mathit{zero}_A : A \Rightarrow \mathcal{N}$, $\mathit{succ} : \mathcal{N} \Rightarrow \mathcal{N}$, $\mathit{pred} : \mathcal{N} \Rightarrow \mathcal{N}$, $\mathit{zero}? : \mathcal{N} \Rightarrow \boldsymbol{2}$, $\mathit{case}_{A} : A \& A \& \boldsymbol{2} \Rightarrow A$ or $\mathit{fix}_A : (A \Rightarrow A) \Rightarrow A$, where $A$ is a normalized game generated from $\mathcal{N}$, $\boldsymbol{2}$ and/or $T$ by $\&$ and/or $\Rightarrow$ (n.b., the construction of $A$ is `orthogonal' to that of $\sigma : G$);

\item $(\Lambda(\sigma) : \Lambda(G)) \in \mathcal{DPCF}$ if $(\sigma : G) \in \mathcal{DPCF}$ and $\mathcal{H}^\omega(G) \trianglelefteqslant A \& B \Rightarrow C$ for some normalized games $A$, $B$ and $C$;

\item $(\langle \varphi, \psi \rangle : \langle L, R \rangle) \in \mathcal{DPCF}$ if $(\varphi : L), (\psi : R) \in \mathcal{DPCF}$, $\mathcal{H}^\omega(L) \trianglelefteqslant C \Rightarrow A$ and $\mathcal{H}^\omega \trianglelefteqslant C \Rightarrow B$ for some normalized games $A$, $B$ and $C$;

\item $(\iota^\dagger \ddagger \kappa : J^\dagger \ddagger K) \in \mathcal{DPCF}$ if $(\iota : J), (\kappa : K) \in \mathcal{DPCF}$, $\mathcal{H}^\omega(J) \trianglelefteqslant A \Rightarrow B$ and $\mathcal{H}^\omega(K) \trianglelefteqslant B \Rightarrow C$ for some normalized games $A$, $B$ and $C$

\end{enumerate}
where \emph{projections} and \emph{evaluations}  \cite{abramsky1999game} are derelictions up to inner tags, and therefore we count them as `atomic' ones.
\end{definition}

It has been shown in \cite{yamada2019game} that the set $\mathcal{DPCF}$ contains for every term $\mathsf{\Gamma \vdash M : A}$ of PCF the denotation of the term $\mathsf{M}$. 
Hence, our mathematical problem has been reduced to showing that every strategy in the set $\mathcal{DPCF}$ is `implementable' by an automaton that is strictly weaker than a TM.

\section{Seemingly counter-Chomsky}
\label{JPointingAutomata}
Our main contribution is the present section.
We first define our pushdown automata, called \emph{j-pushdown automata}, and show that they are strictly weaker than TMs in a certain sense in Section~\ref{J}.
Then, we prove that j-pushdown automata are PCF-complete, which is our main theorem, in Section~\ref{PCFCompletenessOfJPushdownAutomata}. 
Then, we further proceed to give a corollary of the theorem, establishing that similar stack automata, called \emph{j-stack automata}, are Turing complete without interaction with another computational agent in Section~\ref{StandAloneTuringCompleteness}.

\subsection{J-pushdown automata}
\label{J}
Let us first give the formal definition of j-pushdown automata.

\begin{definition}[J-pointing tapes]
\label{DefJPointingTapes}
The \emph{\bfseries j-pointing tape} for a game $G$ is the infinite tape (which is standard in automata theory \cite{kozen2012automata,hopcroft1979introduction}) equipped with at most one directed edge\footnote{We regard the standard tape as a directed graph $\square \leftarrow \square \leftarrow \dots$, where vertices $\square$ are cells, and edges $\leftarrow$ give the order of the cells; we then equip the tape with additional edges, which we're talking about.} between two cells such that a tape content must be the string $\overline{\boldsymbol{s}}$ equipped with edges for some $\boldsymbol{s} \in P_G$ given by induction on $|\boldsymbol{s}|$: 
\begin{itemize}

\item If $\boldsymbol{s} . [o]_{g_1 g_2 \dots g_l} \in P_G$, and $[o]_{g_1 g_2 \dots g_l}$ is initial, then $\overline{\boldsymbol{s} . [o]_{g_1 g_2 \dots g_l}}$ is: \\
\begin{center}
\begin{tikzpicture}[every node/.style={block},
        block/.style={minimum height=1.5em,outer sep=0pt,draw,rectangle,node distance=0pt}]
   \node (X) {$\mathsf{\$}$};
   \node (Y) [left=of X] {$\mathsf{o}$};
   \node (Z) [left=of Y] {$\mathsf{g_1}$};
   \node (V) [left=of Z] {$\mathsf{g_2}$};
   \node (W) [left=of V] {$\ldots$}; 
   \node (S) [left=of W] {$\mathsf{g_{l}}$};
   \node (T) [left=of S] {$\overline{\boldsymbol{s}}$};
   \node (U) [left=of T] {$\mathsf{\vdash}$};
   \draw (U.north west) -- ++(-1cm,0) (U.south west) -- ++ (-1cm,0) 
                 (X.north east) -- ++(1cm,0) (X.south east) -- ++ (1cm,0);
\end{tikzpicture}
\end{center}

\item If $\boldsymbol{s} . [n]_{g_1 g_2 \dots g_l} . \boldsymbol{t} . [m]_{e_1 e_2 \dots e_k} \in P_G$, and $[n]_{g_1 g_2 \dots g_l}$ is the justifier of $[m]_{e_1 e_2 \dots e_k}$, then $\overline{\boldsymbol{s} . [n]_{g_1 g_2 \dots g_l} . \boldsymbol{t} . [m]_{e_1 e_2 \dots e_k}}$ is:
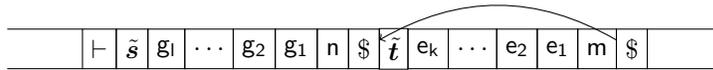
\begin{center}
\begin{tikzpicture}[every node/.style={block},
        block/.style={minimum height=1.5em,outer sep=0pt,draw,rectangle,node distance=0pt}]
   \node (A) {$\mathsf{e_{k}}$};
   \node (B) [left=of A] {$\tilde{\boldsymbol{t}}$};
   \node (C) [left=of B] {$\mathsf{\$}$};
   \node (G) [left=of C] {$\mathsf{n}$};
   \node (H) [left=of G] {$\mathsf{g_1}$};
   \node (K) [left=of H] {$\mathsf{g_2}$}; 
   \node (I) [left=of K] {$\ldots$};
   \node (J) [left=of I] {$\mathsf{g_{l}}$};
   \node (O) [left=of J] {$\tilde{\boldsymbol{s}}$};
   \node (P) [left=of O] {$\vdash$};
   \node (D) [right=of A] {$\ldots$};
   \node (E) [right=of D] {$\mathsf{e_2}$};
   \node (L) [right=of E] {$\mathsf{e_1}$};
   \node (M) [right=of L] {$\mathsf{m}$};
   \node (N) [right=of M] {$\mathsf{\$}$}; 
   \draw (P.north west) -- ++(-1cm,0) (P.south west) -- ++ (-1cm,0) 
                 (N.north east) -- ++(1cm,0) (N.south east) -- ++ (1cm,0);
   \draw [->] (N) [bend right] to (C);
\end{tikzpicture}
\end{center}

\end{itemize}
where $\tilde{\boldsymbol{s}} \mathsf{g_l \dots g_2 g_1 n} \$ \tilde{\boldsymbol{t}} \stackrel{\mathrm{df. }}{=} \overline{\boldsymbol{s} . [n]_{g_1 g_2 \dots g_l} \boldsymbol{t}}$ (n.b., there are possibly other edges from/towards the $\$$ between $\mathsf{n}$ and $\tilde{\boldsymbol{t}}$, but we omit them), $\mathsf{\vdash}$ is a distinguished symbol to indicate where $\overline{\boldsymbol{s}}$ begins, and $\mathsf{\$}$ is another to serve as a separator of moves.
\end{definition}

\begin{definition}[J-pushdown automata]
\label{DefJPushdownAutomata}
A \emph{\bfseries j-pushdown automaton} for a game $G$ is a deterministic, \emph{non-erasing} (i.e., it never pops off a symbol on the top of the stack) pushdown automaton \cite{kozen2012automata,sipser2006introduction,hopcroft1979introduction} such that:
\begin{itemize}

\item \textsc{(JPA1)} The tape and the stack symbols are elements of $\pi_1(M_G) \cup \mathcal{T} \cup \{ \vdash, \$, \mathsf{i}, \mathsf{ii}, \mathsf{iii} \}$, where we assume that they are pairwise distinct;

\item \textsc{(JPA2)} Its input tape is the j-pointing tape for $G$, and its reading head moves from the rightmost occurrence of $\$$ towards left up to $\vdash$;

\item \textsc{(JPA3)} If the head is on a cell of the tape containing $\mathsf{\$}$ that occurs on the immediate right of the cells containing symbols representing a non-initial O-move, then the head must \emph{jump} to the cell containing $\mathsf{\$}$ connected by the (unique) directed edge. 

\end{itemize}
\end{definition}

By the axiom JPA3, a j-pushdown automaton may move only to cells on the j-pointing tape that contains symbols representing moves in P-views (Appendix~\ref{DefViews}).
Hence, for a finite sequence of symbols (without edges), i.e., an ordinary input for an automaton, a j-pushdown automaton computes exactly in the same way as a deterministic, non-erasing pushdown automaton. Thus, we have:
\begin{proposition}[Weakness of j-pushdown automata]
\label{PropStrictWeakness}
There is a formal language that can be recognized by a TM but not by any j-pushdown automaton. 
\end{proposition}

Nevertheless, in the game-semantic framework, note that:
\begin{enumerate}

\item \textsc{(Game-semantic compromise).} J-pushdown automata only need to compute the next P-move \emph{into the stack}, not onto the input tape, assuming that \emph{Judge (J)} of the game copies the stack content onto the input tape;

\item \textsc{(Edges on the input tape).} The participants of a game perform moves with \emph{pointers}, which j-pushdown automata may utilize.

\end{enumerate}

As we shall see in the next section, this game-semantic setting brings j-pushdown automata not only Turing completeness but also PCF-completeness.

\subsection{PCF-completeness of j-pushdown automata}
\label{PCFCompletenessOfJPushdownAutomata}
We now prove that elements of the set $\mathcal{DPCF}$ (Definition~\ref{DefStrategiesForPCF}) are all `implementable' by j-pushdown automata (Theorem~\ref{ThmJ}) in the following sense: 

\begin{definition}[JPA-computability]
\label{DefAComputableStrategies}
A strategy $\sigma : G$ is \emph{\bfseries JPA-computable} if there is a j-pushdown automaton $\mathscr{A}$ for $G$ such that for each $\boldsymbol{s} . [o]_{g_1 g_2 \dots g_l} . [p]_{e_1 e_2 \dots e_k} \in \sigma$ with the prefix $\boldsymbol{s} . [o]_{g_1 g_2 \dots g_l} \in P_{G}^{\mathsf{Odd}}$ written on the j-pointing tape for $G$ in the form defined in Definition~\ref{DefJPointingTapes} the computation of $\mathscr{A}$ terminates with the stack content \\
\begin{center}
\begin{tikzpicture}[every node/.style={block},
        block/.style={minimum height=1.5em,outer sep=0pt,draw,rectangle,node distance=0pt}]
   \node (H) [left=of G] {$\mathsf{e_{k}}$};
   \node (K) [left=of H] {$\ldots$}; 
   \node (I) [left=of K] {$\mathsf{e_2}$};
   \node (J) [left=of I] {$\mathsf{e_1}$};
   \node (O) [left=of J] {$\mathsf{p}$};
   \node (P) [left=of O] {$\mathsf{J}$};
   \node (Q) [left=of P] {$\vdash$};
   \draw (Q.north west) -- ++(-0.5cm,0) (Q.south west) -- ++ (-0.5cm,0) 
                 (H.north east) -- ++(0.5cm,0) (H.south east) -- ++ (0.5cm,0);
\end{tikzpicture}
\end{center}
where $\mathsf{J} = \mathsf{i}$ (resp. $\mathsf{J} = \mathsf{ii}$, $\mathsf{J} = \mathsf{iii}$) iff the justifier of $[p]_{e_1 e_2 \dots e_k}$ is the last occurrence (resp. the justifier of the second last occurrence, the third last occurrence) of the P-view of $\boldsymbol{s} . [o]_{g_1 g_2 \dots g_l}$ (Appendix~\ref{DefViews}).
In this case, $\mathscr{A}$ is said to \emph{\bfseries realize} $\sigma$.
\end{definition}

\begin{theorem}[Main theorem]
\label{ThmJ}
Every strategy in $\mathcal{DPCF}$ is JPA-computable.
\end{theorem}

\begin{proof}
First, consider the successor strategy $\mathit{succ} : \mathcal{N} \Rightarrow \mathcal{N}$; see Figure~\ref{FigSuccOnTape}, where any string (from left to right) occurring in a diagram is to be read vertically from bottom to top.
We employ this notation in the rest of the paper.
It is then easy to see that $\mathit{succ}$ is JPA-computable (n.b., since automata theory is an established branch, it should be legitimate and even appropriate \emph{not} to give the formal, full description of the automata).

Next, it is even simpler to see JPA-computability of the dereliction $\mathit{der}_A : A \Rightarrow A$ on any game $A$; see Figure~\ref{FigDerOnTape} below. 
\begin{figure}[H]
\begin{center}
\begin{tabular}{ccccccccccccccc}
$[\oc \mathcal{N}_{\mathscr{W}}]_{\Lbag^{\underline{0}} \boldsymbol{e}^+ \Rbag^{\underline{0}} \hbar}$ && $\stackrel{\mathit{succ}}{\multimap}$ && $[\mathcal{N}_{\mathscr{E}}]$ &&&&&& $[\oc \mathcal{N}_{\mathscr{W}}]_{\Lbag^{\underline{0}} \boldsymbol{e}^+ \Rbag^{\underline{0}} \hbar}$ && $\stackrel{\mathit{succ}}{\multimap}$ && $[\mathcal{N}_{\mathscr{E}}]$ \\ \cline{1-5} \cline{10-15}
&&&& $\mathsf{\hat{q}_{E}}$ &&&& &&&&&& $\mathsf{\hat{q}_{E}}$ \\
&&&&\tikzmark{csuccF1} $\$$ \tikzmark{csuccF3} &&&&&&&&&& \tikzmark{csuccF21} $\$$ \tikzmark{csuccF23} \\
$\hbar$&&&&&&&&&&$\hbar$&&&& \\
$\Lbag^{\underline{0}} \Rbag^{\underline{0}}$&&&&&&&&&&$\Lbag^{\underline{0}} \Rbag^{\underline{0}}$&&&& \\
$\mathsf{\hat{q}_W}$&&&&&&&&&&$\mathsf{\hat{q}_{W}}$&&&& \\
\tikzmark{csuccF2} $\$$ \tikzmark{dsuccF1}&&&&&& &&&& \tikzmark{csuccF22} $\$$ \tikzmark{dsuccF21}&&&& \\
$\hbar$&&&&&&&&&&$\hbar$&&&& \\
$\Lbag^{\underline{0}} \Rbag^{\underline{0}}$&&&&&&&&&&$\Lbag^{\underline{0}} \Rbag^{\underline{0}}$&&&& \\
$\mathsf{yes_W}$&&&&&&&&&&$\mathsf{no_W}$&&&& \\
\tikzmark{dsuccF2} $\$$ \tikzmark{csuccF5}&&&& &&&&&& \tikzmark{dsuccF22} $\$$ \tikzmark{csuccF25}&&&& \\
&&&&$\mathsf{yes_E}$&&&&&&&&&&$\mathsf{yes_E}$ \\
&&&&\tikzmark{csuccF4} $\$$ \tikzmark{dsuccF3} &&&&&& &&&&\tikzmark{csuccF24} $\$$ \tikzmark{dsuccF23} \\
&&&&$\mathsf{q_E}$&&&&&&&&&&$\mathsf{q_E}$ \\
&&&&\tikzmark{dsuccF4} $\$$ \tikzmark{csuccF9} &&&& &&&&&&\tikzmark{dsuccF24} $\$$ \tikzmark{csuccF29} \\
$\hbar$&&&&&&&&&&&&&&$\mathsf{no_E}$ \\
$\Lbag^{\underline{0}} \Rbag^{\underline{0}}$&&&&&&&&&&&&&&$\$$ \tikzmark{dsuccF29} \\
$\mathsf{q_W}$&&&&&&&&&&&&&& \\
\tikzmark{csuccF8} $\$$ \tikzmark{dsuccF5}&&&& &&&&&& &&&&\tikzmark{csuccF27} \\
$\hbar$&&&&&&&&&&&&&& \\
$\Lbag^{\underline{0}} \Rbag^{\underline{0}}$&&&&&&&&&&&&&& \\
$\mathsf{no_W}$&&&&&&&&&&&&&& \\
\tikzmark{dsuccF8} $\$$&&&& &&&&&& &&&& \\
&&&&$\mathsf{yes_E}$&&&&&&&&&& \\
&&&&\tikzmark{csuccF7} $\$$ \tikzmark{dsuccF9} &&&& &&&&&& \\
&&&&$\mathsf{q_E}$&&&&&&&&&& \\
&&&&\tikzmark{dsuccF7} $\$$ \tikzmark{csuccF6} &&&& &&&&&& \\
&&&& $\mathsf{no_E}$ &&&&&&&& \\
&&&& $\$$ \tikzmark{dsuccF6} &&&&&& &&&& 
\end{tabular}
\begin{tikzpicture}[overlay, remember picture, yshift=.25\baselineskip]
    \draw [->] ({pic cs:dsuccF1}) to ({pic cs:csuccF1});
    \draw [->] ({pic cs:dsuccF2}) [bend left] to ({pic cs:csuccF2});
    \draw [->] ({pic cs:dsuccF3}) [bend right] to ({pic cs:csuccF3});
    \draw [->] ({pic cs:dsuccF4}) [bend left] to ({pic cs:csuccF4});
    \draw [->] ({pic cs:dsuccF5}) [bend right] to ({pic cs:csuccF5});
    \draw [->] ({pic cs:dsuccF6}) [bend right] to ({pic cs:csuccF6}); 
    \draw [->] ({pic cs:dsuccF7}) [bend left] to ({pic cs:csuccF7});
    \draw [->] ({pic cs:dsuccF8}) [bend left] to ({pic cs:csuccF8});
    \draw [->] ({pic cs:dsuccF9}) [bend right] to ({pic cs:csuccF9});
    \draw [->] ({pic cs:dsuccF21}) to ({pic cs:csuccF21});
    \draw [->] ({pic cs:dsuccF22}) [bend left] to ({pic cs:csuccF22});
    \draw [->] ({pic cs:dsuccF23}) [bend right] to ({pic cs:csuccF23});
    \draw [->] ({pic cs:dsuccF24}) [bend left] to ({pic cs:csuccF24});
    \draw [->] ({pic cs:dsuccF29}) [bend right] to ({pic cs:csuccF29});
  \end{tikzpicture}
  \caption{Some plays by the successor strategy $\mathit{succ} : \mathcal{N} \Rightarrow \mathcal{N}$ on the j-pointing tape.}
  \label{FigSuccOnTape}
\end{center}
\end{figure}

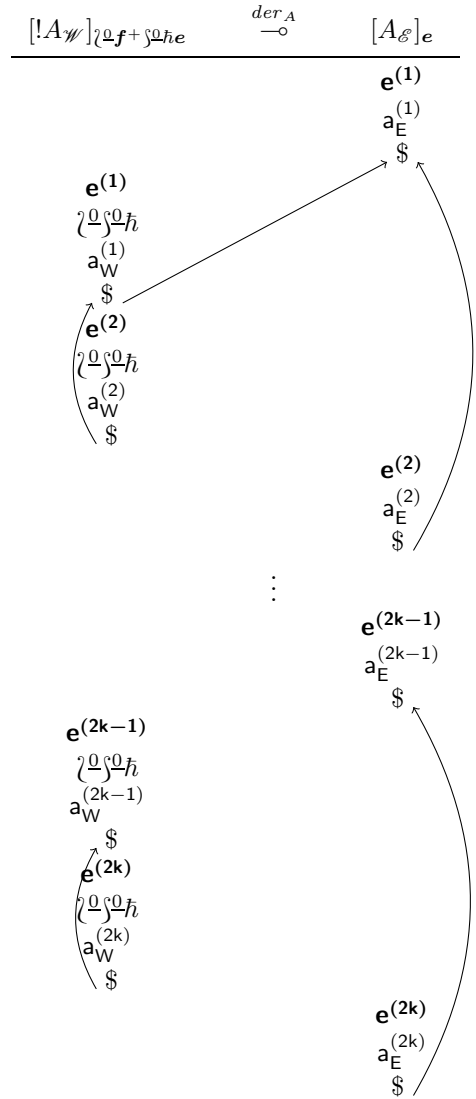
\begin{figure}[H]
\begin{center}
\begin{tabular}{ccccc}
$[\oc A_{\mathscr{W}}]_{\Lbag^{\underline{0}} \boldsymbol{f}^+ \Rbag^{\underline{0}} \hbar \boldsymbol{e}}$ && $\stackrel{\mathit{der}_A}{\multimap}$ && $[A_{\mathscr{E}}]_{\boldsymbol{e}}$ \\ \hline
&&&&$\boldsymbol{\mathsf{e^{(1)}}}$ \\
&&&&$\mathsf{a^{(1)}_E}$ \\
&&&&\tikzmark{cder1} $\$$ \tikzmark{cder3} \\
$\boldsymbol{\mathsf{e^{(1)}}}$&&&& \\
$\Lbag^{\underline{0}} \Rbag^{\underline{0}} \hbar$&&&& \\
$\mathsf{a^{(1)}_W}$&&&& \\
\tikzmark{cder2} $\$$ \tikzmark{dder1}&&&& \\
$\boldsymbol{\mathsf{e^{(2)}}}$&&&& \\
$\Lbag^{\underline{0}} \Rbag^{\underline{0}} \hbar$&&&& \\
$\mathsf{a^{(2)}_W}$&&&& \\
\tikzmark{dder2} $\$$&&&& \\
&&&&$\boldsymbol{\mathsf{e^{(2)}}}$ \\
&&&&$\mathsf{a^{(2)}_E}$ \\
&&&&$\$$ \tikzmark{dder3} \\
&&$\vdots$&& \\
&&&&$\boldsymbol{\mathsf{e^{(2k-1)}}}$ \\
&&&&$\mathsf{a^{(2k-1)}_E}$ \\
&&&&$\$$ \tikzmark{cder5} \\
$\boldsymbol{\mathsf{e^{(2k-1)}}}$ \\
$\Lbag^{\underline{0}} \Rbag^{\underline{0}} \hbar$&&&& \\
$\mathsf{a^{(2k-1)}_W}$&&&& \\
\tikzmark{cder4} $\$$&&&& \\
$\boldsymbol{\mathsf{e^{(2k)}}}$ &&&&\\
$\Lbag^{\underline{0}} \Rbag^{\underline{0}} \hbar$&&&& \\
$\mathsf{a^{(2k)}_W}$&&&& \\
\tikzmark{dder4} $\$$&&&& \\
&&&&$\boldsymbol{\mathsf{e^{(2k)}}}$ \\
&&&&$\mathsf{a^{(2k)}_E}$ \\
&&&&$\$$ \tikzmark{dder5}
\end{tabular}
\begin{tikzpicture}[overlay, remember picture, yshift=.25\baselineskip]
    \draw [->] ({pic cs:dder1}) to ({pic cs:cder1});
    \draw [->] ({pic cs:dder2}) [bend left] to ({pic cs:cder2});
    \draw [->] ({pic cs:dder3}) [bend right] to ({pic cs:cder3});
    \draw [->] ({pic cs:dder4}) [bend left] to ({pic cs:cder4});
    \draw [->] ({pic cs:dder5}) [bend right] to ({pic cs:cder5});
\end{tikzpicture}
 \caption{Some play by the dereliction $\mathit{der}_A : A \Rightarrow A$ on the j-pointing tape.}
  \label{FigDerOnTape}
\end{center}
\end{figure}

\begin{figure}[H]
\begin{center}
\begin{tabular}{ccccccccccccc}
$([!!A_{\mathscr{WW}}]_{\Lbag^{\underline{0}} \boldsymbol{f}^+ \Rbag^{\underline{0}} \hbar \Lbag^{\underline{0}} \boldsymbol{g}^+ \Rbag^{\underline{0}} \hbar \boldsymbol{e}}$ &&& $\multimap$ &&& $[!A_{\mathscr{EW}}]_{\Lbag^{\underline{0}} \boldsymbol{f'}^+ \Rbag^{\underline{0}} \boldsymbol{e'}})$ &&& $\stackrel{\mathit{fix}_A}{\multimap}$ &&& $[A_{\mathscr{E}}]_{\boldsymbol{e''}}$  \\ \hline
&&&&&&&&&&&& $\boldsymbol{\mathsf{e^{(1)}}}$ \\
&&&&&&&&&&&&\tikzmark{cjfix1} $\mathsf{a^{(1)}_E}$ \\
&&&&&&&&&&&&\tikzmark{cjfix1} $\$$ \tikzmark{cjfix3} \\
&&&&&& $\boldsymbol{\mathsf{e^{(1)}}}$ &&&&&& \\
&&&&&& $\Lbag^{\underline{0}} \Rbag^{\underline{0}} \hbar$ &&&&&& \\
&&&&&& $\mathsf{a^{(1)}_{EW}}$ &&&&&& \\
&&&&&&\tikzmark{cjfix2} $\$$ \tikzmark{djfix1}&&&&&& \\
&&&&&& $\boldsymbol{\mathsf{e^{(2)}}}$ &&&&&& \\
&&&&&&$\Lbag^{\underline{0}} \Rbag^{\underline{0}} \hbar$&&&&&& \\
&&&&&&$\mathsf{a^{(2)}_{EW}}$&&&&&& \\
&&&&&&\tikzmark{djfix2} $\$$ \tikzmark{cjfix5}&&&&&& \\
&&&&&&&&&&&&$\boldsymbol{\mathsf{e^{(2)}}}$ \\
&&&&&&&&&&&& $\mathsf{a^{(2)}_E}$ \\
&&&&&&&&&&&&\tikzmark{cjfix4} $\$$ \tikzmark{djfix3} \\
&&&&&&&&&&&&$\boldsymbol{\mathsf{e^{(3)}}}$ \\
&&&&&&&&&&&&$\mathsf{a^{(3)}_{E}}$ \\
&&&&&&&&&&&&\tikzmark{djfix4} $\$$ \\
&&&&&& $\boldsymbol{\mathsf{e^{(3)}}}$ &&&&&& \\
&&&&&& $\Lbag^{\underline{0}} \Rbag^{\underline{0}} \hbar$ &&&&&& \\
&&&&&& $\mathsf{a^{(3)}_{EW}}$ &&&&&& \\
&&&&&& $\$$ \tikzmark{djfix5}&&&&&& \\
$\boldsymbol{\mathsf{e^{(4)}}}$&&&&&&&&&&&& \\
$\Lbag^{\underline{0}} \Rbag^{\underline{0}} \hbar \Lbag^{\underline{0}} \boldsymbol{\mathsf{f}}^+ \Rbag^{\underline{0}} \hbar$ &&&&&&&&&&&& \\
$\mathsf{a^{(4)}_{WW}}$ &&&&&&&&&&&& \\
$\$$ \tikzmark{djfix6}&&&&&&&&&&&& \\
&&&&&& $\boldsymbol{\mathsf{e^{(4)}}}$ &&&&&& \\
&&&&&& $\Lbag^{\underline{0}} \Lbag^{\underline{1}} \Rbag^{\underline{1}} \hbar \Lbag^{\underline{1}} \boldsymbol{\mathsf{f}}^{++} \Rbag^{\underline{1}} \Rbag^{\underline{0}} \hbar$ &&&&&& \\
&&&&&& $\mathsf{a^{(4)}_{EW}}$ &&&&&& \\
&&&&&&\tikzmark{cjfix8} $\$$ \tikzmark{djfix7}&&&&&& 
\end{tabular}
\begin{tikzpicture}[overlay, remember picture, yshift=.25\baselineskip]
    \draw [->] ({pic cs:djfix1}) to ({pic cs:cjfix1});
    \draw [->] ({pic cs:djfix2}) [bend left] to ({pic cs:cjfix2});
    \draw [->] ({pic cs:djfix3}) [bend right] to ({pic cs:cjfix3});
    \draw [->] ({pic cs:djfix4}) [bend left] to ({pic cs:cjfix4});
    \draw [->] ({pic cs:djfix5}) [bend right] to ({pic cs:cjfix5});
    \draw [->] ({pic cs:djfix6}) to ({pic cs:cjfix2});
    \draw [->] ({pic cs:djfix7}) to ({pic cs:cjfix1}); 
  \end{tikzpicture}
  \caption{Some play by the fixed-point strategy $\mathit{fix}_A$ on the j-pointing tape.}
  \label{FigFixOnTape}
\end{center}
\end{figure}

Then, perhaps surprisingly, the fixed-point strategy $\mathit{fix}_A : (A \Rightarrow A) \Rightarrow A$ on any game $A$ is also JPA-computable almost in the same manner as $\mathit{der}_A$; recall Definition~\ref{DefFixedPointStrategies} and see Figure~\ref{FigFixOnTape}.
The only non-trivial point is the calculation of extended outer tags on the domain $\oc (A \Rightarrow A)$.
However, e.g., see the O-move $\$ \mathsf{a^{(4)}_{WW}} \hbar \Lbag^{\underline{0}} \Rbag^{\underline{0}} \hbar \Lbag^{\underline{0}} \boldsymbol{\mathsf{f}}^+ \Rbag^{\underline{0}} \hbar \boldsymbol{\mathsf{e^{(4)}}}$ in Figure~\ref{FigFixOnTape}, in which every occurrence of $\Lbag$ or $\Rbag$ in $\boldsymbol{\mathsf{f}}^+$ has depth $\geqslant 1$; hence, it is just straightforward for the j-pushdown automaton to compute the next P-move $\$ \mathsf{a^{(4)}_{EW}} \hbar \Lbag^{\underline{0}}  \Lbag^{\underline{1}} \Rbag^{\underline{1}} \hbar \Lbag^{\underline{1}} \boldsymbol{\mathsf{f}}^{++} \Rbag^{\underline{1}} \Rbag^{\underline{0}} \hbar \boldsymbol{\mathsf{e^{(4)}}}$.


At this point, it is easy to see that the remaining `atomic' strategies of the set $\mathcal{DPCF}$ (Definition~\ref{DefStrategiesForPCF}) are all JPA-computable, and so we leave the details to the reader. 
It remains to show that JPA-computability is preserved under currying, pairing, promotion and concatenation of strategies. 
Currying is trivial for it suffices to modify computation of inner tags appropriately.
Pairing and concatenation are just straightforward for the required j-pushdown automata can be obtained essentially as the `disjoint union' of the j-pushdown automata that realize the respective component strategies.\footnote{It is possible essentially because the P-view of a position of a pairing $\langle L, R \rangle$ (resp. a concatenation $J \ddagger K$) is the P-view of a position of $L$ or $R$ (resp. $J$ and/or $K$) up to inner tags, which is easy to verify \cite{yamada2016dynamic}.} 

Finally, we consider promotion.
Let us focus on normalized strategies, say, $\phi : \oc A \multimap \oc B$, where $A$ and $B$ are normalized games, because for non-normalized ones we only need a trivial extension.
Assume $\phi$ is JPA-computable, i.e., there is a j-pushdown automaton $\mathscr{A}_\phi$ that realizes $\phi$.
Suppose $\phi$ plays as in Figure~\ref{FigPhiOnTape}.
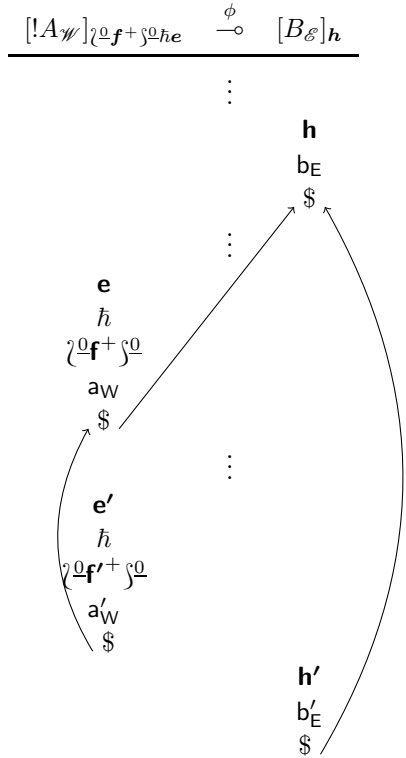
\begin{figure}[H]
\begin{center}
\begin{tabular}{ccc}
$[\oc A_{\mathscr{W}}]_{\Lbag^{\underline{0}} \boldsymbol{f}^+ \Rbag^{\underline{0}} \hbar \boldsymbol{e}}$ & $\stackrel{\phi}{\multimap}$ & $[B_{\mathscr{E}}]_{\boldsymbol{h}}$ \\ \hline 
&$\vdots$& \\
&& $\boldsymbol{\mathsf{h}}$ \\
&& $\mathsf{b_E}$ \\
&& \tikzmark{cjpromotion4} $\$$ \tikzmark{cjpromotion3} \\
&$\vdots$& \\
$\boldsymbol{\mathsf{e}}$ && \\
$\hbar$ && \\
$\Lbag^{\underline{0}} \boldsymbol{\mathsf{f}}^+ \Rbag^{\underline{0}}$ && \\
$\mathsf{a_W}$ && \\
\tikzmark{cjpromotion2} $\$$ \tikzmark{djpromotion4} && \\
& $\vdots$ & \\
$\boldsymbol{\mathsf{e'}}$ && \\
$\hbar$ && \\
$\Lbag^{\underline{0}} \boldsymbol{\mathsf{f'}}^+ \Rbag^{\underline{0}}$&& \\
$\mathsf{a'_W}$ \\
\tikzmark{djpromotion2} $\$$ && \\
&&$\boldsymbol{\mathsf{h'}}$\\
&&$\mathsf{b'_E}$ \\
&&$\$$ \tikzmark{djpromotion3} 
\end{tabular}
\begin{tikzpicture}[overlay, remember picture, yshift=.25\baselineskip]
\draw [->] ({pic cs:djpromotion2}) [bend left] to ({pic cs:cjpromotion2});
\draw [->] ({pic cs:djpromotion3}) [bend right] to ({pic cs:cjpromotion3});
\draw [->] ({pic cs:djpromotion4}) to ({pic cs:cjpromotion4});
\end{tikzpicture}
\end{center}
\caption{A computation of $\phi : \oc A \multimap B$ on the j-pointing tape.}
\label{FigPhiOnTape}
\end{figure}
Then, $\phi^\dagger$ plays as in Figure~\ref{FigPhiDaggerOnTape}, where $\boldsymbol{g} \in \mathcal{T}$ is arbitrarily chosen by O.
\begin{figure}[H]
\begin{center}
\begin{tabular}{ccc}
$[\oc A_{\mathscr{W}}]_{\Lbag^{\underline{0}} \Lbag^{\underline{1}} \boldsymbol{g}^{++} \Rbag^{\underline{1}} \hbar \Lbag^{\underline{1}} \boldsymbol{f}^{++} \Rbag^{\underline{1}} \Rbag^{\underline{0}} \hbar \boldsymbol{e}}$ & $\stackrel{\phi^\dagger}{\multimap}$ & $[\oc B_{\mathscr{E}}]_{\Lbag^{\underline{0}} \boldsymbol{g}^+ \Rbag^{\underline{0}} \hbar \boldsymbol{h}}$ \\ \hline 
&$\vdots$& \\
&& $\boldsymbol{\mathsf{h}}$ \\
&& $\hbar$ \\
&& $\Lbag^{\underline{0}} \boldsymbol{\mathsf{g}}^+ \Rbag^{\underline{0}}$ \\
&& $\mathsf{b_E}$ \\
&& \tikzmark{cjpromotion100} $\$$ \tikzmark{cjpromotiond3} \\
&$\vdots$& \\
$\boldsymbol{\mathsf{e}}$&&\\
$\hbar$ && \\
$\Lbag^{\underline{1}} \boldsymbol{\mathsf{f}}^{++} \Rbag^{\underline{1}} \Rbag^{\underline{0}}$ \\
$\hbar$&& \\
$\Lbag^{\underline{0}}  \Lbag^{\underline{1}} \boldsymbol{\mathsf{g}}^{++} \Rbag^{\underline{1}}$ && \\
$\mathsf{a_W}$ && \\
\tikzmark{cjpromotiond2} $\$$ \tikzmark{djpromotion100} && \\
&$\vdots$ & \\
$\boldsymbol{\mathsf{e'}}$&&\\
$\hbar$ && \\
$\Lbag^{\underline{1}} \boldsymbol{\mathsf{f'}}^{++} \Rbag^{\underline{1}} \Rbag^{\underline{0}}$ \\
$\hbar$&& \\
$\Lbag^{\underline{0}} \Lbag^{\underline{1}} \boldsymbol{\mathsf{g}}^{++} \Rbag^{\underline{1}}$ && \\
$\mathsf{a'_W}$ && \\ 
\tikzmark{djpromotiond2} $\$$ && \\
&& $\boldsymbol{\mathsf{h'}}$ \\
&& $\hbar$ \\
&& $\Lbag^{\underline{0}} \boldsymbol{\mathsf{g}}^+ \Rbag^{\underline{0}}$ \\
&& $\mathsf{b'_E}$ \\
&&$\$$ \tikzmark{djpromotiond3} 
\end{tabular}
\begin{tikzpicture}[overlay, remember picture, yshift=.25\baselineskip]
\draw [->] ({pic cs:djpromotiond2}) [bend left] to ({pic cs:cjpromotiond2});
\draw [->] ({pic cs:djpromotiond3}) [bend right] to ({pic cs:cjpromotiond3});
\draw [->] ({pic cs:djpromotion100}) to ({pic cs:cjpromotion100});
\end{tikzpicture}
\caption{A computation of $\phi^\dagger : \oc A \multimap \oc B$.}
\label{FigPhiDaggerOnTape}
\end{center}
\end{figure}

Similarly to the case of fixed-point strategies, it is then easy to construct from $\mathscr{A}_\phi$ a j-pushdown automaton $\mathscr{A}_{\phi^\dagger}$ that realizes $\phi^\dagger$, completing the proof.
\end{proof}

Let us remark again that j-pushdown automata only compute the next P-move into the stack, not the entire play onto the tape.
In other words, they need to interact with O (and J) for PCF-completeness.
Hence, it seems that their computational power comes from the interactive nature of their computation.
However, the answer is not completely `yes' because the \emph{restriction} on the cells which j-pushdown automata may move to is actually a key point as well; see the next section.

\subsection{Stand-alone Turing completeness of j-stack automata}
\label{StandAloneTuringCompleteness}
Let us point out that instead of j-pushdown automata we may employ deterministic, non-erasing \emph{stack automata} \cite{ginsburg1967stack,hopcroft1967nonerasing} such that the stack is equipped with edges similarly to j-pointing tapes, and they can access only the stack cells that correspond to P-views, where positions of a game are recorded in the stack (n.b., the input tape is not used at all).

We call such restricted stack automata \emph{\bfseries j-stack automata} (n.b., we do not have to define them in detail here for Definition~\ref{DefJPushdownAutomata} implies clearly enough what they are).
Applying the proof of Theorem~\ref{ThmJ}, it is easy to see that j-stack automata are PCF-complete. 
Thus, this alternative approach would have certainly achieved the aim of Section~\ref{PCFCompletenessOfJPushdownAutomata} as well for j-stack automata are also strictly weaker than TMs \cite{hopcroft1967nonerasing}.
\begin{proposition}[Weakness of j-stack automata]
\label{PropStrictWeaknessSecond}
There is a formal language that can be recognized by a TM but not by any j-stack automaton. 
\end{proposition}

In addition, there is a conceptual advantage of j-stack automata over j-pushdown automata: For j-pushdown automata, we have to assume that each stack content (representing the next P-move) is \emph{automatically} copied onto the input tape, say, by J (i.e., the game-semantic compromise), while it is not the case at all for j-stack automata because their computation (with O) always occurs in the stack, where O outputs the next O-move into the stack too (i.e., we may dispense with the input tape and J). 

On the other hand, pushdown automata are a priori more restricted than stack automata, which is the main reason why we have employed j-pushdown automata, rather than j-stack automata, for the main theorem (Theorem~\ref{ThmJ}). 

Nevertheless, given an input in the stack, no interaction with O is necessary for the game-semantic \emph{first-order} or \emph{classical} computation of j-stack automata (n.b., recall that P, not O, computes internal O-moves occurring in a position of a game, which is possible thanks to the axiom \textsc{Dum} in Definition~\ref{DefLegalPositions}).
Hence, they are \emph{Turing complete}, where (given an input in the stack) they compute \emph{without any interaction with O (or J)}, which does not hold for j-pushdown automata.
Let us summarize the argument as:
\begin{corollary}[Main corollary]
\label{CoroT}
Deterministic, non-erasing stack automata, equipped with at most one directed (downwards) edge between each pair of stack cells and some restriction on the stack cells that the automata may access, are Turing complete, where the automata never interact with another computational agent. 
\end{corollary}

As remarked before, the corollary in particular shows that Turing completeness of the game-semantic approach does not come from its interaction with O, which rather contributes to its higher-order aspect (or the path from Turing completeness to PCF-completeness), but rather from the additional edges on the input tape or the stack, which serve as a `route recorder' that indicates where necessary information exists. 
Let us also remark that by the same argument as the corollary we may show that j-pushdown automata with interaction with J are Turing complete without interaction with O. 

To summarize, we have shown:
\begin{center}
\begin{tabular}{| c | c |} \hline
{\bfseries Turing completeness} & {\bfseries PCF-completeness} \\ \hline
J-pushdown automata (JPAs) with J & JPAs with J \& O \\ 
J-stack automata (JSAs) & JSAs with O \\ \hline
\end{tabular}
\end{center}

\section{Conclusion and future work}
\label{ConclusionAndFutureWork}
The present work has revisited the game-semantic model of higher-order computation and reformulated it in terms of automata, or equivalently, brought the game-semantic framework into automata theory.
The resulting approach is novel and has established somewhat surprising consequences: PCF-completeness of j-pushdown automata (Theorem~\ref{ThmJ}) and Turing completeness of j-stack automata (Corollary~\ref{CoroT}).
Theorem~\ref{ThmJ} demonstrates the power of combining the interactive computation and the restriction of automata by edges, where the restriction \emph{saves} certain computation in a novel way, and Corollary~\ref{CoroT} deepens the result further by showing that only the restriction is actually enough for Turing completeness of automata that are strictly weaker than TMs, and the interaction (with O) rather contributes to higher-order nature of computation. 

From a methodological point, the present work has demonstrated high potential of the game-semantic method for automata theory. 
That is, assuming the setting that has been specific to game semantics, we may obtain novel, non-trivial results for automata theory. 

As future work, it would be interesting to identify an automata-theoretic lower bound of the game-semantic PCF- or Turing completeness, i.e., the least powerful automata that are PCF- or Turing complete in the game-semantic framework.
More generally, we are interested in a correspondence between automata in the game-semantic framework and formal languages; as the present work indicates, it would form a new hierarchy, which is different from the well-established \emph{Chomsky hierarchy} \cite{chomsky1956three,sipser2006introduction}. 
Finally, it would be fruitful to formulate computational complexity theory \cite{kozen2006theory} by combining automata theory and the game-semantic model of computation as in the present work; for instance, it might be possible to define computational complexity of strategies \emph{relative} to that of oracle computation (i.e., computation by O), which would be an accurate measure for computational complexity of higher-order computation.

\section*{Acknowledgments}
The author acknowleges the financial support from Funai Overseas Scholarship, and also he is grateful to Samson Abramsky, Luke Ong and John Longley for fruitful discussions.

\bibliographystyle{apalike}
\bibliography{CategoricalLogic,GamesAndStrategies,RecursionTheory,PCF,TypeTheoriesAndProgrammingLanguages,HoTT,GoI,LinearLogic}

\begin{thebibliography}{}

\bibitem[Abramsky et~al., 1997]{abramsky1997semantics}
Abramsky, S. et~al. (1997).
\newblock Semantics of interaction: An introduction to game semantics.
\newblock {\em Semantics and Logics of Computation}, 14:1--31.

\bibitem[Abramsky and Jagadeesan, 1994]{abramsky1994games}
Abramsky, S. and Jagadeesan, R. (1994).
\newblock Games and full completeness for multiplicative linear logic.
\newblock {\em The Journal of Symbolic Logic}, 59(02):543--574.

\bibitem[Abramsky and McCusker, 1999]{abramsky1999game}
Abramsky, S. and McCusker, G. (1999).
\newblock Game semantics.
\newblock In {\em Computational Logic: Proceedings of the 1997 Marktoberdorf
  Summer School}, pages 1--55, Berlin, Heidelberg. Springer.

\bibitem[Abramsky and Mellies, 1999]{abramsky1999concurrent}
Abramsky, S. and Mellies, P.-A. (1999).
\newblock Concurrent games and full completeness.
\newblock In {\em Proceedings of the 14th Symposium on Logic in Computer
  Science}, pages 431--442. IEEE.

\bibitem[Amadio and Curien, 1998]{amadio1998domains}
Amadio, R.~M. and Curien, P.-L. (1998).
\newblock {\em Domains and Lambda-Calculi}.
\newblock Number~46. Cambridge University Press, Cambridge.

\bibitem[Chomsky, 1956]{chomsky1956three}
Chomsky, N. (1956).
\newblock Three models for the description of language.
\newblock {\em IRE Transactions on Information Theory}, 2(3):113--124.

\bibitem[Cutland, 1980]{cutland1980computability}
Cutland, N. (1980).
\newblock {\em Computability: An Introduction to Recursive Function Theory}.
\newblock Cambridge University Press, Cambridge.

\bibitem[Ginsburg et~al., 1967]{ginsburg1967stack}
Ginsburg, S., Greibach, S.~A., and Harrison, M.~A. (1967).
\newblock Stack automata and compiling.
\newblock {\em Journal of the ACM (JACM)}, 14(1):172--201.

\bibitem[Girard, 1987]{girard1987linear}
Girard, J.-Y. (1987).
\newblock Linear logic.
\newblock {\em Theoretical computer science}, 50(1):1--101.

\bibitem[Gunter, 1992]{gunter1992semantics}
Gunter, C.~A. (1992).
\newblock {\em Semantics of Programming Languages: Structures and Techniques}.
\newblock MIT press, Cambridge, MA.

\bibitem[Hague et~al., 2008]{hague2008collapsible}
Hague, M., Murawski, A.~S., Ong, C.-H., and Serre, O. (2008).
\newblock Collapsible pushdown automata and recursion schemes.
\newblock In {\em Proceedings of the 23rd Annual IEEE Symposium on Logic in
  Computer Science}, pages 452--461. IEEE.

\bibitem[Hopcroft et~al., 1979]{hopcroft1979introduction}
Hopcroft, J.~E., Motwani, R., and Ullman, J.~D. (1979).
\newblock {\em Introduction to Automata Theory, Languages, and Computation}.
\newblock Addison-Wesley, Reading, MA.

\bibitem[Hopcroft and Ullman, 1967]{hopcroft1967nonerasing}
Hopcroft, J.~E. and Ullman, J.~D. (1967).
\newblock Nonerasing stack automata.
\newblock {\em Journal of Computer and System Sciences}, 1(2):166--186.

\bibitem[Hyland and Ong, 2000]{hyland2000full}
Hyland, J. M.~E. and Ong, C.-H. (2000).
\newblock On {F}ull {A}bstraction for {PCF}: {I}, {II}, and {III}.
\newblock {\em Information and computation}, 163(2):285--408.

\bibitem[Hyland, 1997]{hyland1997game}
Hyland, M. (1997).
\newblock Game semantics.
\newblock In {\em Semantics and Logics of Computation}, volume~14, page 131.
  Cambridge University Press, New York.

\bibitem[Kleene, 1952]{kleene1952introduction}
Kleene, S.~C. (1952).
\newblock Introduction to metamathematics.

\bibitem[Knapik et~al., 2002]{knapik2002higher}
Knapik, T., Niwi{\'n}ski, D., and Urzyczyn, P. (2002).
\newblock Higher-order pushdown trees are easy.
\newblock In {\em International Conference on Foundations of Software Science
  and Computation Structures}, pages 205--222. Springer.

\bibitem[Kozen, 2006]{kozen2006theory}
Kozen, D.~C. (2006).
\newblock {\em Theory of Computation}.
\newblock Springer Science \& Business Media, London.

\bibitem[Kozen, 2012]{kozen2012automata}
Kozen, D.~C. (2012).
\newblock {\em Automata and Computability}.
\newblock Springer Science \& Business Media, New York.

\bibitem[Laurent, 2002]{laurent2002polarized}
Laurent, O. (2002).
\newblock Polarized games.
\newblock In {\em Proceedings of the 17th Annual IEEE Symposium on Logic in
  Computer Science}, pages 265--274. IEEE.

\bibitem[Longley and Normann, 2015]{longley2015higher}
Longley, J. and Normann, D. (2015).
\newblock {\em Higher-Order Computability}.
\newblock Springer, Heidelberg.

\bibitem[Plotkin, 1977]{plotkin1977lcf}
Plotkin, G.~D. (1977).
\newblock {LCF} considered as a programming language.
\newblock {\em Theoretical computer science}, 5(3):223--255.

\bibitem[Scott, 1993]{scott1993type}
Scott, D.~S. (1993).
\newblock A type-theoretical alternative to {ISWIM}, {CUCH}, {OWHY}.
\newblock {\em Theoretical Computer Science}, 121(1):411--440.

\bibitem[Sipser, 2006]{sipser2006introduction}
Sipser, M. (2006).
\newblock {\em Introduction to the Theory of Computation}, volume~2.
\newblock Thomson Course Technology, Boston.

\bibitem[Turing, 1937]{turing1937computable}
Turing, A.~M. (1937).
\newblock On computable numbers, with an application to the
  {E}ntscheidungsproblem.
\newblock {\em Proceedings of the London Mathematical Society}, 2(1):230--265.

\bibitem[Yamada, 2019]{yamada2019game}
Yamada, N. (2019).
\newblock A game-semantic model of computation.
\newblock {\em Research in the Mathematical Sciences}, 6(1):3.

\bibitem[Yamada and Abramsky, 2016]{yamada2016dynamic}
Yamada, N. and Abramsky, S. (2016).
\newblock Dynamic games and strategies.
\newblock {\em arXiv preprint arXiv:1601.04147}.

\end{thebibliography}

\appendix
\section{Constructions on games}
\label{AppendixConstructionsOnGames}
In this section, we present the formal definitions of standard constructions on games given in the previous work \cite{yamada2019game}.

\subsection{Concatenation of games \cite{yamada2019game}}
\label{DefConcatenationOfGames}
The \emph{\bfseries concatenation} $J \ddagger K$ of games $J$ and $K$ such that $\mathcal{H}^\omega(J) \trianglelefteqslant A \multimap B$ and $\mathcal{H}^\omega(K) \trianglelefteqslant B \multimap C$ for some normalized games $A$, $B$ and $C$ is given by:
\begin{itemize}

\item $M_{J \ddagger K} \stackrel{\mathrm{df. }}{=} \{ [(a, \mathscr{W})]_{\boldsymbol{e}} \mid [(a, \mathscr{W})]_{\boldsymbol{e}} \in M_J^{\mathsf{Ext}}, [a]_{\boldsymbol{e}} \in M_A \ \! \} \\ \cup \{ [(c, \mathscr{E})]_{\boldsymbol{f}} \mid [(c, \mathscr{E})]_{\boldsymbol{f}} \in M_K^{\mathsf{Ext}}, [c]_{\boldsymbol{f}} \in M_C \ \! \} \\ \cup \{ [((b, \mathscr{E}), \mathscr{S})]_{\boldsymbol{g}} \mid [(b, \mathscr{E})]_{\boldsymbol{g}} \in M_J^{\mathsf{Ext}}, [b]_{\boldsymbol{g}} \in M_B \ \! \} \\ \cup \{ [((b, \mathscr{W}), \mathscr{N})]_{\boldsymbol{g}} \mid [(b, \mathscr{W})]_{\boldsymbol{g}} \in M_K^{\mathsf{Ext}}, [b]_{\boldsymbol{g}} \in M_B \ \! \} \\ \cup \{ [(m, \mathscr{S})]_{\boldsymbol{l}} \mid [m]_{\boldsymbol{l}} \in M_J^{\mathsf{Int}} \ \! \} \cup \{ [(n, \mathscr{N})]_{\boldsymbol{r}} \mid [n]_{\boldsymbol{r}} \in M_K^{\mathsf{Int}} \ \! \}$;

\item $M_{J \ddagger K}^{\mathsf{Init}} \stackrel{\mathrm{df. }}{=}\{ [(c, \mathscr{E})]_{\boldsymbol{f}} \mid [(c, \mathscr{E})]_{\boldsymbol{f}} \in M_K^{\mathsf{Ext}}, [c]_{\boldsymbol{f}} \in M_C^{\mathsf{Init}} \ \! \}$;

\item $\lambda_{J \ddagger K}([(m, X)]_{\boldsymbol{e}}) \stackrel{\mathrm{df. }}{=} \begin{cases} (\lambda_J^{\mathsf{OP}}([m]_{\boldsymbol{e}}), \mathsf{I}) &\text{if $X = \mathscr{S} \wedge \exists [b]_{\boldsymbol{e}} \in M_B . \ \! [m]_{\boldsymbol{e}} = [(b, \mathscr{E})]_{\boldsymbol{e}} \in M_J^{\mathsf{Ext}}$;} \\ \lambda_J([m]_{\boldsymbol{e}}) &\text{if $X = \mathscr{W} \vee (X = \mathscr{S} \wedge [m]_{\boldsymbol{e}} \in M_J^{\mathsf{Int}})$;} \\ (\lambda_K^{\mathsf{OP}}([m]_{\boldsymbol{e}}), \mathsf{I}) &\text{if $X = \mathscr{N} \wedge \exists [b]_{\boldsymbol{e}} \in M_B . \ \! [m]_{\boldsymbol{e}} = [(b, \mathscr{W})]_{\boldsymbol{e}} \in M_K^{\mathsf{Ext}}$;} \\ \lambda_K([m]_{\boldsymbol{e}}) &\text{if $X = \mathscr{E} \vee (X = \mathscr{N} \wedge [m]_{\boldsymbol{e}} \in M_K^{\mathsf{Int}})$;} \end{cases}$

\item $\Delta_{J \ddagger K} ([(m, X)]_{\boldsymbol{e}}) \stackrel{\mathrm{df. }}{=} \begin{cases} [(m', \mathscr{S})]_{\boldsymbol{e}} &\text{if $X = \mathscr{S}$ and $\Delta_J([m]_{\boldsymbol{e}}) = [m']_{\boldsymbol{e}}$;} \\ [(m'', \mathscr{N})]_{\boldsymbol{e}} &\text{if $X = \mathscr{N}$ and $\Delta_K([m]_{\boldsymbol{e}}) = [m'']_{\boldsymbol{e}}$;} \\ [((b, \mathscr{W}), \mathscr{N})]_{\boldsymbol{e}} &\text{if $X = \mathscr{S}$, $\Delta_J([m]_{\boldsymbol{e}}) \uparrow$ and $m = (b, \mathscr{E})$;} \\ [((b, \mathscr{E}), \mathscr{S})]_{\boldsymbol{e}} &\text{if $X = \mathscr{N}$, $\Delta_K([m]_{\boldsymbol{e}}) \uparrow$ and $m = (b, \mathscr{W})$;} \end{cases}$

\item $P_{J \ddagger K} \stackrel{\mathrm{df. }}{=} \{ \boldsymbol{s} \in \mathscr{J}_{J \ddagger K} \mid \boldsymbol{s} \upharpoonright J \in P_J, \boldsymbol{s} \upharpoonright K \in P_K, \boldsymbol{s} \upharpoonright B^{[0]}, B^{[1]} \in \mathit{pr}_B \ \! \}$, where the map $\mathit{att}_{J \ddagger K} : M_{J \ddagger K} \to \{ J, K \}$ is given by $[(a, \mathscr{W})]_{\boldsymbol{e}} \mapsto J$, $[(m, \mathscr{S})]_{\boldsymbol{l}} \mapsto J$, $[((b, \mathscr{E}), \mathscr{S})]_{\boldsymbol{g}} \mapsto J$, $[(c, \mathscr{E})]_{\boldsymbol{f}} \mapsto K$, $[(n, \mathscr{N})]_{\boldsymbol{r}} \mapsto K$, $[((b, \mathscr{W}), \mathscr{N})]_{\boldsymbol{g}} \mapsto K$, and the map $\mathit{peel}_{J \ddagger K} : M_{J \ddagger K} \to M_J \cup M_K$ by $[(a, \mathscr{W})]_{\boldsymbol{e}} \mapsto [(a, \mathscr{W})]_{\boldsymbol{e}}$, $[(c, \mathscr{E})]_{\boldsymbol{f}} \mapsto [(c, \mathscr{E})]_{\boldsymbol{f}}$, $[((b, \mathscr{E}), \mathscr{S})]_{\boldsymbol{g}} \mapsto [(b, \mathscr{E})]_{\boldsymbol{g}}$, $[((b, \mathscr{W}), \mathscr{N})]_{\boldsymbol{g}} \mapsto [(b, \mathscr{W})]_{\boldsymbol{g}}$, $[(m, \mathscr{S})]_{\boldsymbol{l}} \mapsto [m]_{\boldsymbol{l}}$, $[(n, \mathscr{N})]_{\boldsymbol{r}} \mapsto [n]_{\boldsymbol{r}}$, $\boldsymbol{s} \upharpoonright J$ (resp. $\boldsymbol{s} \upharpoonright K$) is the j-subsequence of $\boldsymbol{s}$ that consists of moves $m$ such that $\mathit{att}_{J \ddagger K}(m) = J$ (resp. $\mathit{att}_{J \ddagger K}(m) = K$) yet changed into $\mathit{peel}_{J \ddagger K}(m)$, $\boldsymbol{s} \upharpoonright B^{[0]}, B^{[1]}$ is the j-subsequence of $\boldsymbol{s}$ that consists of moves $[((b, X), Y)]_{\boldsymbol{e}}$ such that $[b]_{\boldsymbol{e}} \in M_B \wedge ((X = \mathscr{E} \wedge Y = \mathscr{S}) \vee (X = \mathscr{W} \wedge Y = \mathscr{N}))$ yet changed into $[(b, \overline{X})]_{\boldsymbol{e}}$, for which $\overline{\mathscr{E}} \stackrel{\mathrm{df. }}{=} \mathscr{W}$ and $\overline{\mathscr{W}} \stackrel{\mathrm{df. }}{=} \mathscr{E}$, and $\mathit{pr}_B \stackrel{\mathrm{df. }}{=} \mathsf{Pref}(\mathit{cp}_B)$.

\end{itemize}

\subsection{Pairing of games \cite{yamada2019game}}
\label{DefPairingOfGames}
The \emph{\bfseries pairing} $\langle L, R \rangle$ of games $L$ and $R$ such that $\mathcal{H}^\omega(L) \trianglelefteqslant C \multimap A$ and $\mathcal{H}^\omega(R) \trianglelefteqslant C \multimap B$ for any normalized games $A$, $B$ and $C$ is given by:
\begin{itemize}

\item $M_{\langle L, R \rangle} \stackrel{\mathrm{df. }}{=} \{ [(c, \mathscr{W})]_{\boldsymbol{e}} \mid [(c, \mathscr{W})]_{\boldsymbol{e}} \in M_L^{\mathsf{Ext}} \cup M_R^{\mathsf{Ext}}, [c]_{\boldsymbol{e}} \in M_C \ \! \} \\ \cup \{ [((a, \mathscr{W}), \mathscr{E})]_{\boldsymbol{f}} \mid [(a, \mathscr{E})]_{\boldsymbol{f}} \in M_L^{\mathsf{Ext}}, [a]_{\boldsymbol{f}} \in M_A \ \! \} \\ \cup \{ [((b, \mathscr{E}), \mathscr{E})]_{\boldsymbol{g}} \mid [(b, \mathscr{E})]_{\boldsymbol{g}} \in M_R^{\mathsf{Ext}}, [b]_{\boldsymbol{g}} \in M_B \ \! \} \\ \cup \{ [(l, \mathscr{S})]_{\boldsymbol{h}} \mid [l]_{\boldsymbol{h}} \in M_L^{\mathsf{Int}} \ \! \} \cup \{ [(r, \mathscr{N})]_{\boldsymbol{k}} \mid [r]_{\boldsymbol{k}} \in M_R^{\mathsf{Int}} \ \! \}$;

\item $M_{\langle L, R \rangle}^{\mathsf{Init}} \stackrel{\mathrm{df. }}{=} \{ m \in M_{\langle L, R \rangle} \mid \mathit{peel}_{\langle L, R \rangle}(m) \in M_{\mathit{att}_{\langle L, R \rangle}(m)}^{\mathsf{Init}} \ \! \}$, where the map $\mathit{peel}_{\langle L, R \rangle} : M_{\langle L, R \rangle} \to M_L \cup M_R$ is given by $[(c, \mathscr{W})]_{\boldsymbol{e}} \mapsto [(c, \mathscr{W})]_{\boldsymbol{e}}$, $[((a, \mathscr{W}), \mathscr{E})]_{\boldsymbol{f}} \mapsto [(a, \mathscr{E})]_{\boldsymbol{f}}$, $[((b, \mathscr{E}), \mathscr{E})]_{\boldsymbol{g}} \mapsto [(b, \mathscr{E})]_{\boldsymbol{g}}$, $[(l, \mathscr{S})]_{\boldsymbol{h}} \mapsto [l]_{\boldsymbol{h}}$, $[(r, \mathscr{N})]_{\boldsymbol{k}} \mapsto [r]_{\boldsymbol{k}}$, and the map $\mathit{att}_{\langle L, R \rangle} : M_{\langle L, R \rangle} \to \{ L, R, C \}$ by $[(c, \mathscr{W})]_{\boldsymbol{e}} \mapsto C$, $[((a, \mathscr{W}), \mathscr{E})]_{\boldsymbol{f}} \mapsto L$, $[((b, \mathscr{E}), \mathscr{E})]_{\boldsymbol{g}} \mapsto R$, $[(l, \mathscr{S})]_{\boldsymbol{h}} \mapsto L$, $[(r, \mathscr{N})]_{\boldsymbol{k}} \mapsto R$;

\item $\lambda_{\langle L, R \rangle} ([(m, X)]_{\boldsymbol{e}}) \stackrel{\mathrm{df. }}{=} \begin{cases} \overline{\lambda_C} ([m]_{\boldsymbol{e}}) &\text{if $X = \mathscr{W}$;} \\ \lambda_A([a]_{\boldsymbol{e}}) &\text{if $X = \mathscr{E}$ and $m$ is of the form $(a, \mathscr{W})$;} \\ \lambda_B([b]_{\boldsymbol{e}}) &\text{if $X = \mathscr{E}$ and $m$ is of the form $(b, \mathscr{E})$;} \\ \lambda_L([m]_{\boldsymbol{e}}) &\text{if $X = \mathscr{S}$;} \\ \lambda_R([m]_{\boldsymbol{e}}) &\text{if $X = \mathscr{N}$;} \end{cases}$

\item $\Delta_{\langle L, R \rangle} ([(m, X)]_{\boldsymbol{e}}) \stackrel{\mathrm{df. }}{=} \begin{cases} [(l', \mathscr{S})]_{\boldsymbol{e}} &\text{if $X = \mathscr{S}$, where $\Delta_L([l]_{\boldsymbol{e}}) = [l']_{\boldsymbol{e}}$;} \\ [(r', \mathscr{N})]_{\boldsymbol{e}} &\text{if $X = \mathscr{N}$, where $\Delta_R([r]_{\boldsymbol{e}}) = [r']_{\boldsymbol{e}}$;} \end{cases}$

\item $P_{\langle L, R \rangle} \stackrel{\mathrm{df. }}{=} \{ \boldsymbol{s} \in \mathscr{L}_{\langle L, R \rangle} \mid (\boldsymbol{s} \upharpoonright L \in P_L \wedge \boldsymbol{s} \upharpoonright B = \boldsymbol{\epsilon}) \vee (\boldsymbol{s} \upharpoonright R \in P_R \wedge \boldsymbol{s} \upharpoonright A = \boldsymbol{\epsilon}) \ \! \}$, where $\boldsymbol{s} \upharpoonright L$ (resp. $\boldsymbol{s} \upharpoonright R$) is the j-subsequence of $\boldsymbol{s}$ that consists of moves $x$ such that $\mathit{peel}_{\langle L, R \rangle}(x) \in M_L$ (resp. $\mathit{peel}_{\langle L, R \rangle}(x) \in M_R$) changed into $\mathit{peel}_{\langle L, R \rangle}(x)$, and $\boldsymbol{s} \upharpoonright B$ (resp. $\boldsymbol{s} \upharpoonright A$) is the j-subsequence of $\boldsymbol{s}$ that consists of moves of the form $[((b, \mathscr{E}), \mathscr{E})]_{\boldsymbol{g}}$ with $[b]_{\boldsymbol{g}} \in M_B$ (resp. $[((a_0, \mathscr{W}), \mathscr{E})]_{\boldsymbol{f}}$ with $[a]_{\boldsymbol{g}} \in M_A$).

\end{itemize}


\subsection{Promotion of games \cite{yamada2019game}}
\label{DefPromotionOfGames}
Given a game $G$ such that $\mathcal{H}^\omega(G) \trianglelefteqslant A \Rightarrow B$ for any normalized games $A$ and $B$, the \emph{\bfseries promotion} $G^\dagger$ of $G$ is given by:
\begin{itemize}

\item $M_{G^\dagger} \stackrel{\mathrm{df. }}{=} \{ [(a, \mathscr{W})]_{\Lbag^{\underline{0}} \boldsymbol{f}^+ \Rbag^{\underline{0}} \hbar \boldsymbol{e}} \mid [(a, \mathscr{W})]_{\Lbag^{\underline{0}} \boldsymbol{f}^+ \Rbag^{\underline{0}} \hbar \boldsymbol{e}} \in M_G, [a]_{\boldsymbol{e}} \in M_A \ \! \} \\
\cup \{ [(b, \mathscr{E})]_{\Lbag^{\underline{0}} \boldsymbol{f}^+ \Rbag^{\underline{0}} \hbar \boldsymbol{e}} \mid [(b, \mathscr{E})]_{\boldsymbol{e}} \in M_G, [b]_{\boldsymbol{e}} \in M_B, \boldsymbol{f} \in \mathcal{T} \ \! \} \\ \cup \{ [(m, \mathscr{S})]_{\Lbag^{\underline{0}} \boldsymbol{f}^+ \Rbag^{\underline{0}} \hbar \boldsymbol{e}} \mid [m]_{\boldsymbol{e}} \in M_G^{\mathsf{Int}}, \boldsymbol{f} \in \mathcal{T} \ \! \}$;

\item $M_{G^\dagger}^{\mathsf{Init}} \stackrel{\mathrm{df. }}{=} \{ [(b, \mathscr{E})]_{\Lbag^{\underline{0}} \boldsymbol{f}^+ \Rbag^{\underline{0}} \hbar \boldsymbol{e}} \mid [(b, \mathscr{E})]_{\boldsymbol{e}} \in M_G, [b]_{\boldsymbol{e}} \in M_B^{\mathsf{Init}}, \boldsymbol{f} \in \mathcal{T} \ \! \}$;

\item $\lambda_{G^\dagger} : [(a, \mathscr{W})]_{\Lbag^{\underline{0}} \boldsymbol{f}^+ \Rbag^{\underline{0}} \hbar \boldsymbol{e}} \mapsto \lambda_G([(a, \mathscr{W})]_{\Lbag^{\underline{0}} \boldsymbol{f}^+ \Rbag^{\underline{0}} \hbar \boldsymbol{e}})$, $[(b, \mathscr{E})]_{\Lbag^{\underline{0}} \boldsymbol{f}^+ \Rbag^{\underline{0}} \hbar \boldsymbol{e}} \mapsto \lambda_G([(b, \mathscr{E})]_{\boldsymbol{e}})$, $[(m, \mathscr{S})]_{\Lbag^{\underline{0}} \boldsymbol{f}^+ \Rbag^{\underline{0}} \hbar \boldsymbol{e}} \mapsto \lambda_G([m]_{\boldsymbol{e}})$;

\item $\Delta_{G^\dagger} : [(m, \mathscr{S})]_{\Lbag^{\underline{0}} \boldsymbol{f}^+ \Rbag^{\underline{0}} \hbar \boldsymbol{e}} \stackrel{\mathrm{df. }}{=} [(m', \mathscr{S})]_{\Lbag^{\underline{0}} \boldsymbol{f}^+ \Rbag^{\underline{0}} \hbar \boldsymbol{e}}$, where $\Delta_G([m]_{\boldsymbol{e}}) = [m']_{\boldsymbol{e}}$;

\item $P_{G^\dagger} \stackrel{\mathrm{df. }}{=} \{ \boldsymbol{s} \in \mathscr{L}_{G^\dagger} \mid \forall \boldsymbol{f} \in \mathcal{T} . \ \! \boldsymbol{s} \upharpoonright \boldsymbol{f} \in P_G \wedge (\boldsymbol{s} \upharpoonright \boldsymbol{f} \neq \boldsymbol{\epsilon} \Rightarrow \forall \boldsymbol{g} \in \mathcal{T} . \ \! \boldsymbol{s} \upharpoonright \boldsymbol{g} \neq \boldsymbol{\epsilon} \Rightarrow \mathit{ede}(\boldsymbol{f}) \neq \mathit{ede}(\boldsymbol{g})) \ \! \}$, where $\boldsymbol{s} \upharpoonright \boldsymbol{f}$ is the j-subsequence of $\boldsymbol{s}$ that consists of moves of the form $[(a, \mathscr{W})]_{\Lbag^{\underline{0}} \Lbag^{\underline{1}} \boldsymbol{f}^{++} \Rbag^{\underline{1}} \hbar \Lbag^{\underline{1}} \boldsymbol{g}^{++} \Rbag^{\underline{1}} \Rbag^{\underline{0}} \hbar \boldsymbol{e}}$ with $[a]_{\boldsymbol{e}} \in M_A$, $[(b, \mathscr{E})]_{\Lbag^{\underline{0}} \boldsymbol{f}^+ \Rbag^{\underline{0}} \hbar \boldsymbol{e}}$ with $[b]_{\boldsymbol{e}} \in M_B$ or $[(m, \mathscr{S})]_{\Lbag^{\underline{0}} \boldsymbol{f}^+ \Rbag^{\underline{0}} \hbar \boldsymbol{e}}$ with $[m]_{\boldsymbol{e}} \in M_G^{\mathsf{Int}}$ yet changed into $[(a, \mathscr{W})]_{\Lbag^{\underline{0}} \boldsymbol{g}^+ \Rbag^{\underline{0}} \hbar \boldsymbol{e}}$, $[(b, \mathscr{E})]_{\boldsymbol{e}}$ or $[m]_{\boldsymbol{e}}$, respectively.

\end{itemize}

\subsection{Currying of games \cite{yamada2019game}}
\label{DefCurryingOnGames}
For a game $G$ and normalized games $A$, $B$ and $C$ such that $\mathcal{H}^\omega(G) \trianglelefteqslant A \otimes B \multimap C$, the \emph{\bfseries currying} $\Lambda(G)$ of $G$ is given by:
\begin{itemize}

\item $M_{\Lambda(G)} \stackrel{\mathrm{df. }}{=} \{ [(a, \mathscr{W})]_{\boldsymbol{e}} \mid [((a, \mathscr{W}), \mathscr{W})]_{\boldsymbol{e}} \in M_G^{\mathsf{Ext}}, [a]_{\boldsymbol{e}} \in M_A \ \! \} \\ \cup \{ [((b, \mathscr{W}), \mathscr{E})]_{\boldsymbol{f}} \mid [((b, \mathscr{E}), \mathscr{W})]_{\boldsymbol{f}} \in M_G^{\mathsf{Ext}}, [b]_{\boldsymbol{f}} \in M_B \ \! \} \\ \cup \{ [((c, \mathscr{E}), \mathscr{E})]_{\boldsymbol{g}} \mid [(c, \mathscr{E})]_{\boldsymbol{g}} \in M_G^{\mathsf{Ext}}, [c]_{\boldsymbol{g}} \in M_C \ \! \} \\ \cup \{ [(m, \mathscr{N})]_{\boldsymbol{h}} \mid [m]_{\boldsymbol{h}} \in M_G^{\mathsf{Int}} \ \! \}$;

\item $M_{\Lambda(G)}^{\mathsf{Init}} \stackrel{\mathrm{df. }}{=} \{ [((c, \mathscr{E}), \mathscr{E})]_{\boldsymbol{g}} \mid [(c, \mathscr{E})]_{\boldsymbol{g}} \in M_G^{\mathsf{Ext}}, [c]_{\boldsymbol{g}} \in M_C^{\mathsf{Init}} \ \! \}$;

\item $\lambda_{\Lambda(G)} \stackrel{\mathrm{df. }}{=} \lambda_G \circ \mathit{peel}_{\Lambda(G)}$, where the map $\mathit{peel}_{\Lambda(G)} : M_{\Lambda(G)} \to M_G$ is given by $[(a, \mathscr{W})]_{\boldsymbol{e}} \mapsto [((a, \mathscr{W}), \mathscr{W})]_{\boldsymbol{e}}$, $[((b, \mathscr{W}), \mathscr{E})]_{\boldsymbol{f}} \mapsto [((b, \mathscr{E}), \mathscr{W})]_{\boldsymbol{f}}$, $[((c, \mathscr{E}), \mathscr{E})]_{\boldsymbol{g}} \mapsto [(c, \mathscr{E})]_{\boldsymbol{g}}$, $[(m, \mathscr{N})]_{\boldsymbol{h}} \mapsto [m]_{\boldsymbol{h}}$;

\item $\Delta_{\Lambda(G)} : [(m, \mathscr{N})]_{\boldsymbol{e}} \mapsto [(m', \mathscr{N})]_{\boldsymbol{e}}$, where $\Delta_G : [m]_{\boldsymbol{e}} \mapsto [m']_{\boldsymbol{e}}$;

\item $P_{\Lambda(G)} \stackrel{\mathrm{df. }}{=} \{ \boldsymbol{s} \in \mathscr{L}_{\Lambda(G)} \mid \mathit{peel}_{\Lambda(G)}^\ast(\boldsymbol{s}) \in P_G \ \! \}$, where the structure of pointers in $\mathit{peel}_{\Lambda(G)}^\ast(\boldsymbol{s})$ is the same as the one in $\boldsymbol{s}$.

\end{itemize}

\section{Constructions on strategies}
\label{AppendixConstructionsOnStrategies}
Next, we present the formal definitions of standard constructions on strategies given in the previous work \cite{yamada2019game}.

\subsection{Concatenation and composition of strategies \cite{yamada2019game}}
\label{DefConcatenationAndCompositionOfStrategies}
Let $\sigma : J$ and $\tau : K$ such that $\mathcal{H}^\omega(J) \trianglelefteqslant A \multimap B$ and $\mathcal{H}^\omega(K) \trianglelefteqslant B \multimap C$ for some normalized games $A$, $B$ and $C$. 
The \emph{\bfseries concatenation} $\sigma \ddagger \tau : J \ddagger K$ of $\sigma$ and $\tau$ is given by: 
\begin{equation*}
\sigma \ddagger \tau \stackrel{\mathrm{df. }}{=} \{ \boldsymbol{s} \in \mathscr{L}_{J \ddagger K} \mid \boldsymbol{s} \upharpoonright J \in \sigma, \boldsymbol{s} \upharpoonright K \in \tau, \boldsymbol{s} \upharpoonright B^{[0]}, B^{[1]} \in \mathit{pr}_B \ \! \}
\end{equation*}
and their \emph{\bfseries composition} $\sigma ; \tau : \mathcal{H}^\omega(J \ddagger K)$ by $\sigma ; \tau \stackrel{\mathrm{df. }}{=} \mathcal{H}^\omega (\sigma \ddagger \tau)$ (see Theorem~\ref{ThmHidingTheorem}).

\subsection{Generalized pairing of strategies \cite{yamada2019game}}
\label{DefGeneralizedPairingOfStrategies}
Given strategies $\sigma : L$ and $\tau : R$ such that $\mathcal{H}^\omega(L) \trianglelefteqslant C \multimap A$, $\mathcal{H}^\omega(R) \trianglelefteqslant C \multimap B$ for some normalized games $A$, $B$ and $C$, the \emph{\bfseries (generalized) pairing} $\langle \sigma, \tau \rangle : \langle L, R \rangle$ of $\sigma$ and $\tau$ is defined by:
\begin{align*}
\langle \sigma, \tau \rangle \stackrel{\mathrm{df. }}{=} \{ \boldsymbol{s} \in \mathscr{L}_{\langle L, R \rangle} \mid (\boldsymbol{s} \upharpoonright L \in \sigma \wedge \boldsymbol{s} \upharpoonright B = \boldsymbol{\epsilon}) \vee (\boldsymbol{s} \upharpoonright R \in \tau \wedge \boldsymbol{s} \upharpoonright A = \boldsymbol{\epsilon}) \ \! \}.
\end{align*}


\subsection{Generalized promotion of strategies \cite{yamada2019game}}
\label{DefGeneralizedPromotionOfStrategies}
Given a strategy $\phi : G$ such that $\mathcal{H}^\omega(G) \trianglelefteqslant A \Rightarrow B$ for some normalized games $A$ and $B$, the \emph{\bfseries (generalized) promotion} $\phi^{\dagger} : G^\dagger$ of $\phi$ is defined by: 
\begin{equation*}
\phi^{\dagger} \stackrel{\mathrm{df. }}{=} \{ \boldsymbol{s} \in \mathscr{L}_{G^\dagger} \mid \forall e \in \mathcal{T} . \ \! \boldsymbol{s} \upharpoonright \boldsymbol{e} \in \phi \ \! \}.
\end{equation*}

\subsection{Currying of strategies \cite{yamada2019game}}
\label{DefCurryingOnStrategies}
If $\phi :  G$ and $\mathcal{H}^\omega(G) \trianglelefteqslant A \otimes B \multimap C$, where $A$, $B$ and $C$ are normalized, then the \emph{\bfseries currying} $\Lambda(\phi) : \Lambda(G)$ of $\phi$ is given by:
\begin{align*}
\Lambda(\phi) &\stackrel{\mathrm{df. }}{=} \{ \boldsymbol{s} \in \mathscr{L}_{\Lambda(G)} \mid \mathit{peel}_{\Lambda(G)}^\ast(\boldsymbol{s}) \in \phi \ \! \}.
\end{align*}

\section{Views \cite{hyland2000full,abramsky1999game}}
\label{DefViews}
Finally, we give the formal definition of \emph{views}.
Given a legal position $\boldsymbol{s} \in \mathscr{L}_G$ of an arena $G$, the \emph{\bfseries Player (P-) view} $\lceil \boldsymbol{s} \rceil_G$ and the \emph{\bfseries Opponent (O-) view} $\lfloor \boldsymbol{s} \rfloor_G$ (we often omit the subscript $G$) are defined by induction on $| \boldsymbol{s} |$ as follows: 
\begin{itemize}

\item$\lceil \boldsymbol{\epsilon} \rceil_G \stackrel{\mathrm{df. }}{=} \boldsymbol{\epsilon}$;

\item $\lceil \boldsymbol{s} m \rceil_G \stackrel{\mathrm{df. }}{=} \lceil \boldsymbol{s} \rceil_G . m$ if $m$ is a P-move;

\item $\lceil \boldsymbol{s} m \rceil_G \stackrel{\mathrm{df. }}{=} m$ if $m$ is initial;

\item $\lceil \boldsymbol{s} m \boldsymbol{t} n \rceil_G \stackrel{\mathrm{df. }}{=} \lceil \boldsymbol{s} \rceil_G . m n$ if $n$ is an O-move with $\mathcal{J}_{\boldsymbol{s} m \boldsymbol{t} n}(n) = m$;

\item $\lfloor \boldsymbol{\epsilon} \rfloor_G \stackrel{\mathrm{df. }}{=} \boldsymbol{\epsilon}$;

\item $\lfloor \boldsymbol{s} m \rfloor_G \stackrel{\mathrm{df. }}{=} \lfloor \boldsymbol{s} \rfloor_G . m$ if $m$ is an O-move;

\item $\lfloor \boldsymbol{s} m \boldsymbol{t} n \rfloor_G \stackrel{\mathrm{df. }}{=} \lfloor \boldsymbol{s} \rfloor_G . m n$ if $n$ is a P-move with $\mathcal{J}_{\boldsymbol{s} m \boldsymbol{t} n}(n) = m$

\end{itemize}
where the justifiers of the remaining occurrences in $\lceil \boldsymbol{s} \rceil_G$ (resp. $\lfloor \boldsymbol{s} \rfloor_G$) are unchanged if they occur in $\lceil \boldsymbol{s} \rceil_G$ (resp. $\lfloor \boldsymbol{s} \rfloor_G$) and undefined otherwise. 

\end{document}